\documentclass{article}

\usepackage{amsmath,amsthm, thmtools ,amsfonts,amssymb,times,bbm,graphicx,subcaption, hyperref, multirow}

\graphicspath{ {figures/} }

\usepackage{float} 
\usepackage{adjustbox}

\usepackage{tikz}
\usetikzlibrary{quantikz}

\usepackage{algorithm}
\usepackage{algpseudocode}

\usepackage[font=small]{caption}

\newtheorem{theorem}{Theorem}

\newtheorem{lemma}[theorem]{Lemma}

\newtheorem{definition}[theorem]{Definition}

\newtheorem{estimate}[theorem]{Estimate}

\def\RR{\mathbbm{R}}

\def\NN{\mathbbm{N}}
\def\EE{\mathbbm{E}}

\def\Id{\mathbbm{1}}

\def\O{\mathcal{O}}

\def\N{\mathcal{N}}

\newcommand{\norm}[1]{\left\lVert#1\right\rVert}
\newcommand{\trnorm}[1]{\norm{#1}_\mathrm{tr}}

\DeclareMathOperator{\tr}{tr}
\DeclareMathOperator{\Tr}{Tr}

\DeclareMathOperator{\sgn}{sign}
\DeclareMathOperator{\diag}{diag}

\DeclareMathOperator{\Pois}{Pois}

\usepackage{newfloat}
\DeclareFloatingEnvironment[
    fileext=loa,
    name=Algorithm, % This ensures the label is 'Algorithm'
    placement=H,
]{algorithmCap} % wrapper environment for algorithm env which allows for a caption below an algorithm

\usepackage[style=numeric, maxnames=3, minnames=1, url=false, isbn=false]{biblatex}
\AtEveryBibitem{\clearfield{pages}} 
\AtEveryBibitem{\clearfield{month}}
\AtEveryBibitem{\clearfield{editor}} 
\AtEveryBibitem{\clearfield{volume}} 
\AtEveryBibitem{\clearfield{number}} 
\bibliography{manuscript}

\usepackage{authblk} 
\usepackage{tabularx}

\begin{document}

\title{Solving quadratic binary optimization problems using quantum SDP methods: \\
Non-asymptotic running time analysis}

\author[1]{Fabian~Henze\footnote{fhenze2@thp.uni-koeln.de}}
\author[2]{Viet~Tran}
\author[3]{Birte~Ostermann}
\author[2]{Richard~Kueng}
\author[3]{Timo~de~Wolff}
\author[1]{David~Gross}

\renewcommand\Authands{ and }

\affil[1]{\small Institute for Theoretical Physics, University of Cologne, Germany}
\affil[2]{\small Institute for Integrated Circuits and Quantum Computing, JKU Linz, Austria}
\affil[3]{\small Institute for Analysis and Algebra, TU Braunschweig, Germany}

\date{}

\maketitle

\begin{abstract} 
Quantum computers can solve semidefinite programs (SDPs) using resources that scale better than state-of-the-art classical methods as a function of the problem dimension. 
At the same time, the known quantum algorithms scale very unfavorably in the precision,
which makes it non-trivial to find applications for which the quantum methods are well-suited.
Arguably, precision is less crucial for SDP relaxations of combinatorial optimization problems (such as the Goemans-Williamson algorithm), 
because these include a final rounding step that maps SDP solutions to binary variables.
With this in mind, 
Brandão, França, and Kueng 
have proposed to use quantum SDP solvers in order to achieve an \emph{end-to-end} speed-up for obtaining approximate solutions to combinatorial optimization problems.
They did indeed succeed in identifying an algorithm that realizes a 
polynomial quantum advantage in terms of its asymptotic running time.
However, asymptotic results say little about the problem sizes for which advantages manifest.
Here, we present an analysis of the non-asymptotic resource requirements of this algorithm. 
The work consists of two parts. 
First, we optimize the original algorithm with a particular emphasis on performance for realistic problem instances.
In particular, we formulate a version with adaptive step-sizes, an improved detection criterion for infeasible instances, and a more efficient rounding procedure. 
In a second step, we benchmark both the classical and the quantum version of the algorithm. 
The benchmarks did not identify a regime where even the optimized quantum algorithm would beat standard classical approaches for input sizes that can be realistically solved at all.
In the absence of further significant improvements, these algorithms therefore fall into a category sometimes called \emph{galactic}: 
Unbeaten in their asymptotic scaling behavior, but not practical for realistic problems.
\end{abstract}
\newpage
\tableofcontents
\newpage

\section{Introduction}

\subsection{Non-asymptotic analysis of algorithms}

Quantum algorithms are typically compared to classical approaches based on their asymptotic behavior.
Indeed, the lack of large-scale, fault-tolerant quantum hardware makes a direct comparison of practical performance difficult for the time being.

The danger of this approach lies in the fact that an asymptotic assessment hides constant factors, which can be substantial. 
Even if an advantageous scaling behavior has been rigorously established, such an analysis does not directly give information about the minimal size of instances for which a practical speed-up actually manifests.
One therefore encounters the risk of designing what is sometimes referred to as \emph{galactic algorithms}:
Computational methods that do outperform all others, but only at instance sizes that are so enormous that there is no hope that they will ever be practically relevant.
It is therefore a timely and important task to find non-asymptotic estimates for the resource use of quantum algorithms \cite{Abbas_2024, ammann2023realisticruntimeanalysisquantum, Dalzell_2023, dalzell2023quantumalgorithmssurveyapplications}.

In doing so, one faces the problem that algorithms published in the academic literature are often not optimized for practical performance. 
Up until recently, such optimizations were not seen as a priority:
The description of \emph{any} novel quantum algorithm beyond the well-known classes discovered by Shor~\cite{Shor94}, Grover~\cite{Grover96}, as well as Hassidim Harrow and Lloyd~\cite{HHL09}, 
displaying some form of quantum advantage,
is considered a significant achievement.
This, coupled with the fact that scaled-up quantum hardware remains out of reach, means that there has been little incentive for researchers to optimize implementations.

In particular if one expects to find negative results about the practical usability of published quantum algorithms, 
it is therefore incumbent upon those performing benchmarks, to first expend reasonable efforts to find an optimized implementation.
For this reason, a large part of the present work is spent on improving existing algorithms (Secs.~\ref{sec:class_imp} -- \ref{sec:quant_imp}).

\subsection{Combinatorial optimization}

We are concerned here with quadratic unconstrained binary optimization problems, often abbreviated as \emph{QUBOs}:
\begin{align}\label{eqn:quadratic_opt}
	\underset{x \in \mathbb{R}^n}{\text{maximize}} & \quad x^T C x = \sum_{i,j=1}^n C_{ij} x_i x_j \\
	\text{subject to}& \quad x_i = \pm 1 \quad \text{for $i=1,\ldots,n$}. \nonumber
\end{align}
for a symmetric cost matrix $C\in\RR^{n\times n}$.
Finding the optimizer is NP-hard, as can be seen e.g.\ by reduction from the problem of finding a maximum cut in a graph (\textsc{MaxCut}) -- one of Karp's 21 NP-complete problems~\cite{Karp1972}.

The seminal work by Goemans-Williamson \cite{goemansWilliamson} provides a randomized rounding procedure with rigorous approximation guarantees for \textsc{MaxCut}, using the semidefinite programming (SDP) relaxation of the original problem.
This SDP relaxation of Problem \eqref{eqn:quadratic_opt} is given by
\begin{align}
	\begin{split}\label{eqn:relaxed}
		\underset{X \in \mathbb{R}^{n \times n}}{\text{maximize}} & \quad \tr \left( C^T X \right) = \sum_{i,j=1}^n C_{ij} X_{ij} \\
		\text{subject to} & \quad \mathrm{diag}(X) = 1, X \succeq 0, 
	\end{split}
\end{align}
where the final constraint demands that $X$ is symmetric ($X^T=X$) and positive semidefinite, or psd for short. 
The quantitative relation between the original problem (\ref{eqn:quadratic_opt}) and the relaxed version (\ref{eqn:relaxed}) depends on properties of the coefficient matrix $C$ 
(e.g.\ whether is it positive in the SDP- or element-wise sense; see, e.g.\ Refs.~\cite{RoundingGrothendieck, friedland2020symmetricgrothendieckinequality, Briet_2014}).  
In this paper, we restrict attention to matrices of the form
\begin{align} \label{eq:C_block}
    C =
        \begin{pmatrix}
            0 & B \\ B^T & 0
        \end{pmatrix}, 
\end{align}
where $B\in \RR^{n/2\times n/2}$. 
In order to fix a consistent normalization, we will throughout assume that the coefficient matrix is normalized in operator norm (or spectral norm, i.e.\ the largest singular value), $\|C\|=1$.
For such instances, the value obtained from applying a \emph{randomized rounding} procedure to the solution of (\ref{eqn:relaxed}), that gives a feasible solution for the original problem, is at most a factor of 
$(\frac{4}{\pi}-1)\simeq.273$ 
smaller than the optimal value of the original problem~\cite[Sec.~4.1]{RoundingGrothendieck}. 
Remarkably, the relation does not depend on the dimension $n$.

The relaxed problem~\eqref{eqn:relaxed} can be represented as a semidefinite program (SDP).
As such, it can be tackled, e.g.\ via \textit{interior point methods} or the \textit{ellipsoid method}. 
These methods, in particular, interior point methods, are effectively solvable, and ellipsoid methods are, moreover, solvable in polynomial time if proper starting criteria are met. 
While it the polynomial time solvability of SDPs is not fully settled in general yet due to an observation by O'Donnell \cite{ODonnell17}, the polynomial time solvability for SDPs relaxed from QUBOs via Sums of Squares follows from Raghavendra and Weitz \cite{RaghavendraW17}. 
We refer to standard textbooks, such as~\cite{Bar02,BD04}, for further details. 
In practice, though, SDP solvers perform poorly in the problem size, both in terms of running time and memory requirements, thus even a polynomial advantage might therefore have significant practical impact.

To simplify the comparison between various methods, we state their performance for the particular case where
the matrix $B$ is chosen to be a matrix with column sparsity $s$, where the non-zero entries are i.i.d.\ Gaussian random variables.
We also assume that the exponent of matrix multiplication equals 3,
reflecting the scaling in practical applications.
Let $\mu$ be the desired precision of the SDP solution.
Then, the two prevalent approaches to solve the SDP in Eq.~\eqref{eqn:relaxed} have the following asymptotic running times:
\begin{enumerate}
\item[(i)] \emph{Interior point methods} require $\O(n^{4}\log(\mu^{-1}))$ floating-point operations; see, e.g.~\cite{InteriorPoint}.
\item[(ii)] \emph{Matrix multiplicative weight methods} require $\tilde\O(\min\{n^{2.5}s\mu^{-2.5}, n^{2.5}s^{0.5}\mu^{-3.5}\})$
floating-point operations; see, e.g.~\cite{MMW}.
\end{enumerate}

We note that $\mu$ refers to the approximation of the optimal value of the relaxed problem (\ref{eqn:relaxed}), not the original one (\ref{eqn:quadratic_opt}).
The latter, being NP-hard to approximate by the PCP theorem,
does of course not admit any polynomial-time approximation under standard complexity-theoretic assumptions.

With this goal in mind, Ref.~\cite{Brandao2022fasterquantum} developed 
a classical and a quantum 
\emph{Hamiltonian Update algorithm} for solving the SDP relaxation (\ref{eqn:relaxed}) using resources that scale better than off-the-shelf solvers in the problem size.
On the flip side, their algorithm converges very slowly. 
That is, the proven running time bounds are extremely unfavorable in the precision $\mu$.
More precisely, their main theorem is:

\begin{theorem}[{\cite[Thm.~1]{Brandao2022fasterquantum}}]
    Let $C$ be a (real-valued) symmetric $n \times n$ matrix with column sparsity $s$. 
		Then the 
		problem \eqref{eqn:relaxed} can be solved up to additive accuracy $n\norm{C}\mu$ in running time
    \begin{align*}
        \tilde\O(n^{1.5} (\sqrt{s})^{1+o(1)} \mu^{-28+o(1)}\exp(1.6\sqrt{12\log(\mu^{-1})}))
    \end{align*} 
    on a quantum computer and in running time
    \begin{align*}
        \tilde \O (\min\{n^2 s \mu^{-12}, n^3 \mu^{-8} \}) 
    \end{align*}
    on a classical computer.
\end{theorem}

We are thus faced with the following situation:
The algorithm of Ref.~\cite{Brandao2022fasterquantum} shows that, for fixed precision $\mu$, quantum computers can, in principle, solve the SDP relaxation (\ref{eqn:relaxed}) 
with a running time that scales more favorably in the problem size than any known classical approach with rigorous performance guarantees.
At the same time, the proven dependency on $\mu$ suggests that the theoretical advantage will only manifest for enormously large problem instances, which cannot be practically executed at all.

The purpose of the present paper is to optimize the results of Ref.~\cite{Brandao2022fasterquantum}, derive sharper bounds, and finally benchmark the performance of the optimized version, in order to estimate its performance for realistic scenarios.

Another recent follow-up to Ref.~\cite{Brandao2022fasterquantum} is Ref.~\cite{augustino2023solvingsemidefiniterelaxationqubos}, which explores the use of 
\emph{iterative refinement} techniques to speed up convergence.
We have not included a quantitative comparison between these two approaches in this manuscript,
because no reference implementation is available,
and we are still in the process of clarifying some questions about algorithmic details with the authors \cite{brandonPrivate}.

We refer to Refs.~\cite{BrandaoInequality, brandao2019quantumsdpsolverslarge, vanapeldoorn_improvements, vanApeldoorn2020quantumsdpsolvers} for more details on quantum algorithms for solving SDPs. 

\subsection{Organization of the paper}
The structure of the paper is as follows:

In Sec.~\ref{sec:HU}, we provide an overview of the Hamilton Updates (HU) algorithm introduced in Ref.~\cite{Brandao2022fasterquantum}. 
Sec.~\ref{sec:class_imp} presents several non-asymptotic improvements to HU that enhance both its classical and quantum performance. 
In particular:
Secs.~\ref{sec:adaptive_step_size}-\ref{sec:entropy} detail methods for reducing the number of iterations needed in practice to find a feasible solution or to certify infeasibility. 
In Sec.~\ref{sec:numerics_improvements}, we numerically compare the performance of the improved algorithm to the original one.
Sec.~\ref{sec:rounding} provides an asymptotically tighter bound on the precision of the objective value after rounding, improving upon the results in Ref.~\cite{Brandao2022fasterquantum}. Complementing this, Sec.~\ref{sec:eps_scaling} numerically examines the behavior of the precision after rounding.
Sec.~\ref{sec:quant_imp} presents an improved quantum subroutine that uses algorithms with better dependence on the precision for Gibbs state preparation and Hamiltonian simulation, offering improvements over Ref.~\cite{Brandao2022fasterquantum}. 
Sec.~\ref{sec:asymptotic_improvements} provides an overview of the analytical and numerical results from Secs.~\ref{sec:rounding_section} and \ref{sec:quant_imp}. 
Finally, in Sec.~\ref{sec:benchmarking}, we benchmark the improved algorithm by running it classically and estimating the minimum number of gates required for a quantum implementation based on the constructions in Refs.~\cite{Low2019Qubitization, vanApeldoorn2020quantumsdpsolvers}.

\subsection{Notation}

The following table summarizes the notational conventions we use throughout the manuscript:

\begin{table}[H]
    \centering
    \small
    \begin{tabular}{lll}
        \hline
        \textbf{Symbol} & \textbf{Description} & \textbf{Definition }\\
        \hline
				$A \preceq B$ & positive semidefinite (psd) order & $A\preceq B \Leftrightarrow$ 
				$A, B$ symmetric, and  
				\\
									&&
				$x^T A x \leq x^T B x$
				$\; \forall x\in \RR^n$ \\
        $\norm{x}_{\ell_1}$ &  $\ell_1$ norm & $\norm{x}_{\ell_1}=\sum_i|x_i|$ \\
        $\norm{A}$ & operator 
				(or Schatten-$\infty$) 
				norm & $\norm{A} = \sup_{x\in\RR^n: \norm{x}=1}\{\norm{Ax}\}$ \\
        $\trnorm{A}$ & trace (or Schatten-1) norm &  $\trnorm{A}= \tr(|A|),\quad |A|=\sqrt{A^TA}$ \\
        $\norm{A}_\mathrm{max}$ & max (or $\ell_\infty$) norm & $\norm{A}_\mathrm{max} = \max_{ij} |A_{ij}|$ \\
        $\norm{A}_{\ell_1\rightarrow\ell_2}$ & $\ell_1\rightarrow \ell_2$ norm & $\norm{A}_{\ell_1\rightarrow\ell_2} = \sup_{x\in\RR^n: \norm{x}_{\ell_1}=1}\{\norm{Ax}\}$ \\
        $\norm{A}_{\tr\rightarrow\tr}$ & trace-to-trace norm & 
        $\norm{A}_{\tr\rightarrow\tr} = \sup_{B\in\RR^{n\times n}: \trnorm{B}=1}\{\trnorm{A(X)}\}$ \\
        $\norm{A}_{\infty\rightarrow\infty}$ & infinity-to-infinity norm & 
        $\norm{A}_{\infty\rightarrow\infty} = \sup_{B\in\RR^{n\times n}: \norm{B}=1}\{\norm{A(X)}\}$ \\
        $[n]$ & index set & $[n]=\{i\in\NN | 1\leq i\leq n\}$ \\
        $R(\rho \| \sigma )$ & quantum relative entropy & $S(\rho \| \sigma) = \tr \left( \rho (\ln\rho - \ln\sigma) \right)$ \\
        & & with $\rho,\sigma\succeq 0, \tr(\rho)=\tr(\sigma)=1$ \\
        $C$ & normalized cost matrix & $C\in\RR^{n\times n}, C^T=C, \|C\|=1$ \\
        $x$ & solution for QUBO \eqref{eqn:quadratic_opt} & $x\in \{-1,1\}^n$ \\
        $X$ & solution for SDP \eqref{eqn:relaxed} & $X\in \RR^{n\times n}$, $X$ is psd, $\diag(X)=\Id$\\
        $H$ & Hamiltonian and (potential) solution for \eqref{eq:approx_feasibility_constraint} & $H\in\RR^{n\times n}, H^T=H$ \\
        $\rho$ or $\rho_H$ & Gibbs state and (potential) solution for \eqref{eqn:epsilon_exact_program}
        & $\rho_H=\exp(-H)/\tr(\exp(-H))$ \\
        $\gamma$ & threshold for objective value & Defined in \eqref{eq:approx_feasibility_constraint}. \\
        $n$ & dimension of cost matrix & \\
        $s$ & (column) sparsity & maximum number of non-zero entries \\
             && per column in a matrix \\
        $\mu$ & SDP precision & Let $X^\star$ be an optimal solution to \eqref{eqn:relaxed}.  \\
        && A solution $X$ to \eqref{eqn:relaxed} has precision $\mu$ if \\
        && $\tr(X^\star C)-\tr(XC)\leq n\mu$. \\
        $\epsilon$ & precision of HU constraints & Defined in \eqref{eq:approx_feasibility_constraint}. \\
        $\nu$ & precision after randomized rounding & Defined in \eqref{eq:nu}. \\
        $\Delta H$ & matrix for HU (cost/diagonal) update & $H\mapsto H+\lambda \Delta H$, defined in \eqref{eq:Delta_H} \\  
        $\lambda$ or $\lambda_c,\lambda_d$ & step size/length for HU (cost/diagonal) update & \\
        $P_c$ & matrix for cost update & $P_c=\gamma\Id-C$ \\
        $\begin{aligned}P_d^{\ell_1}\\ \  \end{aligned}$ 
        & $\begin{aligned}
            \text{matrix for diagonal update with } \ell_1 \text{ norm} \\ \ 
        \end{aligned}$
        & $\begin{aligned}
            P_d^{\ell_1}&=\sgn(\diag(\rho)-\Id/n)  \\ &- \tr(\sgn((\diag\rho)-\Id/n))/n \Id
        \end{aligned}$ \\
        $P_d^{\ell_2}$ & matrix for diagonal update with $\ell_2$ norm 
        & $P_d^{\ell_2}=(\diag(\rho)-\Id/n) / \max_i|\rho_{ii} - 1/n|$ \\
        $M^{(k)}$ & momentum in $k^\mathrm{th}$ iteration & $M^{(k)}=\lambda_{c/d}(\Delta H)^{(k)}$ \\
        $\beta$ & momentum hyperparameter & $0\leq\beta<1$ \\
        $F$ & free energy & $F(H)=-\ln( \tr( \exp(-H)))$ \\
        $b$ & number of qubits used to store $H$ & \\
        \hline 
    \end{tabular}
    \caption{Summary of notation used in this paper.}
    \label{tab:notation}
\end{table}

\section{The Hamiltonian Update algorithm} \label{sec:HU}

Now, we give a summary of the Hamiltonian Updates algorithm for solving the
SDP in Eq.~\eqref{eqn:relaxed}. 
Compared to the original presentation of Ref.~\cite{Brandao2022fasterquantum}, we slightly modify the notation involving the update step, in a way that facilitates stating the improvements in Sec.~\ref{sec:class_imp}. We give the high-level algorithm in Alg.~\ref{alg:HU_simple} and an illustration in Fig.~\ref{fig:HU}.

\subsection{Reduction to feasibility problems} \label{sec:refProblem}

We express the HU algorithm as an optimization over the renormalized psd matrix
\begin{align}
	\rho = \frac1n X.
\end{align}

The two representations are obviously equivalent -- but the convention adopted here will later allow us to formulate a quantum version, where $\rho$ will be a \emph{density matrix} (i.e.\ $\rho \succeq 0$ and $\tr(\rho)=1$) describing the state of a physical quantum system.
In particular, the constraint on the diagonal now reads
$\diag(\rho)=\frac{1}{n}\Id$.

Using the matrix Hölder inequality, we can then bound the objective function of the SDP relaxation: 
\begin{align*}
\left| \tr \left( C\rho \right) \right| \leq \| C \| \norm{\rho}_{\mathrm{tr}}= \|C\| \tr (\rho) =1,
\end{align*}
because the trace norm of a psd matrix is equal to its trace. This ensures that the optimal value of the re-scaled version of SDP~\eqref{eqn:relaxed} falls into the interval $[-1,1] $.

In a next step, we will 
transform the optimization problem into a series of feasibility problems.
Choose a precision parameter $\epsilon_b$ and perform a binary search over the interval $[-1,1]$, to find a value $\gamma^\star$ such that the program 
\begin{align}
	\begin{split}\label{eqn:exact_program}
    \underset{\rho \in \mathbb{R}^{n \times n}}{\text{find}} \quad 
    &\rho \succeq 0,\quad \tr(\rho)=1 \\
    \text{subject to} \quad 
    &\gamma - \tr(C \rho)  \leq 0 \\ 
    \text{and} \quad
    &\sum_i \Big|\rho_{ii} - \frac 1 n\Big| = 0 
	\end{split}
\end{align}
is feasible for $\gamma=\gamma^\star$, but not for $\gamma=\gamma^\star+\epsilon_b$.
A binary search finds such a value in $\O(\log_2(\epsilon^{-1}_b))$ iterations.
Therefore, the feasibility problems (\ref{eqn:exact_program}) can find the optimal value of (\ref{eqn:relaxed}) 
with an overhead that is logarithmic in the desired precision.

However,
even the feasibility problem cannot be decided by a practical algorithm for exact constraints.
Let $\epsilon_f$ be another precision parameter.
We say that $\rho$ is \emph{$\epsilon_f$-feasible} for the program (\ref{eqn:exact_program}) if it satisfies
\begin{align}
	\begin{split}\label{eqn:epsilon_exact_program}
    \underset{\rho \in \mathbb{R}^{n \times n}}{\text{find}} \quad 
    &\rho \succeq 0,\quad \tr(\rho)=1 \\
    \text{subject to} \quad 
    &\gamma - \tr(C \rho)   < \epsilon_f \\
    \text{and} \quad
    &\sum_i \Big|\rho_{ii} - \frac 1 n\Big|  < \epsilon_f.
	\end{split}
\end{align}
The program (\ref{eqn:exact_program}) is \emph{$\epsilon_f$-feasible} if an $\epsilon_f$-feasible $\rho$ exists.

In the section below, we will introduce the \emph{Hamiltonian Update} (HU) routine for solving Eq.~(\ref{eqn:epsilon_exact_program}).
It satisfies the following conditions:
\begin{enumerate}
    \item If (\ref{eqn:exact_program}) is feasible, the HU routine outputs an $\epsilon_f$-feasible solution $\rho$.
    \item If (\ref{eqn:exact_program}) is not $\epsilon_f$-feasible, the HU routine outputs no solution.
    \item If (\ref{eqn:exact_program}) is $\epsilon_f$-feasible but not feasible, the HU routine outputs either an $\epsilon_f$-feasible solution or no solution.
\end{enumerate}
Thus, conversely,
\begin{enumerate}
	\item
		If the HU algorithm succeeds, it will output an $\epsilon_f$-feasible $\rho^\star$.
		In Sec.~\ref{sec:rounding_section}, we will describe \emph{rounding algorithms}, which construct exact solutions given $\rho^\star$.
	\item
		If the HU algorithm fails, we know that (\ref{eqn:exact_program}) is not strictly feasible.
\end{enumerate}
Therefore, a binary search using the HU algorithm will output an $\epsilon_f$-feasible solution $\rho^\star$ with objective value $\gamma^\star$, such that (\ref{eqn:exact_program}) is not strictly feasible for $\gamma=\gamma^\star+\epsilon_b$.

For simplicity, in what follows, we will restrict to the case where $\epsilon_f = \epsilon_b$, and denote this common precision parameter by $\epsilon$.

\subsection{The Hamiltonian Update step}
\label{sec:the_hu_step}

The idea behind Hamiltonian Updates is to express $\rho$ as a \emph{Gibbs states}, i.e.\ one writes
\begin{align}
	\rho_H = \frac{\exp(-H)}{\tr(\exp(-H))}
\end{align}
for some symmetric matrix $H\in\RR^{n\times n}$. 
We refer to $H$ as the \emph{Hamiltonian}, in accordance with standard physics terminology.
This way, the constraints $\rho \succeq 0$ and  $\tr(\rho)=1$ are automatically satisfied. 
Thus, the problem reduces to finding an $H$ such that the constraints on $\rho_H$ are fulfilled:
\begin{align}
	\begin{split} \label{eq:approx_feasibility_constraint}
    \underset{\substack{H \in \mathbb{R}^{n \times n} \\ H=H^T}}{\text{find}} &\quad 
    H, \\
    \text{subject to} &\quad 
    \gamma - \tr(C \rho_H)  < \epsilon  \\
    \text{and} &\quad
    \sum_i \Big|(\rho_H)_{ii} - \frac 1 n\Big| < \epsilon 
	\end{split}
\end{align}
To find a suitable Hamiltonian, we start with the $n \times n$ zero matrix $H=0$, and refine it in a series of \emph{update steps}, described next.

\begin{figure}[H]
    \centering
    \begin{subfigure}[b]{\textwidth}
        \centering
        \includegraphics[width=0.75\textwidth]{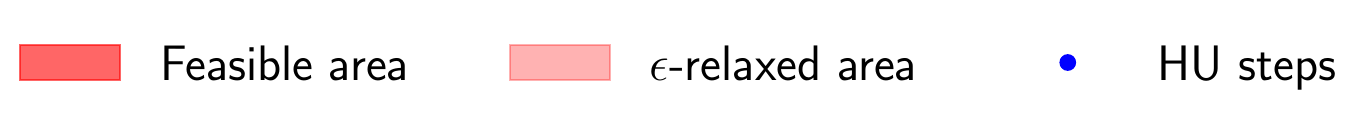}
    \end{subfigure}\\
    \begin{subfigure}[b]{0.32\textwidth}
        \includegraphics[width=\textwidth]{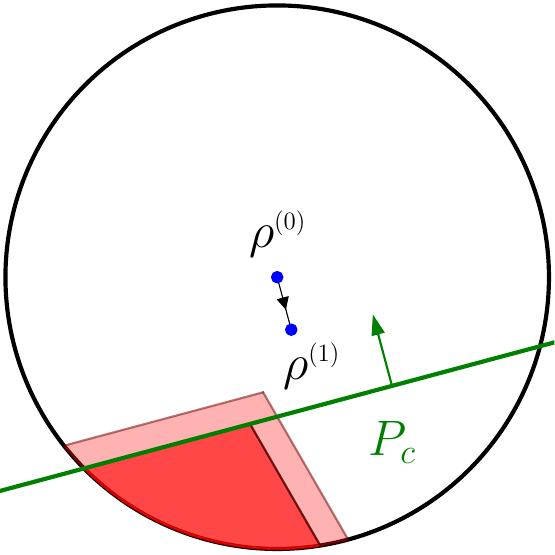}
        \caption{}
    \end{subfigure}
    \begin{subfigure}[b]{0.32\textwidth}
        \includegraphics[width=\textwidth]{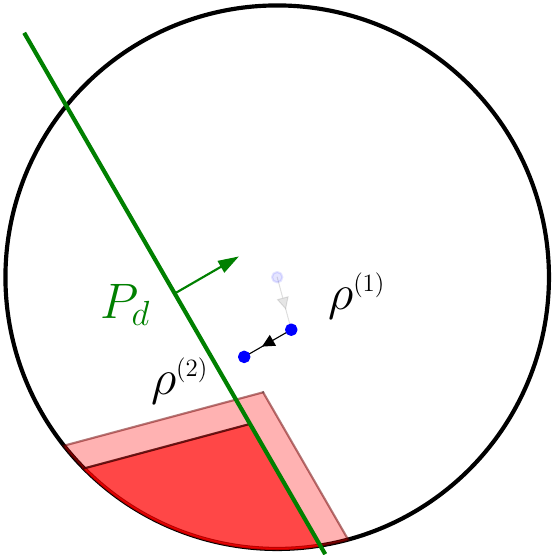}
        \caption{}
    \end{subfigure}
    \begin{subfigure}[b]{0.32\textwidth}
        \includegraphics[width=\textwidth]{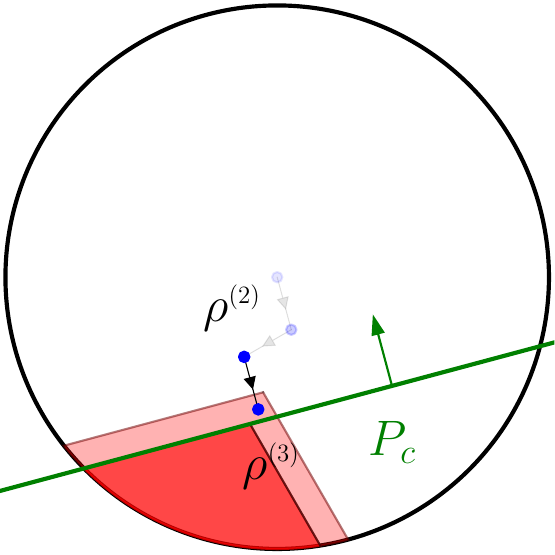}
        \caption{}
    \end{subfigure}
    \hfill
    \caption{
        \textit{Illustration of the Hamiltonian Updates algorithm.}
        The circle depicts the space of trace-one psd matrices, with $\rho_0$ lying in the center. 
				The feasible region is shown in dark red, the $\epsilon$-feasible region is marked light red.
        Each graphic shows a single iteration of HU: At the start of each update, a matrix $P$ is calculated that defines a hyperplane (shown in green) separating the current $\rho$ and the feasible region. By updating $\rho_H\rightarrow\rho_{H+\lambda P}$ (i.e.\ penalizing infeasible directions), $\rho$ moves towards the hyperplane. For simplicity this is depicted by a straight line. This is generally not the case, as $\rho$ depends non-linearly on $H$. The procedure ends when $\rho$ enters the $\epsilon$-feasible region (i.e.\ all constraints are fulfilled up to precision $\epsilon$ as defined in \eqref{eq:approx_feasibility_constraint}).
    }
    \label{fig:HU}
\end{figure}

\begin{algorithmCap}
\begin{algorithm}[H]
    \caption{Simplified Hamiltonian Updates}\label{alg:HU_simple}
    \begin{algorithmic}[1]
        \Require 
            Cost matrix $C$, threshold objective value $\gamma$, precision parameter $\epsilon$ 
            \Statex
        \noindent
        \Ensure \quad \begin{tabular}[t]{l|l}  
              \textbf{Condition} & \textbf{Output} \\
              \hline
              (\ref{eqn:exact_program}) is feasible & an $\epsilon$-feasible $\rho$ \\
              (\ref{eqn:exact_program}) is not $\epsilon$-feasible & false \\
              else & undefined (an $\epsilon$-feasible $\rho$ or false)
        \end{tabular}
		\Statex
                
        \Function{hamiltonian\_updates}{$C, \gamma, \epsilon$}
        \State $H \gets 0_{n\times n}, \rho\gets \Id/n$ and $F\gets -\ln(n)$
        \While{$F\leq0$}
            \Comment{Main loop of HU}
            \If{$\rho$ is $\epsilon$-feasible}
                \State \Return $\rho$ 
            \Else
                \State \textbf{compute} $\Delta H$
                \Comment{Computed from the violations of $\rho$}
                \State $H \gets H + \lambda\Delta H$
                \Comment{Update Hamiltonian}
                \State $\rho \gets \exp(-H)/\tr(\exp(-H))$
                \Comment{Update Gibbs state}
                \State $F\gets -\ln(\tr(\exp(-H)))$
                \Comment{Update free energy}
            \EndIf
        \EndWhile
        \State \Return false
        \Comment{$F>0 \rightarrow$ No feasible solution exists}
    \EndFunction
     \end{algorithmic}
     \end{algorithm}
\end{algorithmCap}

To keep notation succinct, we 
usually make the dependence of $\rho_H$ on $H$ implicit, and write $\rho=\rho_H$ if the Hamiltonian is clear from context.
Let $\tilde{\rho}=\mathrm{diag}(\rho)$ be the diagonal matrix with $\tilde\rho_{ii}=\rho_{ii}$. 
We define the two symmetric matrices:
\begin{align}
    P_c &= \gamma\Id - C, \label{eq:P_c}\\ 
    P_d^{\ell_1} &=\sgn(\tilde\rho-\Id/n) - \tr(\sgn(\tilde\rho-\Id/n))/n \Id,  \label{eq:P_d_1}
\end{align}
where the sign function is applied element-wise to the matrices. 
(The origin of the superscript $\ell_1$ for $P_d$ will become clear in 
Sec.~\ref{sec:Pd}).

The matrices $P_c$ and $P_d^{\ell_1}$ allow us to reformulate the two feasibility constraints \eqref{eq:approx_feasibility_constraint} as 
\begin{align}
    \gamma - \tr(C \rho) &= \tr(P_c \rho) < \epsilon \label{eq:c-constraint} \\
    \sum_i |\rho_{ii} - 1/n| &= \tr(P_d^{\ell_1} \rho) < \epsilon. \label{eq:d-constraint}
\end{align}
While the correctness of \eqref{eq:c-constraint} is easy to see, \eqref{eq:d-constraint} needs a short calculation to confirm:
\begin{align}\begin{split}
    \sum_i |\rho_{ii} - 1/n| 
    &= \sum_i (\rho_{ii} - 1/n) \sgn(\rho_{ii} - 1/n) \\
    &= \sum_i \rho_{ii}\sgn(\rho_{ii} - 1/n) - \sum_i \sgn(\rho_{ii} - 1/n) / n \\
    &= \tr(\rho \sgn(\tilde\rho - \Id/n)) - \tr(\rho) \tr(\sgn(\tilde\rho-\Id/n) / n ) \\
    &= \tr(\rho \sgn(\tilde\rho - \Id/n)) - \tr(\rho \tr(\sgn(\tilde\rho-\Id/n)) / n ) \\
    &= \tr(\rho P_d^{\ell_1}),
\end{split}\end{align}
where we used $\tr(\rho) = 1$ in the third
step. 

The matrices $P_c$ and $P_d^{\ell_1}$ can be interpreted as normal vectors for hyperplanes in the vector space of symmetric $n \times n$ matrices which we also endow with the trace (or Frobenius) inner product $A, B \mapsto \mathrm{tr}(AB)$.
The hyperplanes separate the current iterate $\rho$ from the feasible region defined in the feasibility SDP (\ref{eq:approx_feasibility_constraint}).
The trace inner products $\tr(P_c\rho)$ and $\tr(P_d^{\ell_1}\rho)$ measure how strongly the constraints are violated. If both constraints are violated by less than $\epsilon$, the algorithm stops and outputs the solution $\rho$. 
Otherwise, in order to reduce the violations, HU applies 
\emph{cost updates}
and
\emph{diagonal updates},
defined, respectively, as
\begin{align}
    H \mapsto H + \lambda \Delta H, 
    \quad \text{where} \quad \Delta H = 
     \left\{
			 \begin{array}{ll}
				 P_c & \text{during cost updates,} \\
				 P_d^{\ell_1} & \text{during diagonal updates.} 
			 \end{array}
		 \right. 
   \label{eqn:update-step}
\end{align}
for 
some 
step size $\lambda$.

By adding $\lambda\Delta H$
to the current Hamiltonian $H$, the updated Gibbs state 
\begin{align*}
	\rho_{H+\lambda \Delta H}=\exp(-(H+\lambda \Delta H))/\tr(\exp(-(H+\lambda\Delta H))) 
\end{align*}
will move towards to the separating plane (or even cross it when $\lambda$ is too large). 
 Note that $P_c$ is constant given $\gamma$, while $P_d^{\ell_1}$ depends on $\rho$ and therefore changes with each update. 
We will discuss quantitative bounds on the distance to the feasible region in Sec.~\ref{sec:entropy}.
The algorithm then iterates this Hamiltonian update procedure many times to get closer and closer to the feasible region -- hence the name.

Stated as is (and how it was proposed in Ref.~\cite{Brandao2022fasterquantum}), this meta-algorithm faces two major problems which renders it inefficient in practice: 
Firstly, an analytical running time analysis provided in Sec.~\ref{sec:entropy_apriori} shows that the required number of iterations depends quadratically on the precision $\epsilon$. 
Secondly, the Hamiltonian Updates algorithm detects infeasible instances by checking if an \emph{a priori} upper bound on the number of iterations has been reached.
This turns out to be rather inefficient in practice.
In this work, we provide substantial improvements that address both of these problems. This is the content of the next section.

\section{Improved convergence for Hamiltonian Updates}
\label{sec:class_imp}

In this section, we introduce several improvements to the Hamiltonian Updates meta-algorithm, which benefit both the classical and the quantum version. 
In Secs.~\ref{sec:Pd}-\ref{sec:momentum}, we present heuristics designed to reduce the number of iterations required for the algorithm to converge. 
Sec.~\ref{sec:entropy} gives improved estimation techniques for the change in relative entropy used to prove infeasibilty of a given problem instance. 
The complete enhanced algorithm is outlined as Alg.~\ref{alg:HU} and \ref{alg:update_mom}. 
Finally, in Sec.~\ref{sec:numerics_improvements} we compare these improvements numerically against the original algorithm.
\begin{algorithmCap}
\begin{algorithm}[H]
    \caption{Improved Hamiltonian Updates}\label{alg:HU}
    \begin{algorithmic}[1]
        \Require 
        Normalized cost matrix $C$, threshold objective value $\gamma$, precision parameter $\epsilon$, initial step lengths $\lambda_c$ and $\lambda_d$, momentum hyperparameter $\beta$ 
        \Statex
        \noindent
        \Ensure \quad \begin{tabular}[t]{l|l}  
              \textbf{Condition} & \textbf{Output} \\
              \hline
              (\ref{eqn:exact_program}) is feasible & an $\epsilon$-feasible $\rho$ \\
              (\ref{eqn:exact_program}) is not $\epsilon$-feasible & false \\
              else & undefined (an $\epsilon$-feasible $\rho$ or false)
        \end{tabular}
		\Statex
        
        \Function{hamiltonian\_updates}{$C, \gamma, \epsilon, \lambda_c, \lambda_d, \beta$}
        \State $P_c \gets - C + \gamma \Id$
        \State $H \gets 0_{n\times n}$
        \State $M \gets 0_{n\times n}$
        \State $F=-\ln(n)$ 
        \State $\rho \gets \frac{1}{n}\Id_{n\times n}$
        \\
        \While{$F\leq 0$}
            \Comment{Main loop of HU}
            \If{$\tr(P_c\rho)>\epsilon$}
                \State $\Delta H \gets \tr(P_c \rho) P_c + \frac{\beta}{\lambda_c} M$
                \State $H, \rho, F, \lambda_c 
                \gets \textsc{update}(H, \Delta H, \lambda_c)$
                \Comment{Apply cost update}
                \State $M \gets \lambda_c\Delta H$
                \Comment{Update momentum}
            \\
            \ElsIf{$\sum_i|\rho_{ii}-1/n|>\epsilon$}
                \State $P_d^{\ell_2} \gets (\diag(\rho)-\Id/n) / \max_i|\rho_{ii} - 1/n|$
                \State $\Delta H \gets P_d^{\ell_2} + \frac{\beta}{\lambda_d} M$
                \State $H, \rho, F, \lambda_d
                \gets \textsc{update}(H, \Delta H, \lambda_d)$
                \Comment{Apply diag.\ update}
                \State $M \gets \lambda_d\Delta H$
                \Comment{Update momentum}
                \\
            \Else
                \State \Return $\rho$
                \Comment{$\rho$ is $\epsilon$-feasible}
            \EndIf
            \\
        \EndWhile
        \State \Return false 
        \Comment{$F>0 \rightarrow$ No feasible solution exists}
    \EndFunction
     \end{algorithmic}
     \end{algorithm}
\end{algorithmCap}

\begin{algorithmCap}
     \begin{algorithm}[H]
         \caption{Update function for Hamiltonian Updates}\label{alg:update_mom}
         \begin{algorithmic}[1]
         \Require Hamiltonian $H$, update matrix $\Delta H$, current step length $\lambda_c$ or $\lambda_d$ 
         \Statex
        \Ensure updated Hamiltonian $H_\mathrm{new}$, current Gibbs state $\rho$, current free energy $F$, updated step length $\lambda_c$ or $\lambda_d$
        \Statex
        \Function{update}{$H, \Delta H, \lambda$}
            \State $H_\mathrm{new} \gets H + \lambda \Delta H$ 
            \Comment{Compute new Hamiltonian}
    
            \State \textbf{compute} $\exp(-H_\mathrm{new})$ 
            \Comment{Compute matrix exponential for $\rho$ and $F$}
            \State $\rho_\mathrm{new} \gets \exp(-H_\mathrm{new})/\tr(\exp(-H_\mathrm{new}))$
            \\
            \While{$\tr(\Delta H \rho_\mathrm{new}) < 0$}: 
						\Comment{Check for \emph{overshoots}}
            \State $\lambda \gets 0.5 \lambda$
            \Comment{Reduce step size}
            \State $H_\mathrm{new} \gets H + \lambda\Delta H$ 
            \Comment{Re-compute new Hamiltonian}
            \State \textbf{compute} $\exp(-H_\mathrm{new})$ 
            \State $\rho_\mathrm{new} \gets \exp(-H_\mathrm{new})/\tr(\exp(-H_\mathrm{new}))$
            \EndWhile
            \\
            \State $F \gets -\ln(\tr(\exp(-H_\mathrm{new})))$ 
            \Comment{Compute free energy}
            \State $\lambda \gets 1.3 \lambda$
            \Comment{Increase step size for the next iteration}
            \\
            \State \Return $H_\mathrm{new}, \rho_\mathrm{new}, F, \lambda$
        \EndFunction
    \end{algorithmic}
\end{algorithm}
\end{algorithmCap}

\subsection{Adaptive step length} \label{sec:adaptive_step_size}

Recall from Eq.~\eqref{eqn:update-step} above, that the individual updates of the Hamiltonian $H \in \mathbb{R}^{n \times n}$ take the following form:
\begin{align}
	H
	\mapsto 
	H
	+ \lambda
	P,
\end{align}
where $P$ can be either $P_c$ or $P_d^{\ell_1}$.
As is commonly the case for iterative algorithms, there is a trade-off in choosing the  \emph{step size} or \emph{learning rate} $\lambda$:
Small choices of $\lambda$ mean that violations of the constraints take many iterations to be corrected,
while too large values of $\lambda$ increase the danger of \emph{overshooting}. In our case, overshooting corresponds to moving to the other side of the separating hyperplane, in which case we do not have a guaranteed improvement anymore (c.f.\ Sec.~\ref{sec:entropy_tracking}).

The algorithm of Ref.~\cite{Brandao2022fasterquantum} uses a constant step length $\lambda = \epsilon/16$ that only depends on the desired target accuracy $\epsilon$.
In contrast, here we propose to use two \emph{adaptive step lengths} 
\begin{align*}
	\lambda_c \qquad \text{and} \qquad \lambda_d,
\end{align*}
one for each constraint in \eqref{eq:approx_feasibility_constraint}.
We choose the concrete step lengths
according to the following heuristic:
The algorithm maintains the current step sizes $\lambda_c$ and $\lambda_d$, which are increased by a constant factor after each corresponding update -- we find that multiplying with $1.3$ works well in practice.
We say that the algorithm has \emph{overshot} if, after an update, the sign of the constraint has reversed, i.e.\ if $\tr(P\rho_{H+\lambda_{c/d} P}) < 0$.
This sign is checked after every update.
In case an overshot did occur, $\lambda$ is halved, and $\rho$ is re-computed for the now smaller step size (c.f.\ Alg.~\ref{alg:update_mom}).

This approach means that if an overshoot occurs, we must recompute the Gibbs state, which incurs an overhead in computation time. Hence, the possibility of overshooting manifests itself in a slight increase of the average computation time per iteration, when compared to the constant step length method. 
However, as we demonstrate numerically in Sec.~\ref{sec:numerics_improvements}, this is more than compensated for by an overall faster step-wise progress which significantly reduces the total number of iterations required, especially in the early phase of the HU meta-algorithm.

\subsection{Euclidean-norm based $P_d$} \label{sec:Pd}
In the original Ref.~\cite{Brandao2022fasterquantum}, the authors address the violation of the diagonal constraint in \eqref{eq:approx_feasibility_constraint} 
using a matrix $P_d^{\ell_1}$ that corresponds to the $\ell_1$ norm: $\tr(P_d^{\ell_1} \rho)=\sum_i |\rho_{ii} - 1/n|$ (see Eq.~\eqref{eq:P_d_1}).
This approach is a natural choice, as this norm reflects the feasibility constraint in Eq.~\eqref{eq:approx_feasibility_constraint} which is also formulated in terms of the $\ell_1$ norm.
However, closer inspection reveals that the correction provided by $P_d^{\ell_1}$ may be suboptimal. 
After all, it only considers the sign of the deviations in each entry, not their magnitude. 

We propose a new approach where $P_d$ is proportional to the violation in each component. 
To achieve this, we modify \eqref{eq:P_d_1} by removing the sign function.
Because the trace term in \eqref{eq:P_d_1} becomes zero under this modification, the result is
\begin{align}
    P^{\ell_2}_d=\tilde\rho-\Id/n.
\end{align}
The trace with the modified matrix evaluates to the squared Euclidean or $\ell_2$ norm of the deviation:
\begin{align}\begin{split}
    \tr(\rho P_d^{\ell_2})
		&= \tr(\rho(\tilde\rho -2\Id/n)) + \tr(\rho \Id/n) \\
    &= \tr(\rho(\tilde\rho -2\Id/n)) + 1/n \\
    &= \sum_i \left(\rho^2_{ii} -2\rho_{ii}/n + 1/n^2 \right)\\
		&=\sum_i (\rho_{ii} - 1/n)^2.
\end{split}\end{align}
Note that this is a much more common choice as a loss function in gradient descent algorithms.
To further optimize and to improve numerical stability, one can experiment with different normalizations for $P_d^{\ell_2}$. 
Although the differences are generally minor, because the adaptive step size defined in the previous section adjusts well to different normalizations, we find experimentally that dividing $P_d^{\ell_2}$ by its maximum absolute entry yields the best results.

Additionally, we scale $P_c$ in each update with the corresponding distance $\tr(P_c\rho)$. Thus, the new matrix \begin{align}
\tilde P_c = \tr(P_c\rho) P_c,
\end{align} 
corresponds to  larger corrections in the cost update when $\rho$ is further away from the feasible region. 

In Fig.~\ref{fig:beta_scaling}, we compare the performance of HU using $P_d^{\ell_2}$ instead of $P_d^{\ell_1}$, observing that this modification results in a speedup of approximately a factor of two to three.

\subsection{Adding a momentum term}\label{sec:momentum}

In gradient descent methods, it is common to also add a so-called \emph{momentum term}~\cite{polyak19641}, which often empirically increase the speed of convergence.

In the following, we use a superscript $(\cdot)^{(k)}$ 
to refer to the value of a variable in the $k^{\mathrm{th}}$ iteration.
Define the \emph{momentum term} to be
\begin{align*}
    M^{(k)}
		&=
     \left\{
			 \begin{array}{ll}
		     \lambda_{c}^{(k)} (\Delta H)^{(k)}  & \text{for cost update in $k^\mathrm{th}$ step,} \\
		     \lambda_{d}^{(k)} (\Delta H)^{(k)}  & \text{for diag. update in $k^\mathrm{th}$ step,} 
			 \end{array}
		 \right. 
		 \\
		 M^{(0)}&=0. 
\end{align*}

Next, choose a new  hyperparameter $\beta\in(0,1)$ and modify the update rule (\ref{eqn:update-step}) to read
\begin{align}
	H^{(k+1)} &= 
    \left\{
			 \begin{array}{ll}
		  H^{(k)} + \lambda_c^{(k)} (\Delta H)^{(k)} & \text{for cost update in $k^\mathrm{th}$ step,}
		\vspace{.1cm}\\
		H^{(k)} + \lambda_d^{(k)} (\Delta H)^{(k)} & \text{for diag. update in $k^\mathrm{th}$ step,} 
			 \end{array}
		 \right. \\ 
    \notag \\
	\label{eq:Delta_H}
	(\Delta H)^{(k)} &=
     \left\{
			 \begin{array}{ll}
		  \big(\tilde P_c\big)^{(k)} + \frac{\beta}{\lambda_{c}^{(k-1)}} M^{(k-1)} 
          & \text{for cost update in $k^\mathrm{th}$ step,}
		\vspace{.1cm}\\
		\big(P_d^{\ell_2}\big)^{(k)} + \frac{\beta}{\lambda_{d}^{(k-1)}} M^{(k-1)}
		  & \text{for diag. update in $k^\mathrm{th}$ step.} 
			 \end{array}
		 \right. 
\end{align}

Numerically, we find that a values of $\beta$ between $0.4$ and $0.5$ achieve the best results, with reductions in the number of iterations by roughly 30-40\% (c.f.\ Fig.~\ref{fig:beta_scaling} and Tab.~\ref{tab:cumu_improvements}).

\subsection{Free energy tracking}\label{sec:entropy}

When we cannot guarantee that a particular SDP instance is feasible, it is crucial to find a \emph{termination criterion}; otherwise, the HU routine would run indefinitely. 
This can be achieved by bounding the \emph{quantum relative entropy}
\begin{align} \label{eq:rel_ent_def}
    R(\rho^\star \| \rho) = \tr \left( \rho^\star (\log\rho^\star - \log\rho) \right)
\end{align}
between 
any solution $\rho^\star$
(assuming that one exists)
and the current state $\rho$. 
For properties of the quantum relative entropy, we refer to standard textbooks, e.g.\ Ref.~\cite{Renes2022}.

It is known that the relative entropy distance between the maximally mixed state $\rho_0=\Id/n$ and any other state
is upper-bounded by $\ln(n)$. 
We demonstrate in this section that one can lower-bound the decrease in relative entropy distance between the current Gibbs state $\rho$ and $\rho^\star$ (if it exists) in each update. 
By choosing the maximally mixed state as the initial state $\rho_0$, we are guaranteed that the cumulative change in relative entropy cannot exceed $\ln(n)$, if a solution does indeed exist. 
By the same token, if the estimate of total relative entropy distance reduction exceeds $\ln(n)$, it is certain that a feasible solution does not exist. 

This ``relative entropy tracking'' procedure serves two different purposes:
First, if the algorithm detects that it would have covered a relative entropy distance of $\ln(n)$, but has not yet found a solution, we know that the problem is infeasible.
Using terminology standard in quantum information, we could call this a \emph{heralded} event: 
Meaning that if it occurs, we can draw rigorous conclusions from it, but, at the beginning of the algorithm, it is unclear after how many iterations it will be detected.
Second, we would like to have an \emph{a priori} upper bound on the number of iterations required before infeasibility is detected.

Ref.~\cite{Brandao2022fasterquantum} uses a single bound to serve both these purposes.
In contrast, we report a large gain in practical performance by using different estimates for the two goals.
In Sec.~\ref{sec:entropy_tracking}, we introduce a new method for tracking the entropy change that makes use of another quantity from statistical mechanics called \emph{free energy}.
We observe numerically that this new approach can achieve speedups over the original method of a factor $>10000$ (c.f.\ Tab.~\ref{tab:cumu_improvements}).

Our a priori bound on the number of iterations is presented in Section~\ref{sec:entropy_apriori}.
Compared to 
Ref.~\cite{Brandao2022fasterquantum},
it includes a treatment of the momentum term (c.f.\ Section~\ref{sec:momentum}) 
and it improves the estimate by a constant factor.
The theoretical guarantee does not cover all heuristics we employ, in particular it does not take adaptive step sizes and the improved way of choosing $P_d$ into account.
We re-iterate that these are pessimistic worst-case bounds and that the observed practical performance is much better.

\subsubsection{Termination criterion}\label{sec:entropy_tracking}

On a high level, our improved relative entropy tracking method exploits the observation that the decrease in distance to the feasible set is larger, the further away the current state is from being feasible.
In contrast, \cite{Brandao2022fasterquantum} uses a worst-case bound that does not take into account the faster decrease of the relative entropy distance that happens especially in the early steps of the algorithm. 

We now show how the relative entropy $R$ can be estimated during the HU routine.  
We can rewrite the definition of the relative entropy \eqref{eq:rel_ent_def} as
\begin{align}
        R(\rho^\star||\rho)
		=
		\tr( \rho^\star \ln (\rho^\star))
		+
		\tr (\rho^\star H)
		+
		\ln( \tr( \exp(-H))).
\end{align}
In statistical mechanics, 
\begin{align} \label{eq:free_energy}
 F(H)=-\ln( \tr( \exp(-H)))   
\end{align}
is called the \emph{free energy} (at inverse temperature $1$). 
Initially, for $H_0=0$ we have $F(0)=-\ln(n)$. Then, the total change in relative entropy with respect to the initial Gibbs state $\rho_0=\Id/n$ is given by
\begin{align}
    \Delta R(\rho_0, \rho) 
    =
    R(\rho^\star||\rho) - R(\rho^\star||\rho_0) 
    =
    \tr (\rho^\star H)
    -
    F(H)
    -
    \ln(n).
\end{align}
In the HU algorithm, the Hamiltonian $H$ in the $k^\mathrm{th}$ iteration is of the form 
\begin{align*}
	H^{(k)}=\sum_{k'=1,\dots, k} c_{k'} P^{(k')},
\end{align*}
with coefficients $c_{k'}> 0$.
By construction, the matrices $P^{(k')}$ penalize infeasible directions and thus fulfil $\tr(\rho^\star P^{(k')})\leq 0$. Then, we always have $\tr(\rho^\star H)\leq 0$ and therefore,
\begin{align}
    \Delta R(\rho_0, \rho)
    \leq
    -F(H)
    -
    \ln(n).
\end{align}
Next, we use that the absolute change in relative entropy distance is upper bounded by $\ln(n)$ and therefore
\begin{align}
    -\ln(n)
    \leq
    \Delta R(\rho_0, \rho)
    \leq
    -F(H)
    -
    \ln(n).
\end{align}
Thus, we know that if a feasible solution $\rho^\star$ exists, we always have 
\begin{align} \label{eq:free_energy_condition}
    F(H) 
    \leq
    0
\end{align}
for any $H$ occurring as part of the HU routine. 
Hence, the task of tracking the relative entropy translates into evaluating the free energy. 
Classically, this quantity can be easily computed with no relevant additional computational effort, as one already has to compute $\exp(-H)$ for the Gibbs state in each iteration. 

On a quantum computer, we instead bound this quantity indirectly via its derivative. 
For this, consider a single update. 
Let $H$ be the initial Hamiltonian, and $H+\lambda \Delta H$ the one after the HU step.
The change in free energy is
\begin{align}
    \Delta F
		= F(H+\lambda\Delta H) - F(H),
\end{align}
and its derivative with respect to $\lambda$ is given by the expectation value of the update term $\lambda\Delta H$ with respect to the final state:
\begin{restatable}{lemma}{entropyBoundLinear} \label{thm:entropyBoundLinear}
	For all symmetric $H, \Delta H \in \RR^{n\times n}$ and $\lambda \in \RR$,
		the free energy satisfies
    \begin{align} \label{eq:F_derivative}
			   \partial_\lambda F(H+\lambda\Delta H) = \tr(\rho_{H+\lambda\Delta H} \lambda\Delta H). 
    \end{align}
\end{restatable}
The proof is given in Appendix \ref{sec:EntropyProof}.

Now, by integrating both sides in \eqref{eq:F_derivative}, the change in free energy is given by:
\begin{align} \label{eq:DeltaF_integral}
    \Delta F 
    = 
		\int_0^\lambda \tr(\rho_{H+\lambda'\Delta H} \lambda'\Delta H)\,\mathrm{d}\lambda'.
\end{align}
It is known that the function
\begin{align*}
	\lambda \mapsto
	F(H+\lambda \Delta H)
\end{align*}
is concave (a fact sometimes referred to as \emph{Bogoliubov inequality} \cite[Lem.~2]{entropyInequalities} in quantum statistical mechanics).
Thus, we can bound $\Delta F$ 
by evaluating $\tr(\rho_{H+\lambda\Delta H} \lambda\Delta H)$ for multiple values $0<\lambda'\leq\lambda$ and computing
\begin{align} \label{eq:F_approx}
    \Delta F 
    \geq
    \sum_{\lambda'_i}
    (\lambda'_i-\lambda'_{i-1}) \tr(\rho_{H+\lambda'_i\Delta H} \Delta H).
\end{align}
These expected values can be computed relatively cheaply on a quantum computer (compared to the cost of estimating the diagonal of $\rho$, see\ Sec.~\ref{sec:quant_imp}). 

We have studied the behavior of the improved termination criterion numerically.
The results are shown in Fig.~\ref{fig:gamma_scaling}.

\begin{figure}[H]
    \centering
    \begin{subfigure}[b]{\textwidth}
        \includegraphics[width=\textwidth]{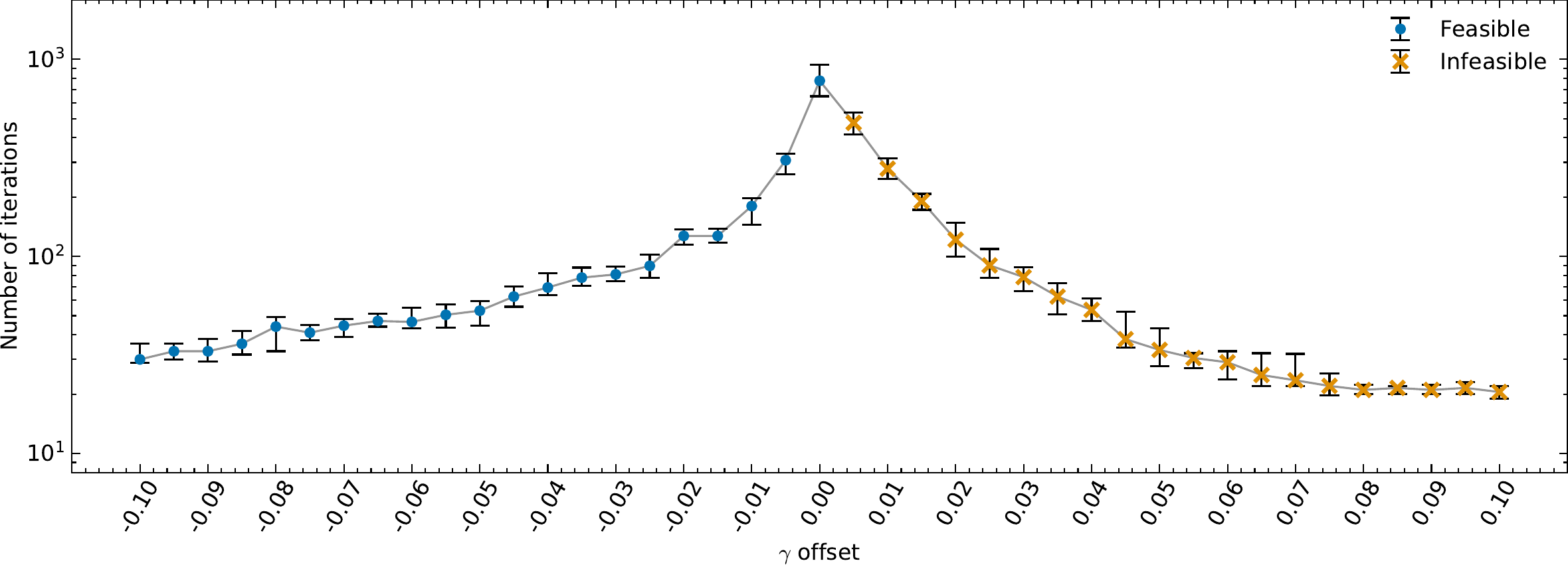}
    \end{subfigure}\\
    \hfill
    \begin{subfigure}[b]{0.987\textwidth} 
        \includegraphics[width=\textwidth]{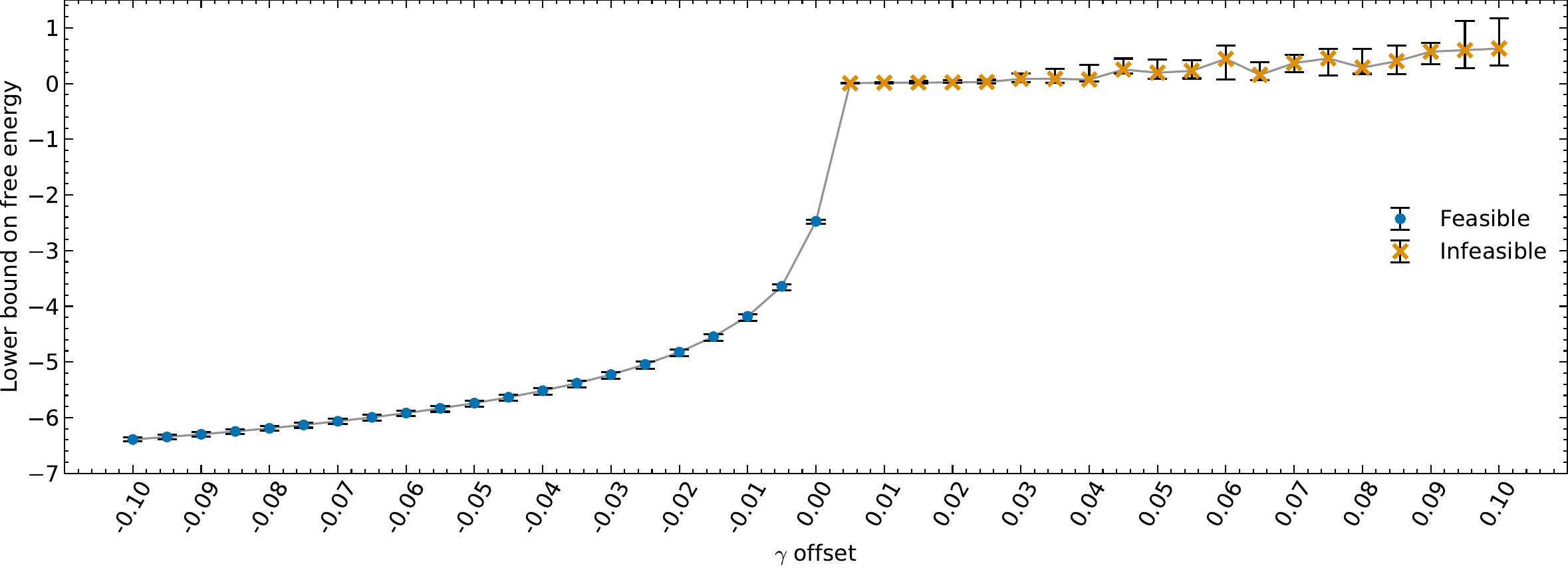}
    \end{subfigure}
		\caption{
			Behavior of the improved HU algorithm on an instance used with parameters $n=1024$, $s=16$ and $\epsilon=0.001$, generated as described in Sec.~\ref{sec:numerics_improvements}. 
			The algorithm terminated after finding an $\epsilon$-feasible solution for
			$\gamma:=\gamma^\star+\mathrm{offset}$, 
			or once the free energy became positive.
			Here, $\gamma^\star$ is the optimal objective value determined by an SDP solver \cite{ODonoghue2016}.
            The error bars show the upper and lower quartiles.
			\\
			\textbf{Upper panel}:
			Number of iterations required until convergence as a function of the $\gamma$-offset.
			The improved termination criterion comes into play on the right tail of the curve.
			We observe that the number of iterations required to certify infeasibility goes down the further the problem specification is from a feasible one.
			(This contrasts to the fixed termination criterion used in \cite{Brandao2022fasterquantum}).
			\\
			\textbf{Lower panel}:
            The lower bound on the free energy at time of termination.
			We observe numerically that the bound increases as $\gamma$ gets closer to the optimal value.
			Whether this effect can be exploited for algorithmic improvements is a question we leave open.
	}
    \label{fig:gamma_scaling}
\end{figure}

\subsubsection{Convergence guarantee} \label{sec:entropy_apriori}

We now provide an \emph{a priori} bound on the maximum number of steps required for the algorithm to find a feasible solution, assuming that one does exist.
The result goes beyond Ref.~\cite{Brandao2022fasterquantum} in two ways:
It includes a treatment of the momentum term (c.f.\ Section~\ref{sec:momentum}) and it improves the estimate by a constant factor.

\begin{restatable}{theorem}{entropyBoundConvergence} \label{thm:entropyBoundConvergence}
    For an HU routine using $P_d^{\ell_1}$ in the diagonal update
     as defined in Eq.~\eqref{eq:P_d_1}, momentum as defined in Sec.~\ref{sec:momentum} and a step length $\lambda=\frac{(1-\beta)^2}{2}\tr(\rho_{H}\Delta H)$, the maximum number of steps needed to find an $\epsilon$-feasible solution to a feasible program \eqref{eqn:exact_program} is upper bounded by
\begin{align}
    T = 16 (1-\beta)^{-6} \epsilon^{-2} \ln(n).
\end{align}
\end{restatable}

The proof is given in Appendix \ref{sec:EntropyProof2}. 

Without the use of momentum terms (i.e.\ setting $\beta=0)$, this becomes $T=16 \epsilon^{-2} \ln(n)$, thus improving the bound of Ref.~\cite{Brandao2022fasterquantum} by a factor of $4$.
We repeat that the numerically observed speedup of our improvements compared to Ref.~\cite{Brandao2022fasterquantum} is much larger than the improvement of the \emph{a priori} bounds.

We note that the proof of convergence no longer holds when using $P_d^{\ell_2}$ for the diagonal update as proposed in Sec.~\ref{sec:Pd}, while still using $P_d^{\ell_1}$ in the feasibility criteria, because the proof assumes $\tr(\rho_{H}\,P_d)\geq\epsilon$, which is not necessarily the case for mixed $P_d$'s. Additionally, Thm.~\ref{thm:entropyBoundConvergence} does not consider an adaptive step size, as in Sec.~\ref{sec:adaptive_step_size}. However, in practice, we find that these modifications do substantially improve algorithmic performance.

\subsection{Numerical benchmarks of the non-asymptotic improvements} \label{sec:numerics_improvements}

Here, we numerically study the effects of the different improvements made in Secs.~\ref{sec:adaptive_step_size}-\ref{sec:entropy_tracking}.
This section consists of three parts: 
\begin{enumerate}
	\item
		We observe the decrease in the 
		number of iterations when applying the diagonal update $P^{\ell_2}_d$ instead of $P^{\ell_1}_d$ as described in Sec.~\ref{sec:Pd},
		together with the momentum term described in Sec.~\ref{sec:momentum} for different values of $\beta$.
	\item
        We observe the decrease in the required number of iteration when successively applying all the improvements compared to the original algorithm. 
	\item
        We analyze how the number of iterations of the fully improved HU routine scales with $\epsilon$. 
\end{enumerate}
The simulations are performed on 20 sparse QUBO instances of the block form given in Eq.~\eqref{eq:C_block}.
The sparsity pattern is chosen uniformly at random with sparsity $s=16$.
The non-zero elements are sampled from a standard Gaussian distribution, the matrices are normalized.

For the first part, we use instances with dimension $n=1024$ and require a target accuracy $\epsilon=0.001$. We run the complete HU routine using either $P^{\ell_1}_d$ or $P^{\ell_2}_d$ in all diagonal updates and test values for the momentum hyperparameter $\beta$ between 0 and 0.7. To have a direct comparison, we choose to compare feasible instances of \eqref{eqn:relaxed} with target objective values $\gamma$ equal to the optimal objective value $\gamma^\star$, instead of applying full binary searches. 
The optimum $\gamma^\star$ is obtained by an SDP solver beforehand for each instance (specifically, the Splitting Conic Solver (SCS) described in Ref.~\cite{ODonoghue2016}).
The results are displayed in Fig.~\ref{fig:beta_scaling}.

\begin{figure}[H]
    \centering
    \begin{subfigure}[b]{0.99\textwidth}
        \includegraphics[width=\textwidth]{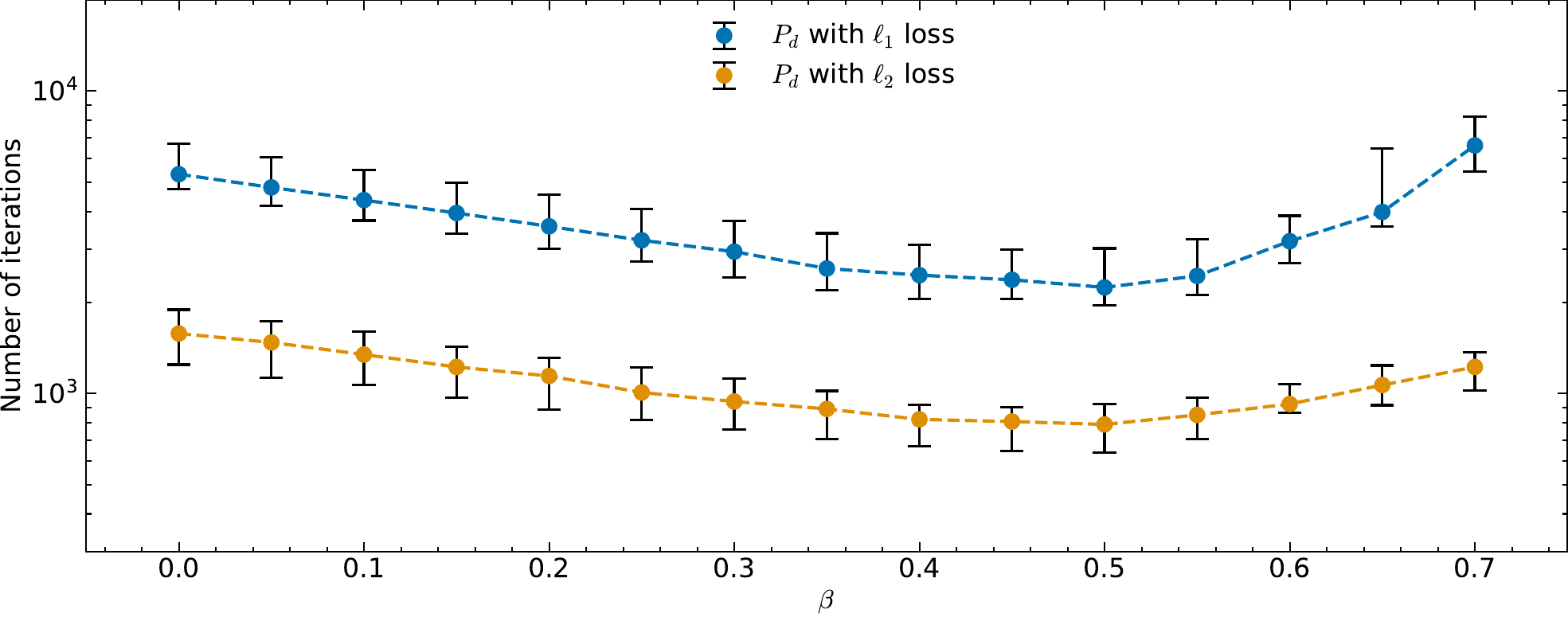}
    \end{subfigure}
    \caption{
		Number of iterations for varying values of the momentum hyperparameter $\beta$. The two different approaches for the diagonal update are compared: The original $\ell_1$ norm based $P_d^{\ell_1}$ (blue) and the new $\ell_2$ norm based $P_d^{\ell_2}$ (orange). The instances have dimension $n=1024$, sparsity $s=16$ and precision $\epsilon=0.001$, and use the optimal SDP solution $\gamma^\star$ as the target objective value. The error bars show the upper and lower quartiles.}
    \label{fig:beta_scaling}
\end{figure}

For the second part, to compare the improved algorithm to the original one, we evaluate three categories: (i) solving a feasible instance, (ii) proving infeasibility of an instance, and (iii) performing a complete binary search.
The numerical simulations were performed on instances with dimension $n=128$ and target accuracy $\epsilon=0.01$. 
For the feasible instances we used an optimal target objective value $\gamma^\star$ obtained from an SDP solver 
(SCS), while infeasible instances used the candidate $\gamma=\gamma^\star+0.02$. 
The comparison is made both in terms of the number of iterations and the number of matrix exponentiations required.
The latter make up the vast majority of computational cost in each iteration, and thus provide a good measure for comparing the running time of the different approaches.
The results are displayed in Table~\ref{tab:cumu_improvements}.

\begin{table}[H] 
    \centering
    \resizebox{1\textwidth}{!}{
    \begin{tabular}{|l||c c |c c |c c|}
    \hline
         \multirow{2}{*}{cumulative improvements}  & \multicolumn{2}{|c|}{feasible instance} & \multicolumn{2}{|c|}{infeasible instance} & \multicolumn{2}{|c|}{binary search} \\
        & iterations & matrix exp. & iterations & matrix exp. & iterations & matrix exp. \\
        \hline
        original algorithm & 88292 & 88292 & 3.11e+06$^\dagger$ & 3.11e+06$^\dagger$ & 1.28e+07$^\dagger$ & 1.28e+07$^\dagger$ \\
        with adaptive step size & 171 & 241 & 3.11e+06$^\dagger$ & 4.27e+06$^*$ & 1.26e+07$^\dagger$ & 1.55e+07$^*$ \\
        with entropy tracking & 171 & 241 & 104 & 144 & 802 & 1116 \\
        with $\ell_2$ norm based $P_d$ & 62 & 86 & 54 & 72 & 323 & 439 \\
        with momentum & 42 & 59 & 38 & 50 & 219 & 296 \\
        \hline
    \end{tabular}}
    \caption{
		Comparison in terms of HU iterations and matrix exponential computations for the original HU algorithm and the new one with the improvements from Secs.~\ref{sec:adaptive_step_size}-\ref{sec:entropy_tracking} applied cumulatively. The number of matrix exponentials is equal to the sum of the number of iterations and the number of overshoots. The results are averaged over 20 instances with dimension $n=128$, sparsity $s=16$ and precision $\epsilon=0.01$.
		Values computed (partially) analytically using the termination criterion of Ref.~\cite[Thm.~2.1]{Brandao2022fasterquantum} are marked with a dagger ($\dagger$).
		Values marked with an asterisk ($*$) are extrapolated: The total number of overshoots is estimated by observing the average number of overshoots per iteration and multiplying this with the analytical number of iterations from the termination criterion.
	We find a speedup by a factor of more than $1400$ for solving the SDP for an optimal candidate $\gamma^\star$, and a speedup of more than $43000$ for the complete binary search. 
	}
    \label{tab:cumu_improvements}
\end{table}

Finally, in the third part, we analyze the scaling of the number of iterations of the 
improved HU algorithm as a function of the precision $\epsilon$.
We have used instances with dimension $n=1024$ and $\epsilon$ values between $10^{-2}$ and $10^{-3.4}$. 
The indicated number of iterations includes a full binary search to find the optimal value up to the desired precision.
A numerical powerlaw extrapolation shows that the number of iterations scales as $0.006\epsilon^{-2.02}$. 
The results are shown in Fig.~\ref{fig:eps_iterations}. 

\begin{figure}[H]
\centering
        \includegraphics[width=0.75\textwidth]{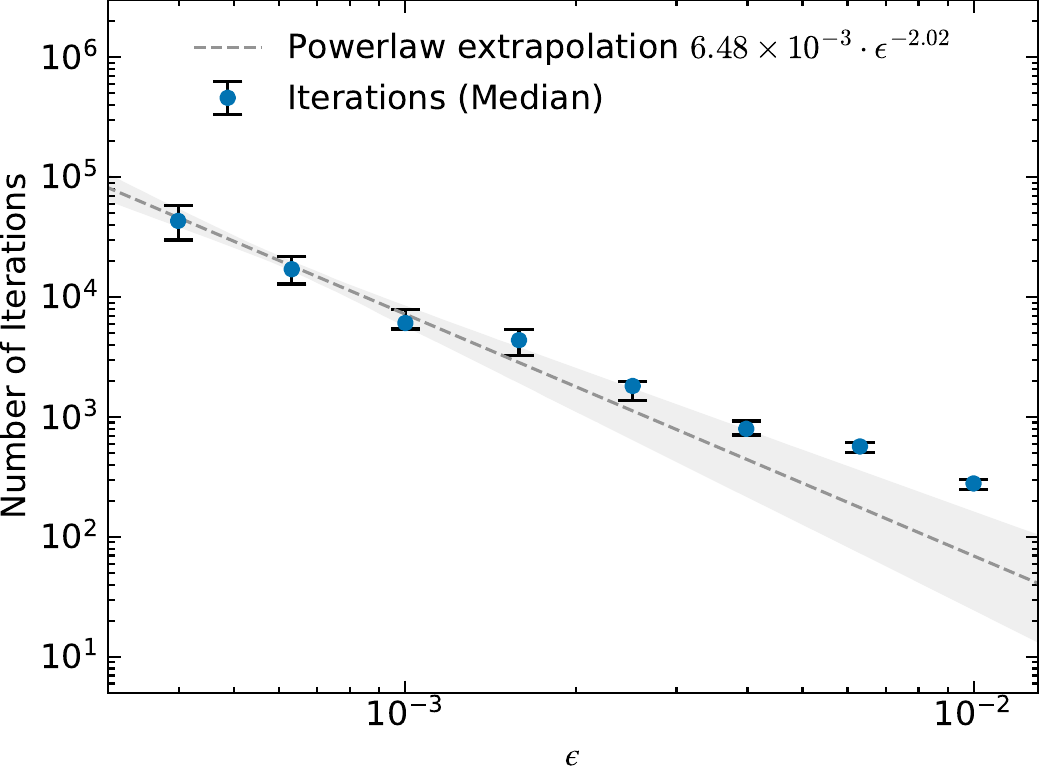}
		\caption{Total number of iterations over a complete binary search as a function of the precision $\epsilon$. The fit is given by $f(\epsilon)=6.48\times 10^{-3} \epsilon^{-2.02}$ with a 95\% confidence interval of $(-2.38, -1.73)$ for the exponent of $\epsilon$. Therefore, the experimentally observed scaling exponent is similar to the theoretical bound of $-2$ (however, the experimentally observed prefactor is significantly better). The error bars show the upper and lower quartiles. The gray area shows the 95\% confidence interval of the fit.}
        \label{fig:eps_iterations}
\end{figure}

\section{Improved randomized rounding}
\label{sec:rounding_section}

\subsection{Analytical results} 
\label{sec:rounding}

After solving the SDP relaxation, we still need to obtain a solution for the original QUBO problem \eqref{eqn:quadratic_opt}. 
The Goemans-Williamson algorithm \cite{goemansWilliamson} provides a randomized rounding procedure for this purpose. 
One first computes the square root of $\rho$ in terms of functional calculus, i.e. $\sqrt{\rho}\in \mathbb{R}^{n\times n}$ is a symmetric matrix with $\sqrt{\rho}\sqrt{\rho}=\rho$. 
Then, the entries of the rounded solution $x\in \{-1,1\}^n$ are given by the sign of the column-wise projection of $\sqrt{\rho}$ onto a Gaussian random vector:
\begin{align}
    & \text{compute} \quad && {\sqrt{\rho}}, \quad & \\
    & \text{sample} \quad && {g_j \overset{\mathrm{iid}}{\sim} \N(0,1),} &  \quad j\in [n], \\
    & \text{compute} \quad && {x_i \gets \sgn\left(\sum_{j} (\sqrt{\rho})_{ij} g_j \right),} &  \quad i\in [n].
\end{align}
The Goemans-Williamson-style bounds on the quality of rounded solutions found in the literature
are not directly applicable if the diagonal entries of $\rho$ are only \emph{approximately} equal to $1/n$. 
To address this issue, Ref.~\cite{Brandao2022fasterquantum} proposes two different solutions:
(1) First map the approximate optimizer $\rho$ to a matrix $\rho^\sharp \in \mathbb{R}^{n\times n}$ that fulfills the diagonal constraints exactly, and then apply the above rounding procedure.
(2) Apply the rounding procedure directly to the approximate optimizer $\rho$.

Both approaches show the same asymptotic scaling behavior.

Ref.~\cite{Brandao2022fasterquantum} shows that if 
$\sum_i|\rho_{ii}-1/n|\leq \epsilon$,
the error of the final result scales as 
$\O(\epsilon^{1/4})$.
Here, we adjust the parameters that go into the correction procedure, and improve the scaling to $\O(\epsilon^{1/3})$.

\subsubsection{Correcting the SDP solution}

The improved performance of the first approach, the one where the SDP solution is mapped to one which satisfies the constraints exactly, results from the theorem below.
 \begin{restatable}{theorem}{muScaling}\label{thm:muScaling}
       There is an efficient procedure which, given an $\epsilon>0$ and a psd matrix $\rho \in \RR^{n\times n}$ such that 
	$\sum_i|\rho_{ii}-1/n|\leq \epsilon$,
	constructs a psd matrix $\rho^\sharp \in \RR^{n\times n}$ where
    \begin{align}
				&\rho^\sharp_{ii}=\frac 1n  \quad &\forall i\in[n] \\
        \mathrm{and} \qquad
				&\norm{\rho^\sharp - \rho}_\mathrm{tr}=\O(\epsilon^{1/3}). \quad &
    \end{align}
\end{restatable}
The proof is given in Appendix \ref{sec:stabilityProof}.

Combining Thm.~\ref{thm:muScaling} with the matrix Hölder inequality allows us to bound the change in the objective value that results from applying this correction:
\begin{align}
    |\tr(C \rho^\sharp)-\tr(C \rho)|
    \leq
    \|C\| \norm{\rho^\sharp - \rho}_\mathrm{tr}
    =
		\O(\epsilon^{1/3} ).
\end{align}

\subsubsection{Rounding directly from the approximate solution}

The improved results for the second approach -- rounding directly -- matches this scaling.

\begin{restatable}{theorem}{rounding}\label{thm:rounding}
    Let $\rho^\star \in \RR^{n\times n}$ be the optimal SDP solution corresponding to a normalized cost matrix $C\in \RR^{n\times n}$ with a block structure as defined in \eqref{eq:C_block}. 
		Let $\rho\in \RR^{n\times n}$ be an approximate solution with 
    \begin{align}
        &\tr(C\rho^\star) - \tr(C\rho)  \leq \epsilon, \quad& \\
        \mathrm{and} \qquad
        &\sum_i |\rho_{ii}-\frac{1}{n}| \leq \epsilon \quad &\forall i\in[n],
    \end{align}
		with $0\leq\epsilon\leq 1/2$.
    Let $x\in \{-1,1\}^{n}$ be the vector obtained by applying the randomized rounding procedure to $\rho$.
		Then
    \begin{align}
        \EE[x^T C x] 
        &\geq 
        \left(\frac{4}{\pi} - 1\right) n \tr(C \rho^\star) - \O(n\epsilon^{1/3}).
    \end{align}
\end{restatable}
The proof is given in Appendix \ref{sec:rounding_proof}.

\subsection{Numerical simulations for $\epsilon$ dependence} \label{sec:eps_scaling}

In the previous section, we have provided an analytic worst-case bound for the precision 
of the corrected SDP solution of $\O(\epsilon^{1/3})$.
Now, we study the actual scaling behavior numerically. 
To this end, we use the matrices $\rho$ that have been obtained as part of the numerical simulation described in Fig.~\ref{fig:eps_iterations}, and apply the rounding procedure $10^5$ times to each one. 
We quantify the performance of the rounding procedure in two ways:
(1)
The average objective value after rounding
(because this is the quantity the theoretical guarantees make direct statements about).
(2)
The maximum objective value (because this is the number that would be used as the output of a numerical procedure).
To reduce statistical fluctuations, the plots below show the average of 100 maximal values, computed for batches of size 1000 each.

In a second step, we compare the objective values obtained from rounding HU solutions to the objective values obtained from rounding exact solutions from an SDP solver (Splitting Conic Solver \cite{ODonoghue2016}).
To this end, let $x_\epsilon, x_\mathrm{opt} \in \{-1,1\}^n$ be vectors obtained from applying randomized rounding to the HU solution and the SDP solver solution respectively. 
We then define the precision
\begin{align} \label{eq:nu}
   \nu=(x_\mathrm{opt}^TC x_\mathrm{opt}- x_\epsilon^TC x_\epsilon)/n    
\end{align}
of the rounded solution.

We display the results in Fig.~\ref{fig:epsilon_scaling} and summarize the estimated asymptotic scaling behaviors in Sec.~\ref{sec:asymptotic_improvements}.

\begin{figure}[H]
    \centering
    \begin{subfigure}[b]{0.48\textwidth}
        \includegraphics[width=\textwidth]{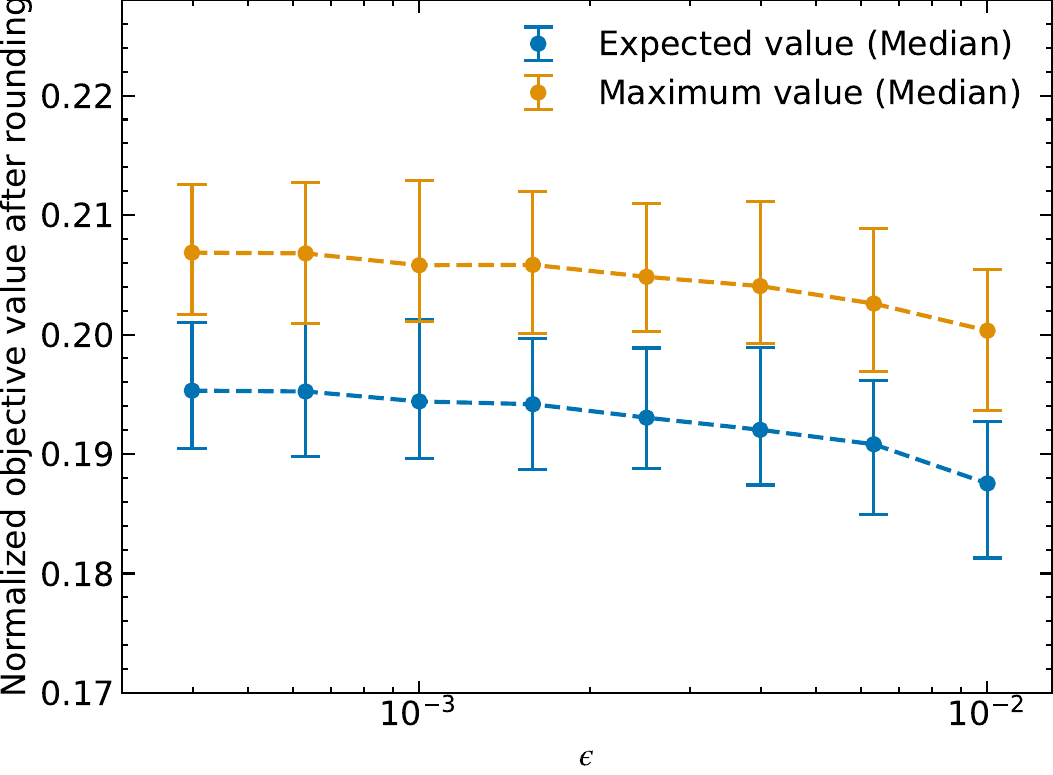}
    \end{subfigure}
    \hfill
    \begin{subfigure}[b]{0.48\textwidth}
        \includegraphics[width=\textwidth]{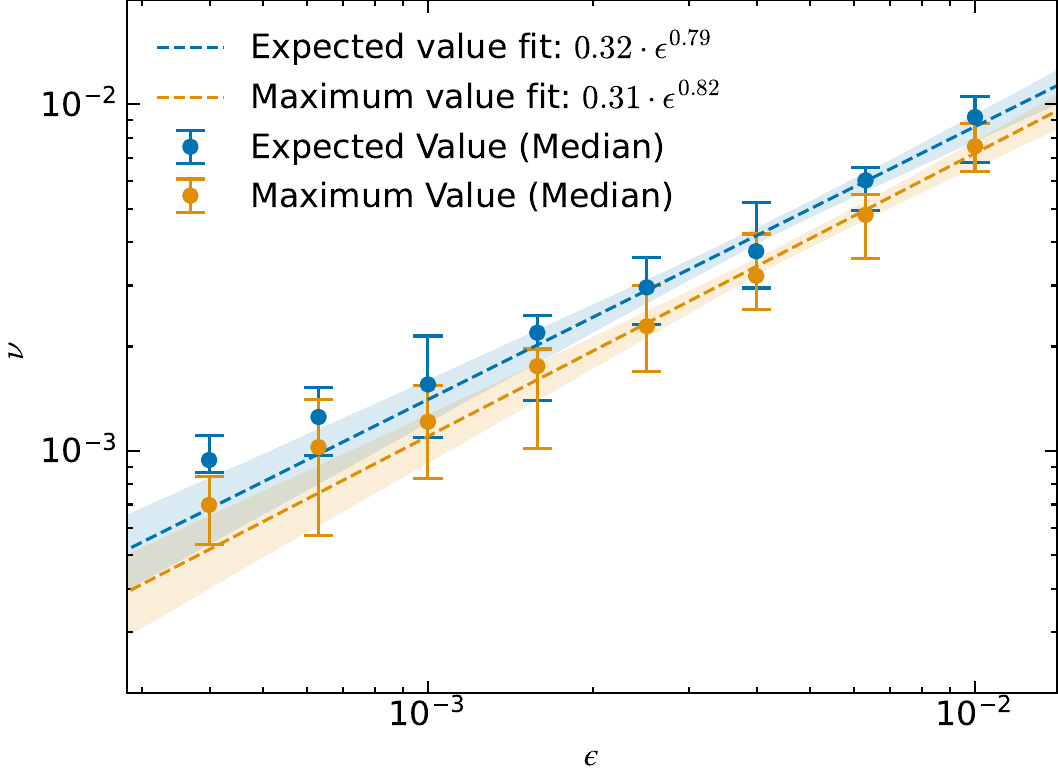}
    \end{subfigure}
    \caption{
    Behavior of the objective values $x^TC x / n$ after rounding. Results are obtained by performing a binary search over HU instances for different precision parameters $\epsilon$ and applying randomized rounding. The results are averaged over 20 cost matrices with dimension $n=1024$ and sparsity $s=16$. The error bars show the upper and lower quartiles.
    \\
    \textbf{Left panel}: 
		Comparison of average objective value (blue) and the maximum value seen in 1000 roundings (orange). 
    \\
    \textbf{Right panel}: 
    Difference $\nu=(x_\mathrm{opt}^TC x_\mathrm{opt}- x_\epsilon^TC x_\epsilon)/n$,
		where 
		$x_\mathrm{opt}$ 
		and
		$x_\epsilon$
		have been obtained, respectively,
		by applying the rounding procedure to an optimal SDP solution,
		and an $\epsilon$-precise HU solution. 
		The fit for the average value is given by $f_{\nu}(\epsilon)=0.32 \epsilon^{0.79}$ with a 95\% confidence interval of $(0.71,0.88)$ for the exponent. The colored areas show the 95\% confidence interval of the fits.
    }
    \label{fig:epsilon_scaling}
\end{figure}

We observe that for the parameters tested, 
taking the best solution obtained from 1000 randomized rounding runs
improves the final objective value significantly more 
than increasing the precision $\epsilon$ from $10^{-2}$ to $10^{-3.4}$.
Because the rounding procedure is relatively computationally cheap, 
it seems advisable to work with a very high number of randomized roundings.
We leave the task of determining the optimal tradeoff between the number of roundins and $\epsilon$ open for future work.

From the second part of the analysis, we find that the average of $\nu$ scales with $\O(\epsilon^{0.79})$. 
In contrast, the lower bounds for the two summands in Eq.~(\ref{eq:nu}) differ by a term of order $\O(\epsilon^{1/3})$.
Therefore, the numerically observed scaling of $\nu$ as a function of $\epsilon$ is significantly better than a naive estimate based on the difference of the lower bounds would have suggested.

\section{Improved Gibbs state simulation}
\label{sec:quant_imp}

In this section, we present improvements to the quantum version of the HU algorithm, which was originally given in Ref.~\cite{Brandao2022fasterquantum}.
On a high level, the idea of the quantum implementation is to realize the Gibbs state $\rho$ as the \emph{physical} state of a quantum system.
The advantage of this method is that the dimension of matrices representable in this way scales exponentially with the number of qubits.
The disadvantage is that information about the violation of the constraints and about the objective value have to be estimated statistically from physical measurements.
We refer to Ref.~\cite{Brandao2022fasterquantum} for a thorough description.
In this section, we only account for the parts of the method for which we suggest improvements.

The quantum version of the HU algorithm delegates a number of subroutines to a quantum computer.
More specifically,
given a classical descriptions of $H$, $P_c$, and $\Delta H$, it uses quantum subroutines to estimate the trace inner products $\tr(P_c\rho)$ and $\tr(\Delta H\rho)$, and the main diagonal elements $\rho_{ii}$. 
Recall that $\rho$ is the Gibbs state for $H$.

As already stated in \cite[Lem.~3.3]{Brandao2022fasterquantum}, the quantum routine for estimating the diagonal elements dominates  the running time.
In their approach,
obtaining an estimate for the probability distribution $\rho_{ii}$ up to an error of $\O(\epsilon)$ in $\ell_1$ norm
for the Gibbs states $\rho$ that appear in the HU algorithm requires a number of gates that scales as
\begin{align*}
	\tilde\O(n^{3/2}s^{1/2+o(1)}\epsilon^{-5+o(1)}),
\end{align*}
where the $\tilde \O$ notation hides logarithmic factors.
In contrast, we will argue that this scaling can be improved to
\begin{align*}
    \tilde \O(n^{3/2}s^{1/2+o(1)}\epsilon^{-3+o(1)}).
\end{align*}

The improvement is mainly achieved by invoking newer and more optimal methods for working with Gibbs states on a quantum computer.
Specifically, \cite{Brandao2022fasterquantum} was based on Ref.~\cite{PoulinGibbs2009}, while we switch to 
Ref.~\cite{vanApeldoorn2020quantumsdpsolvers}.

That reference, in turn, builds on a subroutine for \emph{Hamiltonian $\epsilon$-simulation}, i.e.\ for the task of implementing a unitary $U$ that is $\epsilon$-close to $e^{i t H}$ in operator norm,
given access to an $s$-sparse matrix $H\in\RR^{n\times n}$ stored in QRAM, and a time $t\in\RR$. 

The results of Ref.~\cite{vanApeldoorn2020quantumsdpsolvers} are stated explicitly based on the Hamiltonian simulation algorithm of Ref.~\cite{Berry_2015}. 
However, it is a straight-forward (if lengthy) exercise to swap in an improved method.
For our analysis, we have done just that, using the 
routine from Ref.~\cite{LowBlockEncoding}, which has complexity $\O(t\sqrt{s}\|H\|_{\ell_1\rightarrow\ell_2})^{1+o(1)}/\epsilon^{o(1)})$. 
Because  $\|H\|_{\ell_1\rightarrow\ell_2}\leq\|H\|$, this allows to achieve a similar scaling in the sparsity as the original quantum HU routine from Ref.~\cite{Brandao2022fasterquantum}, while simultaneously having an improved dependence on the precision $\epsilon$.  
The Hamiltonian Simulation subroutine enters the analysis of 
Ref.~\cite{vanApeldoorn2020quantumsdpsolvers} in their Lemma~36.
Plugging in the version of Ref.~\cite{LowBlockEncoding} and retracing the rest of the argument then gives the following performance for preparation of a Gibbs state:

\begin{lemma}[{\cite[Lem. 44]{vanApeldoorn2020quantumsdpsolvers}}] \label{lem:gibbs_sim}
    Given QRAM access to an $s$-sparse matrix $H\in\RR^{n\times n}$ satisfying $\Id \prec H$ and $2\Id\nprec H$, we can probabilistically prepare a purified Gibbs state $|\tilde\rho\rangle_{AB}$,
    s.t.\ with high probability ${\trnorm{\Tr_B(|\tilde\rho\rangle\langle\tilde\rho|_{AB}) - e^{-H}/\tr(e^{-H})} \leq \frac\epsilon8}$ holds, using
    \begin{align*}
            \tilde \O((\norm{H}\sqrt{s})^{1+o(1)}\sqrt{n})
    \end{align*}
    queries and gates.
\end{lemma}

Lem.~\ref{lem:gibbs_sim} requires $H$ to satisfy satisfies $\Id \prec H$ and $2\Id\nprec H$. We achieve this by employing a trick from Ref.~\cite[Cor.~14]{vanApeldoorn2020quantumsdpsolvers}, where we compute an estimate $\tilde\lambda_{\mathrm{min}}$ of the minimum eigenvalue of $H$ with an additive error $1/2$. Using this estimate, we apply Lem.~\ref{lem:gibbs_sim} to a shifted Hamiltonian 
\begin{align}\label{eq:H_+}
    H_+:=H-(\tilde\lambda_{\mathrm{min}} - 3/2)\Id   
\end{align}
that fulfills the requirements. Shifting a Hamiltonian by a multiple of the identity does not change the corresponding Gibbs state. Computing $\tilde\lambda_{\mathrm{min}}$ can be achieved using $\tilde \O(\norm{H} s \sqrt{n})$ queries and gates (c.f.\ Ref.~\cite[Lem.\ 50]{vanApeldoorn2020quantumsdpsolvers} with $\epsilon=1/2$ ).

As argued in Ref.~\cite[Sec.~3.4]{Brandao2022fasterquantum}, for Hamiltonian matrices $H$ that occur in the HU algorithm, one has the operator norm bound $\|H\|=\O(\log(n)\epsilon^{-1})$. 
Furthermore, they point out that, given $\O(n\epsilon^{-2})$ preparations of $\rho$ with precision $\epsilon/8$, one can acquire estimate $\tilde\rho_{ii}$ for the diagonal entries fulfilling $\sum_i |\tilde\rho_{ii} - \rho_{ii}| = \epsilon/4$. Thus, the cost per iteration of the HU algorithm is
\begin{align*}
    \tilde \O(n^{3/2}s^{1/2+o(1)}\epsilon^{-3+o(1)}),
\end{align*}
as claimed.

\section{Asymptotic performance of the improved Hamiltonian Updates} \label{sec:asymptotic_improvements}

In this section, we summarize the improvement in the asymptotic performance of the HU procedure that we have achieved compared to Ref.~\cite{Brandao2022fasterquantum}. 

The first table below gives estimates for the asymptotic complexity of solving the problem (\ref{eq:approx_feasibility_constraint}) as a function of $\epsilon$.
In this paper, we did not attempt to find improved rigorous asymptotic estimates for this scaling behavior.
However, in Sec.~\ref{sec:class_imp}, we have described several practical ways to speed up convergence.
The numerical results presented in Sec.~\ref{sec:numerics_improvements} show that these lead to speedups by very large constant factors, but the numerically observed asymptotic scaling matches the original exponent within the margin of error.

\begin{table}[H]
	\centering \scriptsize
	\begin{tabular}{|p{1.5cm}||p{2.3cm}|p{2cm}|p{1.8cm}|p{2.1cm}|}
	\hline \rule{0pt}{2.5ex}
		  & Previous theor.\ bounds & New theor.\ bounds  & Numerical scaling  & 95\% CI for exponent \\ \hline 
		\raggedright Number of \mbox{iterations}
        & \multirow{2}{*}{$\tilde \O(\epsilon^{-2})$}  
        & \multirow{2}{*}{(no new results)}   
        & \multirow{2}{*}{$\tilde \O(\epsilon^{-2.02})$}  & 
		\multirow{2}{*}{$[-2.39,  -1.74]$}
		\\ \hline 
		\end{tabular}
\end{table}

The next table summarizes the scaling of the precision $\nu$ of the rounded solution (compared to an ideal SDP solution, see Eq.~(\ref{eq:nu})) as a function of $\epsilon$, as detailed in Sec.~\ref{sec:rounding}.

\begin{table}[H]
    \centering \scriptsize
    \begin{tabular}{|p{1.5cm}||p{2.3cm}|p{2cm}|p{1.8cm}|p{2.1cm}|}
        \hline \rule{0pt}{2.5ex}
        & Previous theor.\ bounds & New theor.\ bounds & Numerical scaling & 95\% CI for exponent\\ \hline 
        \raggedright Precision after rounding 
        & \multirow{2}{*}{$\tilde \O(\epsilon^{0.25})$} 
        & \multirow{2}{*}{$\tilde \O(\epsilon^{0.33})$} 
        & \multirow{2}{*}{$\tilde \O(\epsilon^{0.79})$} 
        & \multirow{2}{*}{$[0.71, 0.88]$} \\
          \hline
    \end{tabular}
\end{table}

The improved scaling of the total running time of each iteration as a function of $\epsilon$, as detailed in Sec.~\ref{sec:quant_imp} is as follows

\begin{table}[H]
    \centering \scriptsize
    \begin{tabular}{|p{2cm}||p{3.2cm}|p{3.2cm}|}
	\hline \rule{0pt}{2.5ex}
		& Previous theor.\ bounds        & New theor.\ bounds \\ \hline
        Time per iteration:   &                              &         \\ 
				\quad classical         &$\tilde\O(\min\{n^3, n^2 s \epsilon^{-1}\})$                      & (no new results)                            \\
        \quad quantum           & $\tilde \O(n^{1.5} s^{0.5+o(1)} \epsilon^{-5+o(1)})$ & $\tilde \O(n^{1.5} s^{0.5+o(1)} \epsilon^{-3+o(1)})$  \\ \hline 
    \end{tabular}
\end{table}

The previous tables can now be combined, to estimate the scaling of the total running time required to achieve a given precision $\nu$ of the solution after rounding.
For the final column, we have combined the analytic estimates on the running time from the previous table,
with the numerically found scaling of the precision $\nu$ from the table above.

\begin{table}[H]
    \centering \scriptsize
    \begin{tabular}{|p{0.99cm}||p{3.10cm}|p{3.10cm}|p{3.50cm}|}
	\hline \rule{0pt}{2.5ex}
		& Previous theor.\ bounds        & New theor.\ bounds          & Combined theor.\ / num.\ scaling
		\\ \hline 
        Total time            &                               &                   &                \\
        \,classical         & $\tilde \O(\min\{n^3 \nu^{-8}, n^2s\nu^{-12}\})$      & $\tilde \O(\min\{n^3 \nu^{-6}, n^2s\nu^{-9}\})$         & $\tilde \O(\min\{n^3 \nu^{-2.56}, n^2s \nu^{-3.82}\})$ \\ 
        \,quantum           & $\tilde \O(n^{1.5}s^{0.5+o(1)}\nu^{-28 +o(1)})$   & $\tilde \O(n^{1.5}s^{0.5+o(1)}\nu^{-15+o(1)})$ & $\tilde \O(n^{1.5}s^{0.5+o(1)}\nu^{-6.35})$ \\ \hline
    \end{tabular}
\end{table}

Theorem~1 of \cite{Brandao2022fasterquantum} (reproduced in our introduction)
stated their asymptotic performance in terms of the precision $\mu$ of the SDP solution,
not in terms of the practically more relevant precision $\nu$ of after rounding.
However, as we found in Sec.~\ref{sec:rounding_section}, the two quantities display the same scaling behavior as a function of $\epsilon$.
Hence, conversely, the running time scaling as a function of $\nu$ matches the running time scaling as a function of $\mu$.
Therefore, the results of the previous table are directly comparable with Theorem~1 of \cite{Brandao2022fasterquantum}.

\section{Non-asymptotic benchmarking of quantum implementations} \label{sec:benchmarking}
In this section, we compare the quantum implementation of HU to its classical counterpart by estimating the required number of quantum gates and tracking the classical computation time. 
We then extrapolate these results to larger problem instances to explore for which problem sizes one can expect a future quantum computer to beat a classical one, under optimistic assumptions.
This approach mirrors previous studies for the non-asymptotic behavior of quantum algorithms that are too large to be simulated in the gate model; see, e.g.~\cite{Cade_2023, grover_rudolph, ammann2023realisticruntimeanalysisquantum, qubrabench_repository}.
We assume that the reader is familiar with standard methods in quantum computing. For further background consult e.g.\ Ref.~\cite{Nielsen_Chuang_2010}. 

\subsection{Gate counting}

To compare the classical and quantum running times, we generate random problem instances and solve them using the classical version of the HU algorithm. We estimate the number of gates needed for the quantum subroutines and compare this with the classical running time. From this comparison, we determine the maximum allowable gate time for each instance that would enable a quantum computer to outperform its classical counterpart.
For our estimates of the quantum gate count, we consistently make assumptions that are favorable to the quantum computer.
This approach will be justified in retrospect:
It will turn out that even these optimistic estimates put the threshold for a quantum advantage far beyond the capabilities of realistic hardware.

As discussed in Sec.~\ref{sec:quant_imp}, the running time of the quantum part is dominated by the task of estimating the diagonal elements $\rho_{ii}$ of the Gibbs state.
We will only estimate the cost of this part, and neglect all other quantum sub-routines.
(Recall that this is justified, as we aim for an \emph{optimistic} assessment of the quantum complexity).
In the benchmarks, the precision
of the Gibbs state simulations is set to $\epsilon/8$, in line with Lem.~\ref{lem:diag_estimation} and Ref.~\cite[Lem.\ 3.3]{Brandao2022fasterquantum}. 
For simplicity, and again being generous to the quantum approach, we nevertheless assume that the subroutines return exact values. 
Additionally, we compute the free energy directly using Eq.~\eqref{eq:free_energy} instead of bounding it via Eq.~\eqref{eq:F_approx}.
Finally, we will count solely  logical two-qubit gates, neglecting single-qubit gates and error-correction overheads.

As detailed in Appendix~\ref{sec:quantumCircuits},
in this framework, we find the following result for the complexity of preparing Gibbs states.

\begin{restatable}{estimate}{gateCost} 
    \label{thm:gate_cost}
    Let $H\in\RR^{n\times n}$ be an $s$-sparse matrix and $b$ be the number of bits per entry used to store $H$.
    The expected number of two-qubit gates used to prepare an approximation of the Gibbs state $\rho_{H}$ with precision $\epsilon$ using the constructions in Refs.~\cite{LowBlockEncoding, cuccaro2004Adder, Low2019Qubitization, vanApeldoorn2020quantumsdpsolvers, BoyerSearching} is at least
    \begin{align} \label{eq:gate_cost_sample}
         (32b+32\log_2(n)-18)\,(4.5 \ln(7.8\epsilon^{-1})\, n^{1/2} s\norm{H_+}_\mathrm{max} -1),
    \end{align}
    with $H_+$ constructed from $H$ as defined in Eq.~\eqref{eq:H_+}.
\end{restatable}

The estimates for the diagonal entries of $\rho$ are obtained statistically by measuring multiple prepared Gibbs states in the computational basis. The number of Gibbs state samples required can be bounded as follows:
\begin{restatable}{estimate}{numberSamples} 
    \label{est:number_samples}
    Let $\epsilon\leq 1/4$. Let $\rho\in\RR^{n\times n}$ be a quantum state that can be prepared with precision $\epsilon/8$ in trace distance on a quantum computer. The expected number of preparations needed to compute estimates $\hat\rho_{ii}$ for the diagonal entries of $\rho$, s.t.\ $\sum_i|\rho_{ii}-\hat\rho_{ii}|\leq\frac\epsilon 4$ using the construction in Appendix~\ref{sec:quantumCircuits}, is at least
    \begin{align}  \label{eq:sample_number}
        128\ln(2)\epsilon^{-2}n.
    \end{align}
\end{restatable}

A detailed account for how to arrive at Estimates~\ref{thm:gate_cost} and \ref{est:number_samples} is provided in Appendix~\ref{sec:quantumCircuits}. 

We compute the total number of two-qubit gates required for a diagonal update step in the HU routine by multiplying the gate cost per Gibbs state sample \eqref{eq:gate_cost_sample} with the expected number of samples needed to estimate the diagonal entries of $\rho$ \eqref{eq:sample_number}.
 
We avoid directly computing $\norm{H_+}_\mathrm{max} = \|H-(\tilde\lambda_\mathrm{min} - 3/2)\Id\|_\mathrm{max}$ in the benchmark, as determining the minimum eigenvalue classically for each iteration is computationally expensive. Instead, we use a conservative estimate $\norm{H_+}_\mathrm{max}\gtrapprox \norm{H}_\mathrm{max}$, based on our numerical observation $\norm{H_+}_\mathrm{max}\approx 2\norm{H}_\mathrm{max}$.
Additionally, we assume that the matrix $H$ is represented on the quantum architecture with just $b=8$ logical bits per element.

We generate random cost matrices with a block form as defined in \eqref{eq:C_block} with normally distributed non-zero entries in each block and a sparsity of $s=16$.
A total of 256 instances are created, with dimensions uniformly distributed between $512 \leq n \leq 4096$. 
We then solve the instances using the HU algorithm with $\epsilon=0.01$ and a binary search with a final maximum gap of $0.01$.
The matrix exponentials -- the most computationally expensive part of the process -- are executed on an Nvidia GeForce RTX 4090 GPU card, while an Intel i7-13700 CPU handles the other operations in the routine.

\subsection{Results}

Here, we compare the quantum and the classical running times for the benchmarked non-asymptotic instances.

In a first step, we ignore the possibility of parallelizing the quantum implementation.
With this in mind, we divide the time the classical implementation took by the lower bound on the number of quantum gates given in the last section.
A necessary condition for a quantum architecture to be preferable to classical consumer hardware would be that
its two-qubit gates run in time at most equal to that quotient.

The results are displayed in the left panel of Fig.~\ref{fig:benchmark}.
We find that even for large instances up to dimension $n=4096$ a quantum computer requires a two-qubit gate time of less than $10^{-19}$s to be able to break even. 
This is more than ten orders of magnitude away from the current single-qubit gate speed record of $6.5\times 10^{-9}$s \cite{Chew2022}.

There are two parts of the quantum algorithm where parallelization can provide an advantage:
\begin{itemize}
    \item 
			\emph{Parallel execution of quantum gate operations}.
			Gates acting on different qubits can be applied simultaneously rather than sequentially.
			The smallest number of layers required to execute a circuit taking this possibility into account is called its \emph{depth}.  
			Unfortunately, in our case, it is not known how to reduce the circuit depth significantly below the gate count.
			Indeed, the main building block of the circuit is an adder, for which
			the state-of-the-art two-qubit gate depth is only a factor of two lower than the total two-qubit gate count \cite{cuccaro2004Adder}.
    \item 
			\emph{Parallel sampling of Gibbs states}.
			This step offers significant parallelization potential, as the required number of samples is more than $128\ln(n)\epsilon^{-2}n$.
			From a computer science perspective, this reduction in algorithmic depth might be of theoretical interest.
			However, each simultaneous computation would require its own dedicated quantum computer, limiting the practical relevance.
\end{itemize}

\subsubsection{Extrapolation}

Due to the superior asymptotic scaling of the quantum HU algorithm compared to its classical counterpart, larger problem instances make it easier for the quantum approach to break even. 
We now investigate whether this favorable scaling is sufficient to achieve a quantum advantage for instances that remain computationally feasible at all.
For this we use two different approaches: 
(1) 
Extrapolate using the exponents that have been derived theoretically in Sec.~\ref{sec:asymptotic_improvements}, and let the benchmark data determine just the prefactor. 
(2) Fit a powerlaw function of the form $f(t)=a_1t^{a_2}$ with free parameters $a_1,a_2$ to the data. 

Both approaches are displayed in Fig.~\ref{fig:benchmark}.
We have extrapolated the quantum gate counts for instance sizes that would take more than 100 years to be solved with the given classical hardware. 
Still, even at this scale, the required two-qubit gate time is more than a factor of $10^7$ less than current quantum gate times. 
For a powerlaw extrapolation (2) the gap is even larger with a factor of more than $10^8$.

\begin{figure}[H]
    \centering
    \begin{subfigure}[b]{0.99\textwidth}
        \includegraphics[width=\textwidth]{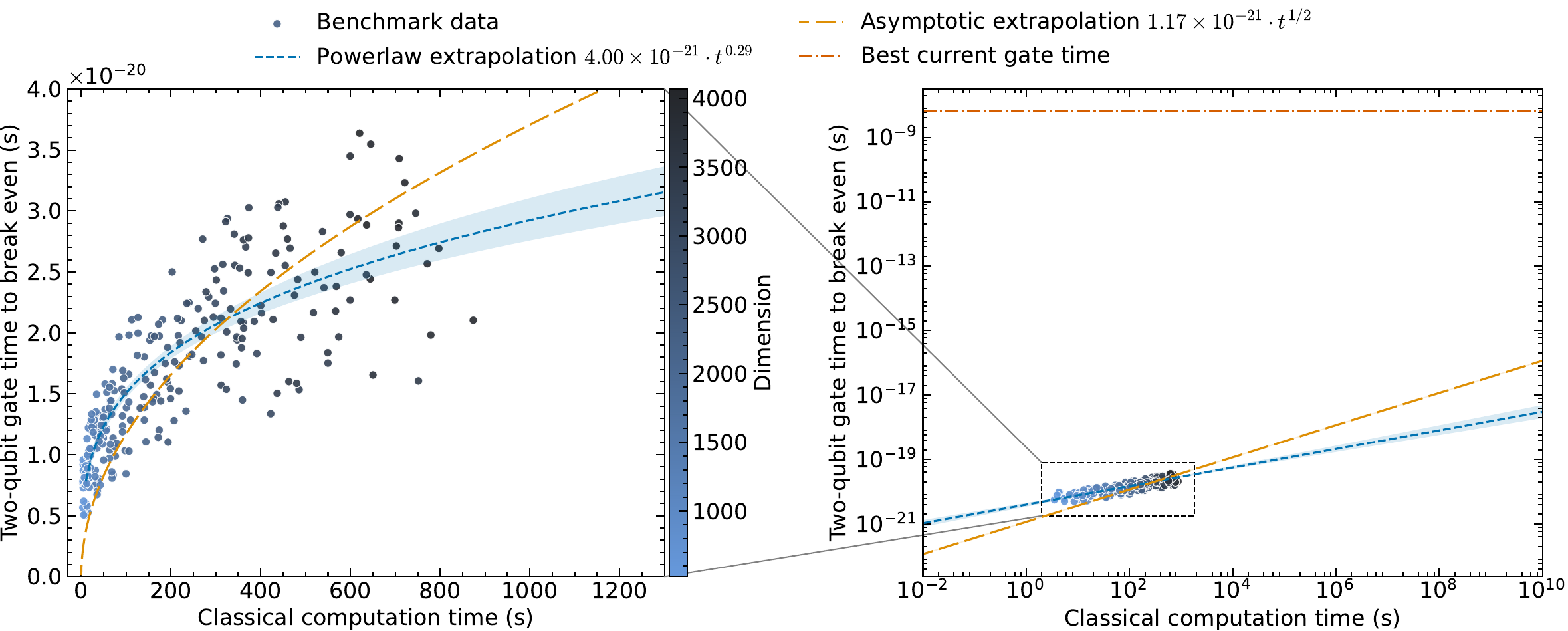}
    \end{subfigure}
    \caption{Allowed maximum two-qubit gate time for a quantum computer to break even to a classical simulation. The benchmark data is extrapolated based on the known asymptotic scaling behaviors (orange) and a power law fit (blue). The constant line (red) shows the best current classical single-qubit gate time of $6.5$ns. The colored areas show the 95\% confidence interval of the fits.
    We see that with this gate speed, even for running times of more than 100 years, the quantum implementation would be more than $10^7$ times slower than the classical implementation.
    }
    \label{fig:benchmark}
\end{figure}

\section{Conclusion}

In this paper, we have investigated whether the theoretically proven asymptotic speed-ups provided by quantum SDP solvers can be used to achieve a practical advantage for solving convex relaxations of combinatorial optimization problems.
A priori, combinatorial problems seem to be a good fit for quantum SDP methods, because their main drawback -- the unfavorable scaling in the precision -- is less important for applications where the solution will be subjected to a rounding procedure.

Unfortunately, our work has found no indication that advantages manifest in regimes that are remotely realistic.
This finding holds in spite of the fact that
we have spent significant efforts to improve the Hamiltonian Updates algorithm before benchmarking it,
and
that we have made a large number of approximations and assumptions in favor of a quantum architecture
(in particular, ignoring all quantum error correction overhead, and only estimating the cost of a subset of the routines a quantum computer would have to run).

It therefore seems that, to the best of our current knowledge, the proposed quantum SDP methods for combinatorial optimization may constitute a \emph{galactic algorithm} -- advantageous in theory, but such that their benefit requires instance sizes that are far beyond what can be realistically solved.

We note that the present work does not rigorously prove the absence of a realistic quantum advantage.
Indeed, any such result would require rigorous and non-asymptotic lower bounds on the classical complexity of the problem.
But unconditional complexity lower bounds are notoriously hard to obtain \cite{arora2009computational}, so that the results of thorough benchmarks of state-of-the-art implementations seem to constitute the best evidence that is realistically achievable.

\bigbreak
\emph{Data availability} - The data that supports the findings of this article is openly available \cite{code_repository}.

\subsection{Acknowledgements}

We thank Brandon Augustino, Johannes Berg, Lionel Dmello, and Sebastian Stiller for insightful discussions.

This work was supported by the 
Federal Ministry for Economic Affairs and Climate Action (BMWK), project ProvideQ,
and the
German Federal Ministry of Education and Research (BMBF), project QuBRA. 
FH and DG are also supported by Germany's
Excellence Strategy -- Cluster of Excellence Matter and Light for Quantum Computing (ML4Q) EXC 2004/1 (390534769).

\printbibliography

@article{Brandao2022fasterquantum,
  doi = {10.22331/q-2022-01-20-625},
  url = {https://doi.org/10.22331/q-2022-01-20-625},
  title = {Faster quantum and classical {SDP} approximations for quadratic binary optimization},
  author = {G.S L. Brand{\~{a}}o, Fernando and Kueng, Richard and Stilck Fran{\c{c}}a, Daniel},
  journal = {{Quantum}},
  issn = {2521-327X},
  publisher = {{Verein zur F{\"{o}}rderung des Open Access Publizierens in den Quantenwissenschaften}},
  volume = {6},
  pages = {625},
  month = jan,
  year = {2022}
}

@INPROCEEDINGS{BrandaoInequality,
  author={Brandao, Fernando G.S.L. and Svore, Krysta M.},
  booktitle={2017 IEEE 58th Annual Symposium on Foundations of Computer Science (FOCS)}, 
  title={Quantum Speed-Ups for Solving Semidefinite Programs}, 
  year={2017},
  volume={},
  number={},
  pages={415-426},
  doi={10.1109/FOCS.2017.45}}

@misc{brandao2019quantumsdpsolverslarge,
      title={Quantum SDP Solvers: Large Speed-ups, Optimality, and Applications to Quantum Learning}, 
      author={Fernando G. S. L. Brandao and Amir Kalev and Tongyang Li and Cedric Yen-Yu Lin and Krysta M. Svore and Xiaodi Wu},
      year={2019},
      eprint={1710.02581},
      archivePrefix={arXiv},
      primaryClass={quant-ph},
      url={https://arxiv.org/abs/1710.02581}, 
}

@InProceedings{vanapeldoorn_improvements,
  author =	{van Apeldoorn, Joran and Gily\'{e}n, Andr\'{a}s},
  title =	{{Improvements in Quantum SDP-Solving with Applications}},
  booktitle =	{46th International Colloquium on Automata, Languages, and Programming (ICALP 2019)},
  pages =	{99:1--99:15},
  ISBN =	{978-3-95977-109-2},
  ISSN =	{1868-8969},
  year =	{2019},
  doi =		{10.4230/LIPIcs.ICALP.2019.99},
}

@article{vanApeldoorn2020quantumsdpsolvers,
  doi = {10.22331/q-2020-02-14-230},
  url = {https://doi.org/10.22331/q-2020-02-14-230},
  title = {Quantum {SDP}-{S}olvers: {B}etter upper and lower bounds},
  author = {van Apeldoorn, Joran and Gily{\'{e}}n, Andr{\'{a}}s and Gribling, Sander and de Wolf, Ronald},
  journal = {{Quantum}},
  issn = {2521-327X},
  publisher = {{Verein zur F{\"{o}}rderung des Open Access Publizierens in den Quantenwissenschaften}},
  volume = {4},
  pages = {230},
  month = feb,
  year = {2020}
}

@article{PoulinGibbs2009,
  title = {Sampling from the Thermal Quantum Gibbs State and Evaluating Partition Functions with a Quantum Computer},
  author = {Poulin, David and Wocjan, Pawel},
  journal = {Phys. Rev. Lett.},
  volume = {103},
  issue = {22},
  pages = {220502},
  numpages = {4},
  year = {2009},
  month = {Nov},
  publisher = {American Physical Society},
  doi = {10.1103/PhysRevLett.103.220502},
  url = {https://link.aps.org/doi/10.1103/PhysRevLett.103.220502}
}

@inproceedings{Berry_2015,
   title={Hamiltonian Simulation with Nearly Optimal Dependence on all Parameters},
   url={http://dx.doi.org/10.1109/FOCS.2015.54},
   DOI={10.1109/focs.2015.54},
   booktitle={2015 IEEE 56th Annual Symposium on Foundations of Computer Science},
   author={Berry, Dominic W. and Childs, Andrew M. and Kothari, Robin},
   month=oct, pages={792–809} }

@misc{augustino2023solvingsemidefiniterelaxationqubos,
      title={Solving the semidefinite relaxation of QUBOs in matrix multiplication time, and faster with a quantum computer}, 
      author={Brandon Augustino and Giacomo Nannicini and Tamás Terlaky and Luis Zuluaga},
      year={2023},
      eprint={2301.04237},
      archivePrefix={arXiv},
      primaryClass={quant-ph},
      url={https://arxiv.org/abs/2301.04237}, 
}

@INPROCEEDINGS{MMW,
  author={Arora, S. and Hazan, E. and Kale, S.},
  booktitle={46th Annual IEEE Symposium on Foundations of Computer Science (FOCS'05)}, 
  title={Fast algorithms for approximate semidefinite programming using the multiplicative weights update method}, 
  year={2005},
  volume={},
  number={},
  pages={339-348},
  doi={10.1109/SFCS.2005.35}}

@INPROCEEDINGS{InteriorPoint,
  author={Lee, Yin Tat and Sidford, Aaron and Wong, Sam Chiu-Wai},
  booktitle={2015 IEEE 56th Annual Symposium on Foundations of Computer Science}, 
  title={A Faster Cutting Plane Method and its Implications for Combinatorial and Convex Optimization}, 
  year={2015},
  volume={},
  number={},
  pages={1049-1065},
  doi={10.1109/FOCS.2015.68}}

@inproceedings{RaghavendraW17,
  author    = {Prasad Raghavendra and
               Benjamin Weitz},
  title     = {On the Bit Complexity of Sum-of-Squares Proofs},
  booktitle = {44th International Colloquium on Automata, Languages, and Programming},
  pages     = {80:1--80:13},
  year      = {2017},
}

@inproceedings{ODonnell17,
  author    = {Ryan O'Donnell},
  title     = {{SOS} Is Not Obviously Automatizable, Even Approximately},
  booktitle = {8th Innovations in Theoretical Computer Science Conference},
  pages     = {59:1--59:10},
  year      = {2017},
}

@inproceedings{Grover96,
  author       = {Lov K. Grover},
  title        = {A Fast Quantum Mechanical Algorithm for Database Search},
  booktitle    = {Proceedings of the Twenty-Eighth Annual {ACM} Symposium on the Theory
                  of Computing},
  pages        = {212--219},
  publisher    = {{ACM}},
  year         = {1996},
  url          = {https://doi.org/10.1145/237814.237866},
  doi          = {10.1145/237814.237866},
  bibsource    = {dblp computer science bibliography, https://dblp.org}
}

@inproceedings{goemansWilliamson,
author = {Goemans, Michel X. and Williamson, David P.},
title = {.879-approximation algorithms for MAX CUT and MAX 2SAT},
year = {1994},
isbn = {0897916638},
booktitle = {Proceedings of the Twenty-Sixth Annual ACM Symposium on Theory of Computing},
pages = {422–431},
doi = {10.1145/195058.195216},
}

@Article{RoundingGrothendieck,
  author  = {Alon, Noga and Naor, Assaf},
  journal = {SIAM Journal on Computing},
  title   = {Approximating the Cut-Norm via {G}rothendieck's Inequality},
  year    = {2006},
  number  = {4},
  pages   = {787-803},
  volume  = {35},
  doi     = {10.1137/S0097539704441629},
  url     = {https://doi.org/10.1137/S0097539704441629},
}

@misc{friedland2020symmetricgrothendieckinequality,
      title={Symmetric Grothendieck inequality}, 
      author={Shmuel Friedland and Lek-Heng Lim},
      year={2020},
      eprint={2003.07345},
      archivePrefix={arXiv},
      primaryClass={math.FA},
      url={https://arxiv.org/abs/2003.07345}, 
}

@article{Briet_2014, volume={10},
    title={Grothendieck inequalities for semidefinite programs with rank constraint},
    ISSN={1557-2862},
    url={http://dx.doi.org/10.4086/toc.2014.v010a004},
    DOI={10.4086/toc.2014.v010a004},
    number={1},
    journal={Theory of Computing},
    publisher={Theory of Computing Exchange},
    author={Briet, Jop and de Oliveira Filho, Fernando Mario and Vallentin, Frank},
    year={2014},
    pages={77–105} 
}

@book{
    Nielsen_Chuang_2010, 
    place={Cambridge}, 
    title={Quantum Computation and Quantum Information: 10th Anniversary Edition}, 
    publisher={Cambridge University Press}, 
    author={Nielsen, Michael A. and Chuang, Isaac L.}, 
    year={2010}
}

@misc{dalzell2023quantumalgorithmssurveyapplications,
      title={Quantum algorithms: A survey of applications and end-to-end complexities}, 
      author={Alexander M. Dalzell and Sam McArdle and Mario Berta and Przemyslaw Bienias and Chi-Fang Chen and András Gilyén and Connor T. Hann and Michael J. Kastoryano and Emil T. Khabiboulline and Aleksander Kubica and Grant Salton and Samson Wang and Fernando G. S. L. Brandao},
      year={2023},
      eprint={2310.03011},
      archivePrefix={arXiv},
      primaryClass={quant-ph},
      url={https://arxiv.org/abs/2310.03011}, 
}

@article{Abbas_2024,
   title={Challenges and opportunities in quantum optimization},
   ISSN={2522-5820},
   url={http://dx.doi.org/10.1038/s42254-024-00770-9},
   DOI={10.1038/s42254-024-00770-9},
   journal={Nature Reviews Physics},
   publisher={Springer Science and Business Media LLC},
   author={Abbas, Amira and Ambainis, Andris and Augustino, Brandon and Bärtschi, Andreas and Buhrman, Harry and Coffrin, Carleton and Cortiana, Giorgio and Dunjko, Vedran and Egger, Daniel J. and Elmegreen, Bruce G. and Franco, Nicola and Fratini, Filippo and Fuller, Bryce and Gacon, Julien and Gonciulea, Constantin and Gribling, Sander and Gupta, Swati and Hadfield, Stuart and Heese, Raoul and Kircher, Gerhard and Kleinert, Thomas and Koch, Thorsten and Korpas, Georgios and Lenk, Steve and Marecek, Jakub and Markov, Vanio and Mazzola, Guglielmo and Mensa, Stefano and Mohseni, Naeimeh and Nannicini, Giacomo and O’Meara, Corey and Tapia, Elena Peña and Pokutta, Sebastian and Proissl, Manuel and Rebentrost, Patrick and Sahin, Emre and Symons, Benjamin C. B. and Tornow, Sabine and Valls, Víctor and Woerner, Stefan and Wolf-Bauwens, Mira L. and Yard, Jon and Yarkoni, Sheir and Zechiel, Dirk and Zhuk, Sergiy and Zoufal, Christa},
   year={2024},
   month=oct }

@article{Dalzell_2023,
   title={End-To-End Resource Analysis for Quantum Interior-Point Methods and Portfolio Optimization},
   volume={4},
   ISSN={2691-3399},
   url={http://dx.doi.org/10.1103/PRXQuantum.4.040325},
   DOI={10.1103/prxquantum.4.040325},
   number={4},
   journal={PRX Quantum},
   publisher={American Physical Society (APS)},
   author={Dalzell, Alexander M. and Clader, B. David and Salton, Grant and Berta, Mario and Lin, Cedric Yen-Yu and Bader, David A. and Stamatopoulos, Nikitas and Schuetz, Martin J. A. and Brandão, Fernando G. S. L. and Katzgraber, Helmut G. and Zeng, William J.},
   year={2023},
   month=nov }

@misc{ammann2023realisticruntimeanalysisquantum,
      title={Realistic Runtime Analysis for Quantum Simplex Computation}, 
			journal={arXiv:2311.09995},
      author={Sabrina Ammann and Maximilian Hess and Debora Ramacciotti and Sándor P. Fekete and Paulina L. A. Goedicke and David Gross and Andreea Lefterovici and Tobias J. Osborne and Michael Perk and Antonio Rotundo and S. E. Skelton and Sebastian Stiller and Timo de Wolff},
      year={2023},
      eprint={2311.09995},
      url={https://arxiv.org/abs/2311.09995}, 
}

@article{Cade_2023,
   title={Quantifying Grover speed-ups beyond asymptotic analysis},
   volume={7},
   ISSN={2521-327X},
   url={http://dx.doi.org/10.22331/q-2023-10-10-1133},
   DOI={10.22331/q-2023-10-10-1133},
   journal={Quantum},
   publisher={Verein zur Forderung des Open Access Publizierens in den Quantenwissenschaften},
   author={Cade, Chris and Folkertsma, Marten and Niesen, Ido and Weggemans, Jordi},
   year={2023},
   month=oct, pages={1133} }

@article{grover_rudolph,
  title = {Simple quantum algorithm to efficiently prepare sparse states},
  author = {Ramacciotti, Debora and Lefterovici, Andreea I. and Rotundo, Antonio F.},
  journal = {Phys. Rev. A},
  volume = {110},
  issue = {3},
  pages = {032609},
  numpages = {10},
  year = {2024},
  month = {Sep},
  publisher = {American Physical Society},
  doi = {10.1103/PhysRevA.110.032609},
  url = {https://link.aps.org/doi/10.1103/PhysRevA.110.032609}
}

@misc{qubrabench_repository,
  title        = {{Q}u{BRA} Quantum Benchmarking Project },
  year         = {2024},
  howpublished = {\url{https://github.com/qubrabench}},
}

@inproceedings{LowBlockEncoding,
author = {Low, Guang Hao},
title = {Hamiltonian simulation with nearly optimal dependence on spectral norm},
year = {2019},
isbn = {9781450367059},
publisher = {Association for Computing Machinery},
address = {New York, NY, USA},
url = {https://doi.org/10.1145/3313276.3316386},
doi = {10.1145/3313276.3316386},
booktitle = {Proceedings of the 51st Annual ACM SIGACT Symposium on Theory of Computing},
pages = {491–502},
numpages = {12},
location = {Phoenix, AZ, USA},
series = {STOC 2019}
}

@article{Low2019Qubitization,
  doi = {10.22331/q-2019-07-12-163},
  url = {https://doi.org/10.22331/q-2019-07-12-163},
  title = {Hamiltonian {S}imulation by {Q}ubitization},
  author = {Low, Guang Hao and Chuang, Isaac L.},
  journal = {{Quantum}},
  issn = {2521-327X},
  publisher = {{Verein zur F{\"{o}}rderung des Open Access Publizierens in den Quantenwissenschaften}},
  volume = {3},
  pages = {163},
  month = jul,
  year = {2019}
}

@article{BoyerSearching,
  title={Tight bounds on quantum searching},
  author={Boyer, Michel and Brassard, Gilles and H{\o}yer, Peter and Tapp, Alain},
  journal={Fortschritte der Physik: Progress of Physics},
  volume={46},
  number={4-5},
  pages={493--505},
  year={1998},
}

@article{ToffoliCost,
author = {Shende, Vivek V. and Markov, Igor L.},
title = {On the CNOT-cost of TOFFOLI gates},
year = {2009},
issue_date = {May 2009},
publisher = {Rinton Press, Incorporated},
address = {Paramus, NJ},
volume = {9},
number = {5},
issn = {1533-7146},
journal = {Quantum Info. Comput.},
month = {may},
pages = {461–486},
numpages = {26}
}

@Article{Chew2022,
  author    = {Chew, Y. and Tomita, T. and Mahesh, T. P. and Sugawa, S. and de Léséleuc, S. and Ohmori, K.},
  journal   = {Nature Photonics},
  title     = {Ultrafast energy exchange between two single Rydberg atoms on a nanosecond timescale},
  year      = {2022},
  issn      = {1749-4893},
  month     = aug,
  number    = {10},
  pages     = {724--729},
  volume    = {16},
  doi       = {10.1038/s41566-022-01047-2},
  publisher = {Springer Science and Business Media LLC},
}

@misc{cuccaro2004Adder,
      title={A new quantum ripple-carry addition circuit}, 
      author={Steven A. Cuccaro and Thomas G. Draper and Samuel A. Kutin and David Petrie Moulton},
      year={2004},
      eprint={quant-ph/0410184},
      archivePrefix={arXiv},
      primaryClass={quant-ph}
}

@Misc{AuKing2003,
  author = {King, Christopher},
  title  = {Inequalities for Trace Norms of 2 × 2 Block Matrices},
  journal= {Communications in Mathematical Physics},
  year   = {2003},
  doi    = {10.1007/s00220-003-0955-9},
}

@Misc{brandonPrivate,
	author = {Augustino, Brandon},
	title = {Private communication},
	year = {2024}
}

@article{entropyInequalities,
title = "Entropy inequalities",
author = "Huzihiro Araki and Lieb, {Elliott H.}",
year = "1970",
month = jun,
doi = "10.1007/BF01646092",
language = "English (US)",
volume = "18",
pages = "160--170",
journal = "Communications In Mathematical Physics",
issn = "0010-3616",
publisher = "Springer New York",
number = "2",
}

@book{feynman1998statistical,
  title={Statistical Mechanics: A Set Of Lectures},
  author={Feynman, R.P.},
  isbn={9780813346106},
  series={Advanced Books Classics},
  year={1998},
  publisher={Avalon Publishing}
}

@article{HHL09,
  title = {Quantum Algorithm for Linear Systems of Equations},
  author = {Harrow, Aram W. and Hassidim, Avinatan and Lloyd, Seth},
  journal = {Phys. Rev. Lett.},
  volume = {103},
  issue = {15},
  pages = {150502},
  numpages = {4},
  year = {2009},
  month = {Oct},
  publisher = {American Physical Society},
  doi = {10.1103/PhysRevLett.103.150502},
  url = {https://link.aps.org/doi/10.1103/PhysRevLett.103.150502}
}

@inproceedings{Shor94,
  author       = {Peter W. Shor},
  title        = {Algorithms for Quantum Computation: Discrete Logarithms and Factoring},
  booktitle    = {35th Annual Symposium on Foundations of Computer Science},
  pages        = {124--134},
  publisher    = {{IEEE} Computer Society},
  year         = {1994},
  url          = {https://doi.org/10.1109/SFCS.1994.365700},
  doi          = {10.1109/SFCS.1994.365700},
  bibsource    = {dblp computer science bibliography, https://dblp.org}
}

@book{Renes2022,
  title={Quantum Information Theory: Concepts and Methods},
  author={Renes, Joseph},
  year={2022},
  publisher={Walter de Gruyter GmbH \& Co KG}
}

@article{polyak19641,
title = {Some methods of speeding up the convergence of iteration methods},
journal = {USSR Computational Mathematics and Mathematical Physics},
volume = {4},
number = {5},
pages = {1-17},
year = {1964},
issn = {0041-5553},
doi = {https://doi.org/10.1016/0041-5553(64)90137-5},
url = {https://www.sciencedirect.com/science/article/pii/0041555364901375},
author = {Boris T. Polyak}
}

@book{BD04,
  title={Convex optimization},
  author={Boyd, Stephen and Vandenberghe, Lieven},
  year={2004},
  publisher={Cambridge university press}
}

@book{Bar02,
  title={A course in convexity},
  author={Barvinok, Alexander},
  volume={54},
  year={2002},
  publisher={American Mathematical Soc.}
}

@article{ODonoghue2016,
  author    = {Brendan O'Donoghue and Eric Chu and Neal Parikh and Stephen Boyd},
  title     = {Conic Optimization via Operator Splitting and Homogeneous Self-Dual Embedding},
  journal   = {Journal of Optimization Theory and Applications},
  year      = {2016},
  volume    = {169},
  number    = {3},
  pages     = {1042--1068},
  doi       = {10.1007/s10957-016-0892-3},
  url       = {https://doi.org/10.1007/s10957-016-0892-3},
  issn      = {1573-2878}
}

@Inbook{Karp1972,
author="Karp, Richard M.",
title="Reducibility among Combinatorial Problems",
bookTitle="Complexity of Computer Computations",
year="1972",
pages="85--103",
isbn="978-1-4684-2001-2",
doi="10.1007/978-1-4684-2001-2_9",
url="https://doi.org/10.1007/978-1-4684-2001-2_9"
}

@book{arora2009computational,
  title={Computational Complexity: A Modern Approach},
  author={Arora, S. and Barak, B.},
  isbn={9780521424264},
  year={2009},
  publisher={Cambridge University Press}
}

@misc{code_repository,
author={Fabian Henze and Viet Tran and Birte Ostermann and Richard Kueng and Time de Wolff and David Gross},
doi={10.5281/zenodo.14871936},
year={2025},
publisher={Zenodo}
}

\newpage

\appendix

\section{Proofs}

\subsection{Free energy tracking} \label{sec:EntropyProof}

Lem.~\ref{thm:entropyBoundLinear} describes a well-known concept in statistical mechanics \cite[Ch.~2.11]{feynman1998statistical}. 
We still give  the calculation for completeness. 
\entropyBoundLinear*
\begin{proof}
We have
    \begin{align}\begin{split}
        \partial_{\lambda} \tr (e^{-H-\lambda \Delta H})  
		=&
		\sum_{k=0}^\infty 
		\frac1{k!} \tr \left(
			\partial_{\lambda} (-H-\lambda \Delta H)^k 
		\right) \\
		=&
		\sum_{k=1}^\infty 
		\frac{-1}{k!} \sum_{i=0}^{k-1} \tr
		\left(
			(-H-\lambda \Delta H)^i  \Delta H (-H -\lambda \Delta H)^{k-i-1} 
		\right) \\
		=&
		\sum_{k=1}^\infty 
		\frac{-1}{(k-1)!} \tr
		\left(
			(-H-\lambda \Delta H)^{k-1} \Delta H
		\right)
		\\
		=&
		- \tr
		\left(
			e^{-H-\lambda \Delta H}  \Delta H  
		\right),
    \end{split}\end{align}
    where the key step was to use the cyclicity of the trace to combine the terms $(-H-\lambda \Delta H)$. Then
    \begin{align}\begin{split}
        \partial_{\lambda} F(H+\lambda \Delta H) 
		=&
        -\partial_{\lambda} \ln \tr (e^{-H-\lambda \Delta H}) \\
        =& 
        - \frac{\partial_{\lambda} \tr (e^{-H-\lambda \Delta H})}{\tr (e^{-H-\lambda \Delta H})} \\
        =&
         \frac{\tr(e^{-H-\lambda \Delta H}  \Delta H) }{\tr (e^{-H-\lambda \Delta H})} \\
        =& 
        \tr(\rho_{H+\lambda\Delta H} \lambda \Delta H).
    \end{split}\end{align}
\end{proof}

\label{sec:EntropyProof2}
Next, we prove Thm.~\ref{thm:entropyBoundConvergence}.  In preparation for this, we need the following Lemmas \ref{lem:bound_R} and \ref{lem:a_priori_entropy}.
\begin{lemma}\label{lem:bound_R}
    Let $0\leq\beta<1$ be the momentum hyperparameter, $\Delta H\in\RR^{n\times n}$ be an update matrix in the HU routine and
 \begin{align}
        \lambda = \frac{(1-\beta)^2}{2}\tr(\rho_H\Delta H)
 \end{align}
 the step length of an update. Then, 
    \begin{align}
        \inf_{c\in\RR}\norm{\Delta H - c\Id} \leq \frac{1}{1-\beta}.   
    \end{align}
\end{lemma}
\begin{proof}
    Let $K$ be the current iteration of HU and $(\Delta H)^{(k)}$, $(P)^{(k)}$ the operators corresponding to the $k^\mathrm{th}$ iteration. The chosen step length $\lambda$ is independent of the update type (i.e.\ cost or diagonal update). Then, following from the definitions of the momentum in Sec.~\ref{sec:momentum}, the update term takes the form
    \begin{align}
        (\Delta H)^{(K)} 
        &=
        \sum_{k=0}^{K-1} \beta^{k} (P)^{(K-k)}, 
    \end{align}
    where, depending on the type of update in the respective iteration, $(P)^{(K-k)}$ takes the form of one of the following:
    \begin{align}\begin{split}
        P_c &= \gamma\Id - C \\
        (P_d^{\ell_1})^{(K-k)} &= \sgn\left((\mathrm{diag}\left(\rho^{(K-k)}\right)-\Id/n\right)\\
        (P_d^{\ell_2})^{(K-k)} &= \frac{\mathrm{diag}\left(\rho^{(K-k)}\right)-\Id/n}{\|(\mathrm{diag}\left(\rho^{(K-k)}\right)-\Id/n\|}
    \end{split}\end{align}
    In either case, this fulfills
    \begin{align}
        \inf_{c\in\RR}\norm{(P)^{(K-k)} - c\Id} \leq 1.
    \end{align}
    Then, from a geometric series argument follows
    \begin{align}\begin{split}
        &\inf_{c\in\RR}\norm{(\Delta H)^{(K)} - c\Id} \\
        \leq&
        \sum_{k=0}^{K-1} \beta^{k} \inf_{c\in\RR}\norm{(P)^{(K-k)} - c\Id}  \\
        \leq&
        \sum_{k=0}^{K-1} \beta^{k}
        \leq
        \frac{1}{1-\beta}.
    \end{split}\end{align}
\end{proof}

Next, we show that for a suitable step length the change in free energy (and thus also the decrease in relative entropy distance) can be lower-bounded in terms of $\tr(\rho_H\Delta H)$:
\begin{lemma} \label{lem:a_priori_entropy}
	Let $0\leq\beta<1$ be the momentum hyperparameter, $\Delta H\in\RR^{n\times n}$ an update matrix in the HU routine, $\rho_H$ a Gibbs state for a Hamiltonian $H\in\RR^{n\times n}$, $F(H) = -\ln(\tr(\exp(-H)))$
    the free energy and
 \begin{align}
        \lambda = \frac{(1-\beta)^2}{2}\tr(\rho_H\Delta H)
 \end{align}
 the step length of an update.
 If $\tr(\rho_H\Delta H)\geq 0$, then
    \begin{align}
        F(H+\lambda\Delta H) - F(H)  \geq \frac{\lambda^2(1-\beta)^2}{4}.    
    \end{align}
\end{lemma}
\begin{proof}
    We denote $J:=1/(1-\beta)$. The first step is to show that
    \begin{align} \label{eq:F_double_partial}
        |\partial^2_\lambda F(H+\lambda\Delta H)|
        =
        |\partial_\lambda\tr(\rho_{H+\lambda\Delta H}\Delta H)| 
        \leq 
        2J^2.
    \end{align}
    First, note that Gibbs states are unchanged when adding multiples of identity to the Hamiltonian, so that
	\begin{align}
		\rho_{H+\lambda\Delta H} = 
		\frac
		{e^{-H -\lambda (\Delta H -c\Id)}}
		{\tr\left( e^{-H -\lambda (\Delta H -c\Id)}\right)},
	\end{align}
 and similarly 
 \begin{align}
 \tr \left((\rho_{H+\lambda\Delta H} - \rho_{H+(\lambda+\varepsilon)\Delta H})c\Id\right)=0,
 \end{align}
	where, as shown in Lem.~\ref{lem:bound_R}, we can choose $c$ such that $\|\Delta H-c\Id\|\leq J$.\\
    Next, Ref.~\cite[Lem.~16]{BrandaoInequality} states for Hermitian operators $H$ and  $H'$
    \begin{align}
        \norm{\frac{e^H}{\tr(e^H)}-\frac{e^{H'}}{\tr(e^{H'})}}
        \leq 2 \left(
            e^{\norm{H-H'}}-1
        \right).
    \end{align}
    Then, combining the above with a matrix Hölder inequality gives 
\begin{align}\begin{split}
    |\partial_\lambda \tr\left(\rho_{H+\lambda\Delta H}\Delta H\right)|
    =&
    \lim_{\varepsilon \to 0} \left|
    \frac{\tr \left(\rho_{H+\lambda\Delta H} \Delta H\right) -  \tr \left(\rho_{H+(\lambda+\varepsilon)\Delta H} \Delta H\right)} {\varepsilon}
    \right| \\
    =&\lim_{\varepsilon \to 0} 
    \frac{ \inf_{c\in\RR}\big|\tr\left(
        (\rho_{H+\lambda\Delta H} - \rho_{H+(\lambda+\varepsilon)\Delta H})
        (\Delta H-c\Id)
    \right)\big|}
    {|\varepsilon|}
     \\
    \leq& \inf_{c\in\RR} \norm{\Delta H-c\Id}\lim_{\varepsilon \to 0}  \frac{
    \| \rho_{H+\lambda\Delta H}-\rho_{H+(\lambda+\varepsilon)\Delta H}\|_{\tr}}
    {|\varepsilon|}\\
    \leq&
    2J \lim_{\varepsilon \to 0} \frac{e^{\inf_{c\in\RR}\|\varepsilon(\Delta H-c\Id)\|}  - 1}
    {|\varepsilon|} \\
    =&
    2J 
    \inf_{c\in\RR}\|(\Delta H-c\Id)\|
    \leq
    2J^2
\end{split}\end{align}
 
	and thus $|\partial_\lambda \tr( \rho_{H+\lambda\Delta H} \Delta H)| \leq 2J^2$. It follows that $\lambda=\tr(\rho_H\Delta H)/(2J^2)$ fulfills $\tr (\rho_{H+\lambda\Delta H} \Delta H)\geq 0$.\\
 
    Next, let $\alpha:=\tr(\rho_H\Delta H)$. Now, for the change in free energy we have
    \begin{align}\begin{split}
        \Delta F
        &=
        F\left(H+\alpha/(2J^2)\Delta H\right) - F(H) \\
        &=
        \int_0^{\alpha/(2J^2)} \partial_{\lambda'} F(H)  d\lambda'
        = \int_0^{\alpha/(2J^2)}
        \tr (\rho_{H+\lambda'\Delta H} \Delta H) d\lambda' \\
        &\geq
        \int_0^{\alpha/(2J^2)}\left( 
        \tr(\rho_H\Delta H) 
        - \lambda' |\partial_{\lambda'} \tr (\rho_{H+\lambda'\Delta H} \Delta H)|
        \right) d\lambda',
    \end{split}\end{align}
    where we bounded the decrease of the integrand by its derivative.
    Then, with \eqref{eq:F_double_partial} follows
    \begin{align}\begin{split}
        \Delta F
        \geq
        \int_0^{\alpha/(2J^2)} 
        (\alpha - 2J^2\lambda')
        d\lambda'  
        =
        \frac{\alpha^2}{4J^2}.
    \end{split}\end{align}
\end{proof}
Now, we have all the tools to bound the maximum number of iterations needed in the HU routine if a feasible solution exists:
\entropyBoundConvergence*
\begin{proof}
     By the definition in Sec.~\ref{sec:momentum}, $\Delta H$ is of the form $P + cM$ with some positive factor $c$ (we can ignore here the rescaling $\tilde P_c = \tr(P_c\rho) P_c$, as this can be absorbed by the $\lambda$).
     Due to the overshoot criterion in each preceding iteration, $\tr(\rho_H M)$ is always nonnegative, and thus, $\tr(\rho_H \Delta H)\geq \tr(\rho_H P)$. With Lem.~\ref{lem:a_priori_entropy} we can lower bound the change in free energy at each step by $\Delta F \geq \frac{\tr(\rho_H \Delta H)^2(1-\beta)^6}{16}$. 
     Then, when using the original $P_d^{\ell_1}$ for diagonal updates, 
     we have $\tr(\rho_H P)\geq\epsilon$ at each update, and thus $\Delta F\geq \frac{\epsilon^2(1-\beta)^6}{16}$ in each iteration. 
     The initial free energy is $F(0)=-\ln(n)$. If a feasible solution exists, the free energy cannot become positive, hence, the maximum number of steps is $T=16 (1-\beta)^{-6} \epsilon^{-2} \ln(n)$. 
\end{proof}

\subsection{Randomized rounding}

\subsubsection{Inexact diagonal constraints}

The following proof of Thm.~\ref{thm:muScaling} is based on the proof of Ref.~\cite[Prop.~3.1]{Brandao2022fasterquantum} with the main difference being Eq.~\eqref{eq:rho22_norm}, which allows us to obtain a scaling of $\O(\epsilon^{1/3})$ instead of $\O(\epsilon^{1/4})$. For completeness we will give the full proof.
\muScaling*
\label{sec:stabilityProof}
\begin{proof}
For convenience we define $\xi:=\epsilon^{1/3}$. Then, by assumption,  
\begin{align}
    \sum_i | \rho_{ii} - 1/n | \leq \xi^3.
\end{align}

The construction of $\rho^\sharp$ consists of two steps. First we adjust the rows and columns of $\rho$ that have large diagonal deviations. For this, we set the diagonal elements in these rows and columns to $1/n$ while setting the corresponding off-diagonal entries to zero, thus ensuring positive semidefiniteness of the resulting matrix. 
This is a rather harsh correction, but $\sum_i|\rho_{ii}-1/n|\leq \xi^3$ guarantees us that only a small number of entries need to be treated this way. In the second step, we set all remaining diagonal entries to $1/n$ and restore positive semidefiniteness by shifting $\rho$ by $\frac{\xi}{n}\Id$ and renormalizing it.

Denote $d_i=\rho_{ii}-1/n$.
Next, we define the index set $B\subset\{1,\dots,n\}$ corresponding to the diagonals with larger deviations:
\begin{align}
    B = 
		\{i:|d_i|>\frac{\xi}{n} \}
		\quad\text{with complement}\quad
		\bar B
		=
		\{i:|d_i|\leq\frac{\xi}{n} \}.
\end{align}
Now, we define the two matrices $\rho', D\in \RR^{n\times n}$ and construct $\rho^\sharp$ from these, given as follows:
\begin{align}
    \text{Step 1:} \qquad \notag     \\
    \rho'_{ij} 
    &=
    \begin{cases} 
        1/n &\mathrm{if} \qquad i=j, \quad i \in B, \\
        0 &\mathrm{if} \qquad  i\neq j, \quad i \in B \vee j \in B, \\
        \rho_{ij} & \mathrm{else}, 
   \end{cases} \label{eq:rho_prime}
    \\
    \text{Step 2:} \qquad \notag     \\
    D 
    &=
    \begin{cases} 
        - d_i &\mathrm{if} \qquad i=j, \quad i \notin B, \\
        0 & \mathrm{else}, 
   \end{cases}
   \\
   \rho^\sharp 
   &=
   \frac{1}{1+\xi}\left(
   \rho' + D + \frac{\xi}{n}\Id
   \right).
\end{align}
Then $\rho'$ is psd, because it arises by first
compressing $\rho$ to the submatrix with indices in $\bar B$,
and then adding $\Id/n$ to submatrix with indices in $B$.
Both steps preserve positive-semidefiniteness.
Next, $D$ is a diagonal matrix with $|D_{ii}|\leq \frac{\xi}{n}$, and therefore, $D+\frac{\xi}{n}\Id$ is psd. 
Thus, $\rho^\sharp$ is psd with diagonal entries $1/n$. 
We now show that these corrections are mild in the sense that $\norm{\rho^\sharp-\rho}_\mathrm{tr} = \O(\xi)$.

From $\sum_i|\rho_{ii}-1/n|\leq \xi^3$ we have
\begin{align}
    |B| \leq \xi^2 n.
\end{align}
Next, write $\rho$ in block matrix form
\begin{align}
    \rho = 
    \begin{pmatrix}
        \rho^{(11)} & \rho^{(12)} \\ 
        \rho^{(21)} & \rho^{(22)}
    \end{pmatrix}, 
\end{align}
with coordinates ordered such that the first block corresponds to the indices in $B$ and the second block to the ones in $\bar B$.
Then, following from the construction in \eqref{eq:rho_prime},  
\begin{align}
	\begin{split}
    \norm{\rho'-\rho}_{\tr} 
    &=
    \norm{
    \begin{pmatrix}
        \Id_B/n -  \rho^{(11)} & -\rho^{(12)} \\ 
        -\rho^{(21)} & 0
    \end{pmatrix}
    }_{\tr} \\
    &\leq 
    \|\rho^{(11)}\|_{\tr} + 2\|\rho^{(12)}\|_{\tr} + \|\Id_B/n\|_{\tr}.
		\label{eqn:rho_block}
	\end{split}
\end{align}
Now, using a result from Ref.~\cite{AuKing2003}, we have
\begin{align} \label{eq:rho_analysis_start}
    \norm{
        \begin{pmatrix}
            \|\rho^{(11)}\|_{\tr} & \|\rho^{(12)}\|_{\tr} \\ 
            \|\rho^{(21)}\|_{\tr} & \|\rho^{(22)}\|_{\tr}
        \end{pmatrix}
    }_2
    &\leq
    \norm{
        \begin{pmatrix}
            \|\rho^{(11)}\|_{\tr} & \|\rho^{(12)}\|_{\tr} \\ 
            \|\rho^{(21)}\|_{\tr} & \|\rho^{(22)}\|_{\tr}
        \end{pmatrix}
    }_{\tr} \\
    &\leq 
    \norm{
    \begin{pmatrix}
        \rho^{(11)} & \rho^{(12)} \\ 
        \rho^{(21)} & \rho^{(22)}
    \end{pmatrix}
    }_{\tr}
    =1,
\end{align}
where $\|\cdot\|_2$ is the Frobenius (or Schatten-2) norm. Then,
\begin{align}
    \norm{\rho^{(11)}}_{\tr}^2 + 2\norm{\rho^{(12)}}_{\tr}^2 + \norm{\rho^{(22)}}_{\tr}^2 \leq 1. \label{eq:rho_square_sum}
\end{align}
Because $\rho^{(22)}$ is a principle submatrix of $\rho$, it is also positive semidefinite. 
Thus,
\begin{align} 
	\begin{split}
    \norm{\rho^{(22)}}_{\tr} 
		 &= \tr\left(\rho^{(22)}\right) 
		= 1-\tr\left(\rho^{(11)}\right) 
    \geq 1-\frac{|B|}{n} - \sum_{i\in B} |d_i| \geq 1-\xi^2-\xi^3  \\
		&= 1-\O(\xi^2). \label{eq:rho22_norm}
	\end{split}
\end{align}
Then, combining \eqref{eq:rho_square_sum} and \eqref{eq:rho22_norm} gives 
$\norm{\rho^{(11)}}_{\tr}^2 + 2\norm{\rho^{(12)}}_{\tr}^2=O(\xi^2)$, and hence 
\begin{align} \label{eq:rho_11_12_bound}
    \norm{\rho^{(11)}}_{\tr} + 2\norm{\rho^{(12)}}_{\tr}=\O(\xi).
\end{align}
Furthermore, we have $\|\Id_B/n\|_{\tr} = \frac{|B|}{n}  \leq \xi^2$.
Then, from (\ref{eqn:rho_block}),
\begin{align}
    \norm{\rho'-\rho}_{\tr}=\O(\xi). \label{eq:step 1}
\end{align}
This concludes the first step. For the second step we have
\begin{align}\begin{split}
    \norm{\rho^\sharp - \rho'}_\mathrm{tr}
    &=
    \norm{
    \left(\frac{1}{1+\xi}-1\right)\rho'
    + \frac{1}{1+\xi}\left(D+\frac{\xi}{n}\Id\right)
    }_\mathrm{tr}  \\
    &=
    \frac{1}{1+\xi} \norm{
        -\xi\rho' + D + \frac{\xi}{n}\Id
    }_\mathrm{tr}  \\
    &\leq \frac{1}{1+\xi} \left( \xi \trnorm{\rho' - \rho} + \xi \trnorm{\rho} + \trnorm{D} + \frac{\xi}{n} \trnorm{\Id}\right) \\
    &= \O(\xi). \label{eq:step 2}
\end{split}\end{align}
Combining \eqref{eq:step 1} and \eqref{eq:step 2} with a triangle inequality gives us
\begin{align}
    \norm{\rho^\sharp - \rho}_\mathrm{tr}=\O(\xi)=\O(\epsilon^{1/3}).
\end{align}
\end{proof}

\subsubsection{Approximation ratio after rounding}
The proof of Lem.~\ref{lem:AN_rounding} closely follows Ref.~\cite[Sec.~4.1]{RoundingGrothendieck}.
Compared to the original version of the proof, we also make a statement about the quality of rounded vectors obtained from solutions that are not strictly feasible.

\begin{lemma} \label{lem:AN_rounding}
    Let $\rho^\star\in\RR^{n\times n}$ be the optimal SDP solution corresponding to a cost matrix $C\in\RR^{n\times n}$ with a block structure as defined in \eqref{eq:C_block}.
    Let $x\in \{-1,1\}^n$ be the vector obtained by applying the randomized rounding procedure to a psd matrix $\rho\in\RR^{n\times n}$ with $\rho_{ii}>0$.
    Let $\sigma\in\RR^{n\times n}$ be the matrix with entries $\sigma_{ij}=\frac{\rho_{ij}}{n\sqrt{\rho_{ii}\rho_{jj}}}$. Then 
    \begin{align}
        \EE \left[ x^T C x \right]
        \geq
        \frac2{\pi} n \tr (C\sigma)  -\left(1 - \frac2{\pi} \right) n \tr (C \rho^\star).
    \end{align}
\end{lemma}
\begin{proof}
    Set $u_i:=\frac{(\sqrt{\rho})_i}{\sqrt{\rho_{ii}}}$, so that
    \begin{align}
        u_i^T u_j
				=
				n\sigma_{ij}.
    \end{align}
    Then we have
    \begin{align}
        x_i:=\sgn((\sqrt{\rho})_i^T g) = \sgn(\sqrt{\rho_{ii}}u_i^T g) = \sgn(u_i^T g).
    \end{align}

    We will bound the expected objective value $\EE[ x^T C x ]$  of the rounded solution. Here, the expected value is taken over a standard Gaussian random vector $g\in\RR^n$. 
This will be aided by two short calculations, valid for any two unit vectors $c, b\in\RR^n$.
First, 
\begin{align*}
	\EE[ b^T g g^T c]
	=
	b^T \EE[gg^T] c 
	=
	b^T c.
\end{align*}
Next, using the fact that the distribution of the Gaussian vector is rotationally invariant,
we may assume that $c=e_1$ is equal to the first element of the standard basis, and $b=b_1 \, e_1 + b_2 \, e_2$ is a linear combination of the first two basis vectors. 
Similarly, $g_1$ and $g_2$ denote the first two components of $g$.
Then
\begin{align*}
	\EE[(b^T g) \sgn( c^T g) ]
	&=
	\EE[ (b_1 g_1 + b_2 g_2) \sgn(g_1) ] \\
	&=
	\EE[ b_1 g_1\sgn(g_1)]  + \EE[b_2 g_2]\underbrace{\EE[ \sgn(g_1)]}_{=0}  \\
	&=
	b_1 \EE[ g_1\sgn(g_1)] \\
	&= 
	b_1\,\frac2{\sqrt{2\pi}} \int_{0}^\infty xe^{-x^2/2} \,\mathrm{d}x
	= \sqrt{\frac 2\pi}\,b^T c.
\end{align*}

From this, we can expand 
\begin{align*}
	\frac{\pi}2 
	\EE 
	\left[ x_i x_j \right]
	=&
    \frac{\pi}2 
	\EE 
	\left[ \sgn(u_i^T g) \sgn(u_j^T g) \right] \\
	=&
	\underbrace{
	\EE \left[
		\left(
			u_i^T g
			-
			\sqrt{\frac{\pi}2}
			\sgn(u_i^T g) 
		\right)
		\left(
			u_j^T g
			-
			\sqrt{\frac{\pi}2}
			\sgn(u_j^T g)
		\right)
	\right] 
	}_{=:\Delta_{ij}}  \\
	 &-
	\EE[u_i^T g \, g^T u_j]
	+
	\sqrt{\frac{\pi}2}
	\EE[u_i^T g \, \sgn(u_j^T g)]
	+
	\sqrt{\frac{\pi}2}
	\EE[u_j^T g \, \sgn(u_i^T g)]
	\\
	=&
	\Delta_{ij} - u_i^Tu_j + u_i^Tu_j + u_i^Tu_j
    =
	\Delta_{ij}
	+n\sigma_{ij}.
\end{align*}
As a convex combination of symmetric rank-one matrices, $\Delta$ is psd.
Factoring out, using the two calculations above, and as well as $u_i^Tu_i=n\sigma_{ii}=1$,
we see that its main diagonal elements satisfy
\begin{align}
\begin{split}
	\Delta_{ii}
	&=
	\EE 
	\left[
		\left(
			u_i^T g
			-
			\sqrt{\frac{\pi}2}
			\sgn(u_i^T g) 
		\right)
		\left(
			u_i^T g
			-
			\sqrt{\frac{\pi}2}
			\sgn(u_i^T g)
		\right)
	\right] \\
	&=
	(1 - 2) u_i^Tu_i
	+
	\frac{\pi}2
	=
	\frac{\pi}2 - 1.
\end{split}
\end{align}
Thus, $(n(\frac\pi2-1))^{-1}\Delta$ is feasible for the SDP.
For a cost matrix with the given block form, the spectrum of objective values for the SDP is symmetric. 
This can be seen by considering the block matrix 
\begin{align}
    U = \begin{pmatrix}
            \Id & 0 \\ 0 & -\Id
        \end{pmatrix}. 
\end{align}
Then, $U\rho^\star U^\dagger$ is also feasible for the SDP and 
\begin{align}
    \tr(C U\rho^\star U^\dagger)=\tr(U^\dagger C U \rho^\star) =-\tr(C \rho^\star).    
\end{align}
Thus, we know that $-\tr(C \rho^\star)$ is a lower bound for the SDP and therefore,
\begin{align}
	|\tr(C\Delta)|
	\leq
	\left(\frac\pi2 -1\right) n \tr(C\rho^\star).
\end{align}
This allows us to bound
\begin{align}
	\EE \left[ x^T C x \right]
	&=
	\frac2{\pi}
	\left(
		\tr (C n\sigma)
		+
		\tr (C \Delta)
	\right) \\
	&\geq
	\frac2{\pi}
	\left(
		\tr (C n\sigma)
		-
		|\tr (C \Delta)|
	\right) \\
	&\geq
    \frac2{\pi} n \tr (C \sigma)  -\left(1 - \frac2{\pi} \right)  n\tr (C \rho^\star).
\end{align}
\end{proof}

The proofs of Lem.~\ref{lem:sigma_norm_bound} and Thm.~\ref{thm:rounding} are based on Ref.~\cite[Sec.~3.5]{Brandao2022fasterquantum}. They are adjusted to also consider the improved scaling from Thm.~\ref{thm:muScaling}.

\begin{lemma} \label{lem:sigma_norm_bound}
    Let $0\leq\xi\leq 1/2$ and $\rho\in\RR^{n\times n}$ be a psd matrix such that $\sum_i|\rho_{ii}-1/n|\leq\xi^3$.
		Let $\sigma\in\RR^{n\times n}$ be the matrix with entries $\sigma_{ij}=\frac{\rho_{ij}}{n\sqrt{\rho_{ii}\rho_{jj}}}$. Then,
    \begin{align}
        \norm{\rho-\sigma}_{\tr} = \O(\xi).
    \end{align}
\end{lemma}
\begin{proof}
    Similar to the proof of Thm.~\ref{thm:muScaling}, we define the index set $B = \{i:|d_i|>\frac{\xi}{n} \}$ with $d_i=\rho_{ii}-1/n$ and write $\sigma$ as block matrix 
    \begin{align}
    \sigma = 
    \begin{pmatrix}
        \sigma^{(11)} & \sigma^{(12)} \\ 
        \sigma^{(21)} & \sigma^{(22)}
    \end{pmatrix}, 
    \end{align}
    where $\sigma^{(22)}$ corresponds to the index set of $\bar B$.

    We now show that $\norm{\rho^{(22)}-\sigma^{(22)}}_{\tr} = O(\xi)$.
		To this end, define the diagonal
		$|\bar B|\times |\bar B|$-matrix by
    \begin{align}
        (D_{\bar B})_{ii} = \frac{1}{\sqrt{n\rho_{ii}}},
				\qquad i \in \bar B.
    \end{align}
		With this we can define a linear map 
		$\mathcal{D}: X \mapsto D_{\bar B} X D_{\bar B}$
		on the space of 
		$|\bar B|\times |\bar B|$-matrices  
		that fulfills $\mathcal{D}(\rho^{(22)})=\sigma^{(22)}$.

		We can then bound
    \begin{align}
        \norm{\rho^{(22)}-\sigma^{(22)}}_{\tr}
        &= 
        \norm{(\Id-\mathcal{D})\rho^{(22)}}_{\tr} \\
        &\leq 
        \norm{\Id-\mathcal{D}}_{\tr\rightarrow\tr}\norm{\rho^{(22)}}_{\tr} \\
        &\leq 
        \norm{\Id-\mathcal{D}}_{\tr\rightarrow\tr},
    \end{align}
		because $\rho^{(22)}$ is a submatrix of $\rho$ and $\norm{\rho}_{\tr}=1$. As both $\Id$ and $\mathcal{D}$ are self-adjoint with respect of the Frobenius inner product $\tr(X^T Y)$, the dualtity of norms implies $\norm{\Id-\mathcal{D}}_{\tr\rightarrow\tr}=\norm{\Id-\mathcal{D}}_{\infty\rightarrow\infty}$.
    
		By the definition of $B$, the diagonal elements of $\rho^{(22)}$ lie in $[(1- \xi)/n, (1+\xi)/n]$. Then
    \begin{align}
        \frac{1}{\sqrt{1+\xi}}
        \leq (D_{\bar B})_{ii} \leq
        \frac{1}{\sqrt{1-\xi}},
    \end{align}
    and thus for $\xi\leq 1/2$:
    \begin{align} \label{eq:D_bound}
        1-\xi
        \leq (D_{\bar B})_{ii} \leq
        1+\xi.
    \end{align}
    Next, set $D_\xi = D_{\bar B}-\Id$. 
    From Eq.~\eqref{eq:D_bound} it then follows that $\norm{D_\xi}\leq\xi$. 
    Now, 
    \begin{align}
        \norm{(\Id - \mathcal{D})(X)}_\infty
        &=
        \norm{X D_\xi + D_\xi X + D_\xi X D_\xi}_\infty \\
        &=
        2\norm{D_\xi}_\infty \norm{X}_\infty
        + \norm{D_\xi}^2_\infty \norm{X}_\infty \\
        &\leq 
        3\xi \norm{X}_\infty,
    \end{align}
    and hence $\norm{\Id - \mathcal{D}}_{\infty\rightarrow\infty}\leq 3\xi$. We therefore have
    \begin{align}
        \norm{\rho^{(22)}-\sigma^{(22)}}_{\tr}
        \leq
        \norm{\Id - \mathcal{D}}_{\tr \rightarrow\tr}
        =
        \norm{\Id - \mathcal{D}}_{\infty \rightarrow\infty}
        =\O(\xi).
    \end{align}
    Next, we bound the remaining blocks of $\sigma$ using a similar analysis as in Eqs.\eqref{eq:rho_analysis_start}-\eqref{eq:rho_11_12_bound} in the proof of Thm.~\ref{thm:muScaling}. We use  that $\sigma$ is psd and its diagonal entries are equal to $1/n$. Then,
    \begin{align}
        \norm{\sigma^{(22)}}_{\tr}
        =
        \tr(\sigma^{(22)})
        = 
        1-\frac{|B|}{n}
        =
        1 - \O(\xi^2),
    \end{align}
    and thus $\sigma^{(11)}+2\sigma^{(12)}=\O(\xi)$. Additionally, from Eq.~\eqref{eq:rho_11_12_bound} we have $\rho^{(11)}+2\rho^{(12)}=\O(\xi)$.
    Combining the above, we have
    \begin{align*}
        \norm{\rho -\sigma}_{\tr}
        \leq&
        \norm{\sigma^{(22)} -\rho^{(22)}}_{\tr}
        + \norm{\rho^{(11)}}_{\tr}+ 2 \norm{\rho^{(12)}}_{\tr}
        + \norm{\sigma^{(11)}}_{\tr}+2 \norm{\sigma^{(12)}}_{\tr} \\
        =&
        \O(\xi)
    \end{align*}
\end{proof}

\rounding*

\begin{proof}\label{sec:rounding_proof}
Set $\xi=\epsilon^{1/3}$ and $\sigma_{ij}=\frac{\rho_{ij}}{n\sqrt{\rho_{ii}\rho_{jj}}}$ for $i,j=1, \dots, n$.
From the assumptions we have $\|C\|=1$.
Then, using Hölder's inequality and applying Lem.~\ref{lem:sigma_norm_bound} gives
\begin{align}
    |\tr(C(\rho-\sigma))|
    \leq
    \norm{\rho-\sigma}_{\tr}
    =
   \O(\xi)
\end{align}
and therefore $|\tr(C(\rho^\star-\sigma))|\leq  \O(\xi) + \epsilon = \O(\epsilon^{1/3})$.
Now, applying Lem.~\ref{lem:AN_rounding} gives
\begin{align}
    \EE \left[ x^T C x \right]
    &\geq
    \frac2{\pi} n \tr (C\sigma)  -\left(1 - \frac2{\pi} \right) n \tr (C \rho^\star) \\
    &\geq
    \frac2{\pi} n \left(\tr (C\rho^\star) -  \O(\epsilon^{1/3})\right)  -\left(1 - \frac2{\pi} \right) n \tr (C \rho^\star) \\
    &=
    \left(\frac4{\pi} - 1 \right) n \tr (C \rho^\star) - \O(\epsilon^{1/3} n ).
\end{align}
\end{proof}

\section{Number of quantum gates for Hamiltonian Updates} \label{sec:quantumCircuits}
In Sec.~\ref{sec:quant_imp}, we gave a procedure for approximating the diagonal entries of a Gibbs state for given Hamiltonian, and its asymptotic scaling. 
Now, we will give lower bounds on the number of two-qubit gates needed for an iteration of HU with selected quantum algorithms \cite{vanApeldoorn2020quantumsdpsolvers,Low2019Qubitization,LowBlockEncoding}. These consists of block encoding, Hamiltonian simulation and Gibbs state preparation. We assume a gate model, where two-qubit gates are realized by CNOT gates. 

\subsection{Block encoding} \label{sec:block_encoding}

This section is based on Ref.~\cite{LowBlockEncoding}, and reproduces some of the their construction in order to make the present document self-contained.

The input is a Hamiltonian $H\in\RR^{n\times n}$.
We assume that it is $s$-sparse in the sense that exactly $s$ elements in each column of $H$ are explicitly specified, with the rest being equal to $0$.
(We refer to the specified elements as the \emph{non-zero} ones, though ``potentially non-zero'' would be more accurate).

The matrix is assumed to be stored in QRAM, 
accessible via two unitary quantum oracles:
\begin{itemize}
    \item $O_F$ takes two indices $i\in [n]$, $l\in [s]$ and maps them to the index of the $l^\mathrm{th}$ non-zero entry in the $i^\mathrm{th}$ row of $H$, denoted by $f(i,l)$:
    \begin{align}
        O_F|i\rangle |l\rangle = |i\rangle |f(i,l) \rangle.
    \end{align}
    
    \item $O_H$ takes two indices $i,j\in [n]$ and returns the matrix entry $H_{ij}$ represented as a $b$-bit number:
    \begin{align}
        O_H|i\rangle |j\rangle |z\rangle = |i\rangle |j\rangle |z\oplus H_{ij} \rangle.
    \end{align}
\end{itemize}

\begin{definition}[{\cite[Def.~6]{LowBlockEncoding}}]
	A unitary $U$ \emph{block-encodes} a Hamiltonian $H$ if
    \begin{align}
        U = 
        \begin{pmatrix}
            H/\alpha & \cdot \\ \cdot & \cdot
        \end{pmatrix} , \quad
        (\langle 0|_a \otimes \Id_d) U (|0\rangle_a \otimes \Id_d) = \frac H \alpha,
    \end{align}
    where $d$ is a $b$-qubit register
		and $\alpha$ a suitable scaling factor.
\end{definition}
In the following encoding $\alpha$ is bounded from below by $s\norm{H}_\mathrm{max}$.

Reference \cite[Thm.~10]{LowBlockEncoding} realize a block encoding $U$ with the following properties:
\begin{itemize}
	\item{ }
		\cite[Lem.~13]{LowBlockEncoding}
		The encoding $U$ is decomposed as a product $U=U_\mathrm{row}^\dagger U_\mathrm{col}$.
	\item{ }
		\cite[Lem.~13]{LowBlockEncoding}
		Both $U_{\mathrm{row}}$ and $U_{\mathrm{col}}$ consist of one invocation each of $O_F, O_H, O_F^{-1}$, as well as a further unitary. Additionally, $U_{\mathrm{col}}$ uses two controlled SWAPs on $\log_2(n)$ qubits each.
	\item{ }
		\cite[Lem.~15]{LowBlockEncoding}
		This additional unitary is a product of 
		single-qubit gates and an $b$-bit comparer. 
\end{itemize}

To lower-bound the number of two-qubit gates used in this construction, we use the results of Ref.~\cite{cuccaro2004Adder}.
They show that a $b$-bit comparer can be implemented using
$(2b-1)$ Toffoli gates and $4b-3$ CNOTs.
In turn, each Toffoli may be realized using $6$ CNOTs and a number of single-qubit gates \cite{ToffoliCost}.
For further steps, we need a controlled implementation of $U$. This can be achieved by just controlling the SWAP gates, as without these, the quantum circuit just simplifies to the identity (c.f.\ Fig.~\ref{fig:c-U}). 
A controlled SWAP on $\log_2(n)$ qubits can be realized with $\log_2(n)$ Tofolli gates and $2\log_2(n)$ CNOTs.

\begin{estimate}
	The controlled implementation of the block encoding $U$ given in Ref.~\cite{LowBlockEncoding} requires at least $16b+16\log_2(n)-9$ two-qubit gates.
\end{estimate}

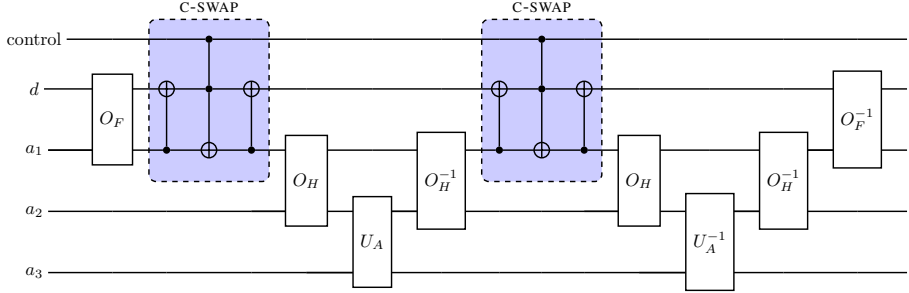
\begin{figure}[H]
    \centering
    \begin{adjustbox}{width=1\textwidth}
    \begin{quantikz}
        \mathrm{control}\; &\qw 
        & \qw \gategroup[3,steps=3,style={dashed, rounded corners,fill=blue!20, inner xsep=2pt}, background]{{\sc c-swap}} &\ctrl{1} &\qw 
        &\qw &\qw &\qw 
        & \qw \gategroup[3,steps=3,style={dashed, rounded corners,fill=blue!20, inner xsep=2pt}, background]{{\sc c-swap}} & \ctrl{1} & \qw
        &\qw &\qw &\qw  
        & \qw  &\qw 
        \\
        d\;& \gate[2]{O_F} 
        &\targ{} & \ctrl{1} & \targ{}
        &\qw &\qw &\qw 
        &\targ{} & \ctrl{1} & \targ{}
        &\qw &\qw &\qw 
        & \gate[2]{O_F^{-1}} &\qw 
        \\
        a_1\;&\qw 
        &\ctrl{-1} &\targ{1} &\ctrl{-1}
        &\gate[2]{O_H} &\qw & \gate[2]{O_H^{-1}} 
        &\ctrl{-1} &\targ{1} &\ctrl{-1}
        &\gate[2]{O_H} &\qw &\gate[2]{O_H^{-1}} 
        &\qw &\qw 
        \\
        a_2\;&\qw &\qw &\qw &\qw &\qw &\gate[2]{U_A} &\qw &\qw &\qw &\qw &\qw &\gate[2]{U_A^{-1}} &\qw  &\qw &\qw
        \\
        a_3\;&\qw &\qw &\qw &\qw &\qw &\qw &\qw &\qw &\qw &\qw &\qw &\qw &\qw &\qw &\qw
    \end{quantikz}
    \end{adjustbox}
    \caption{Circuit for a controlled $U$, resulting by combining $U_\mathrm{col}$ and $U^\dagger_\mathrm{row}$. The SWAP gates are part of $U_\mathrm{col}$, controlling them allows to control the whole circuit of $U$. The registers $d$ and $a_1$ encode the index of row and column condisting of $\log_2(n)$ qubits each, $a_2$ encodes the corresponding entry of $H$ with $b$ qubits, and $a_3$ is an $O(1)$ qubit ancillary register.}
    \label{fig:c-U}
\end{figure}

\subsection{Hamiltonian simulation via Qubitization}
The Hamiltonian simulation uses a procedure called \emph{qubitization}, which is explained in Ref.~\cite{Low2019Qubitization}. It has the following properties:
\begin{itemize}
	\item \cite[Thm~.1, Lem.~6]{Low2019Qubitization} 
		The procedure simulates an encoded Hamiltonian $\frac{H}{s\norm{H}_\mathrm{max}}$ for some time $t$.
    \item \cite[Sec.~5.2, Thm.~4]{Low2019Qubitization} The circuit consists of $Q$ unitaries $V_i$, where $Q$ is the smallest integer $Q=q-1$  s.t.
    \begin{align}
        \epsilon
        \geq
        \frac{4t^q} {2^q q!}. \label{eq:Q_inequality}
    \end{align}
    \item \cite[Lem.~10]{Low2019Qubitization}  Each $V_i$ consists of both the controlled unitaries $U$ and $U^{-1}$ defined in Sec.~\ref{sec:block_encoding}, besides further single- and two-qubit gates. 
\end{itemize}
Equation~\eqref{eq:Q_inequality} can be written as
\begin{align} \label{eq:q_condition}
    q + \log_2 (q!) + \log_2(\epsilon) \geq 2 + q\log_2(t).
\end{align}
To simplify this expression, 
we use an upper bound on Stirling's approximation 
\begin{align}  \label{eq:stirling_upper_bound}
    \log_2(q!) \leq q\log_2(q) - q\log_2(e) + \frac12 \log_2(q) +\log_2(\sqrt{2\pi}) + \frac{1}{12}\log_2(e).
\end{align}
where $e$ is Euler's number. Inserting \eqref{eq:stirling_upper_bound} into \eqref{eq:q_condition} gives
\begin{align}
    q\log_2(q) + q(1-\log_2(e)) + \log_2(\epsilon) + 1.5 
    \geq
    q + \log_2 (q!) + \log_2(\epsilon)
    \geq
    2 + q\log_2(t),
\end{align}
and thus,
\begin{align}
    q (\log_2(q) - \log_2(t) + 1  - \log_2(e))
    \geq 
    0.5 - \log_2(\epsilon). \label{eq:Q_rewritten}
\end{align}
As the right-hand side is positive, we require at least,
\begin{align}
    \log_2(q) - \log_2(t) + 1  - \log_2(e) \geq 0,   
\end{align}
and thus, we can give a lower bound on a $q$ fulfilling Eq.~\eqref{eq:Q_inequality}:
\begin{align}
    q \geq 0.73 t.
\end{align}   

As mentioned above, this circuit only simulates the normalized Hamiltonian $\frac{H}{s\norm{H}_\mathrm{max}}$ for time $t$. We simulate $H$ for some time $\tau$ by setting $t=s\norm{H}_\mathrm{max}\tau$.

\begin{estimate} \label{obs:sim_cost}
    Simulating $H$ for some time $\tau$ using the above construction requires at least $Q$ implementations of $U$ and $U^{-1}$ each. The required number of two-qubit gates is at least 
    \begin{align}\label{eq:H_sim_cost}\begin{split} 
        C_H(\tau, \epsilon) 
        &=(32b+32\log_2(n)-18)\,Q
        \\
        &\geq (32b+32\log_2(n)-18)\,(0.73s\norm{H}_\mathrm{max}\tau -1).
    \end{split}\end{align}
\end{estimate}

\subsection{Gibbs state sampling}
We mainly focus on the estimation of the diagonal of $\rho$ as this is $O(n)$ times more expensive than computing a trace product $\tr(A \rho)$. For this part of the gate count we only consider the gates that result from the Hamiltonian simulation subroutine described above and neglect the additional overhead. The following is based on Ref.~\cite{vanApeldoorn2020quantumsdpsolvers}:

As described in Sec.~\ref{sec:quant_imp}, we need to compute an estimate $\tilde\lambda_\mathrm{min}$ of the smallest eigenvalue of $H$ up to additive error $1/2$. This needs to be done only once per $H$ and thus does not influence the complexity. Using this, we implement a modified version of $O_H$, denoted as $O_{H_+}$ with $H_+ = H-(\tilde\lambda_\mathrm{min} -3/2) \Id$, 
and denote by $C_{H_+}(\tau, \epsilon)$ the corresponding lower bound on the number of two-qubit gates for simulating $H_+$ with the given construction analogous to Eq.~\eqref{eq:H_sim_cost}.

\begin{definition}[{\cite[Def.~35]{vanApeldoorn2020quantumsdpsolvers}}]
    We call $\tilde{W}$ a controlled $(M,\gamma,\epsilon)$-simulation of $H$ if 
    \begin{align}
        \norm{\tilde{W}-W} \leq \epsilon,
    \end{align}
    where 
    \begin{align}
        W := \sum_{m=-M}^{M-1} |m \rangle\langle m| \otimes e^{im\gamma H}
    \end{align}
    is a controlled simulation of $H$.
\end{definition}

\begin{itemize}
	\item{ }\cite[Lem.~36]{vanApeldoorn2020quantumsdpsolvers}
		The implementation of
		a controlled $(M,\gamma,\epsilon)$-simulation of $H_+$ 
		using the construction of
		Ref.~\cite{vanApeldoorn2020quantumsdpsolvers} 
		requires at least
    \begin{align}
        C_{\mathrm{cont}-(M,\gamma,\epsilon)}
        &=
        \sum_{j=0}^{\log_2(M)} C_{H_+}(2^j\gamma, 2^{j-\log_2(M)-1}\epsilon)
				\geq
				C_{H_+}(M\gamma, \epsilon/2) \label{eq:sim_cost_bound}
    \end{align}
    two-qubit gates.
\end{itemize}

\begin{lemma}[{\cite[Thm.~43]{vanApeldoorn2020quantumsdpsolvers}}] \label{lem:sub_gibbs}
    Computing the sub normalized Gibbs state $e^{-H_+}$ to $\epsilon$-precision using the constructions in Ref.~\cite{vanApeldoorn2020quantumsdpsolvers,Low2019Qubitization,LowBlockEncoding} requires at least
    \begin{align}
        C_{(e^{-H_+})} \geq C_{H_+}(2\pi \ln(7.8\epsilon^{-1}), \epsilon/4)
    \end{align}
    two-qubit gates.
\end{lemma}
\begin{proof}
The constants for the non-asymptotic scaling in Lem.~\ref{lem:sub_gibbs} are retrieved from the proofs of Ref.~\cite[Lem.~37, Thm.~40, 43]{vanApeldoorn2020quantumsdpsolvers}:
\begin{itemize}
    \item Ref.~\cite[Thm.~43]{vanApeldoorn2020quantumsdpsolvers} 
    requires $\norm{H_+} \leq 2r$, $\delta=1/2$.
    \item Ref.~\cite[Thm.~43]{vanApeldoorn2020quantumsdpsolvers} 
    defines $f(x+x_0)=\sum_{l=0}^\infty a_lx^l=e^{x-x_0}$, $x_0=r+\delta$. \\
    Thus, $a_l=e^{-x_0}\frac{(-1)^l}{l!}$. 
    \item Ref.~\cite[Thm.~40]{vanApeldoorn2020quantumsdpsolvers}
    defines $g\left(\frac{x}{r+\delta}\right)=f(x+x_0)$ and $g(y)=\sum_{l=0}^\infty b_l y^l.$ \\
    Thus, $b_l=a_l(r+\delta)^l$.
    \item Ref.~\cite[Thm.~40]{vanApeldoorn2020quantumsdpsolvers}
    sets $L=\lceil \frac{1}{\delta'}\ln(8/\epsilon)\rceil$ with $\delta'=\delta/(r+\delta)$ and defines the $L$-truncated polynomial approximation of $g$ and it's corresponding $M$-truncated Fourier approximation given by \cite[Lem.~37]{vanApeldoorn2020quantumsdpsolvers}.
    \item Ref.~\cite[Lem.~37]{vanApeldoorn2020quantumsdpsolvers}
    requires 
    \begin{align} \label{eq:M}
        M=\mathrm{max}\left(2\lceil \ln\left(\frac{4\|b\|_1}{\epsilon'}\right) \frac 1{\delta'}\rceil, 0 \right),     
    \end{align}
    where 
    \begin{align}
        \| b\|_1=\sum_{l=0}^L |b_l| = e^{-x_0} \sum_{l=0}^L\frac{x_0^l}{l!}. \label{eq:b_norm}
    \end{align}
    \item Ref.~\cite[Thm.~40]{vanApeldoorn2020quantumsdpsolvers}
    requires a controlled $(M, \gamma= \frac{\pi}{2r+2\delta}, \frac\epsilon 2)$ simulation.
\end{itemize}
We can directly form an upper bound for $\delta'$:
\begin{align} \label{eq:delta_prime}
    \delta' 
    =
    \frac{\frac12}{r+\frac12}
    \leq
    \frac{1}{\norm{H_+}+1}.
\end{align}
We now want to find a lower bound for $\|b\|_1$. We start with rewriting Eq.~\eqref{eq:b_norm}:
\begin{align}\begin{split}
    \|b\|_1 \
    &=
    e^{-x_0} \left(
    \sum_{l=0}^\infty\frac{x_0^l}{l!}
    - \sum_{l=L+1}^\infty\frac{x_0^l}{l!}
    \right) \\
    &=
    1 - e^{-x_0} \sum_{l=L+1}^\infty\frac{x_0^l}{l!} \\
    &=
    1 - e^{-x_0} \sum_{t=0}^\infty\frac{x_0^{L+1+t}}{(L+1+t)!}. \label{eq:remainder_term}
\end{split}\end{align}
Next, we use
\begin{align}
    L+1
    \geq 
    \lceil \frac{1}{\delta'}\ln(8/\epsilon)\rceil
    \geq
    (2r+1)\ln(8/\epsilon) = 2x_0 \ln(8/\epsilon),
\end{align}
and therefore, $x_0<(L+2)/2$. This allows us to compare Eq.~\eqref{eq:remainder_term} to a geometric series:
\begin{align}\begin{split}
    \|b\|_1 \
    &\geq
    1 - e^{-x_0} \frac{x_0^{L+1}}{(L+1)!} \sum_{t=0}^\infty \left(\frac{x_0}{L+2}\right)^t \\
    &\geq
    1 - e^{-x_0} \frac{x_0^{L+1}}{(L+1)!} \sum_{t=0}^\infty \left(\frac{1}{2}\right)^t \\
    &=
    1 - 2 e^{-x_0} \frac{x_0^{L+1}}{(L+1)!}.
\end{split}\end{align}
Next, we use a lower bound on Stirling's approximation
\begin{align}
    z! \geq e^{z\ln(z) - z},
\end{align}
giving us
\begin{align}\begin{split}
    \|b\|_1 
    &\geq 
    1 - 2e^{-x_0 + \ln(x_0)(L+1) + (L+1) - \ln(L+1)(L+1)}, \\
    &= 
    1 - 2e^{- x_0 - (L+1)\left(\ln\left(\frac{L+1}{x_0}\right)-1\right)}.
\end{split}\end{align}
Then, with $x_0\geq 1/2$ and $\epsilon\leq 1$, 
\begin{align}
    \ln\left(\frac{L+1}{x_0}\right)-1 
    \geq
    \ln(2\ln(8))-1,
\end{align}
and thus,
\begin{align}
    \|b\|_1
    \geq 
    1 - 2 e^{-1/2 - 2\ln(8)(2\ln(8)-1)}
    \geq
    0.98.
\end{align}
 Now, that we know a lower bound on $\|b\|_1$, we can finally bound $M$ using \eqref{eq:M} and \eqref{eq:delta_prime}:
\begin{align}
    M 
    \geq 
    2\ln\left( \frac{4\|b\|_1}{\epsilon'}\right) \frac{1}{\delta'}
    \geq
    2 (\norm{H_+}+1)\ln(7.8\epsilon^{-1}).
\end{align}
With $\gamma=\frac{\pi}{2r+2\delta}=\frac{\pi}{\norm{H_+}+1}$, we thus have
\begin{align}
    M\gamma 
    \geq
    2\pi \ln(7.8\epsilon^{-1}).
\end{align}
Then combining Eq.~\eqref{eq:sim_cost_bound} and Lem.~\ref{lem:sub_gibbs} gives
\begin{align}
     C_{(e^{-H_+})} 
     \geq 
     C_{\mathrm{cont}-(M,\gamma,\epsilon/2)} 
     \geq
     C_{H_+}(2\pi \ln(7.8\epsilon^{-1}), \epsilon/4).
\end{align}
\end{proof}
The procedure above gives a sub normalized Gibbs state $e^{-H_+}$. To obtain a normalized one \emph{amplitude amplification} is applied to increase the amplitude by a factor $1/\tr(e^{-H_+})$. The expected number of amplitude amplification iterations using the construction in Ref.~\cite[Sec.~3]{BoyerSearching} is $0.69\sqrt{n/z}$, where $z$ is a lower bound of $\tr(e^{-H_+})$. By our preparation of $H_+$ we ensured that its minimum eigenvalue is $\lambda_\mathrm{min}\leq 2$ and thus $z=e^{-2}$ suffices. Thus, the expected cost for preparing the Gibbs state is
\begin{align} \label{eq:Gibbs_cost}
    C_{\rho} 
    =
    0.69e\, n^{1/2} C_{(e^{-H_+})}
    \geq 
    1.87 n^{1/2} C_{H_+}(2\pi \ln(7.8\epsilon^{-1}), \epsilon/4).
\end{align}
Finally, combining Eq.~\eqref{eq:Gibbs_cost} and Obs.~\ref{obs:sim_cost} gives the following conservative estimate:
\gateCost*

\subsection{Diagonal entry sampling}
\label{sec:DiagonalSamplingProof}
We now derive the result of Estimate~\ref{est:number_samples} that states the required number of Gibbs state preparations on a quantum computer needed to estimate the diagonal entries. The diagonal entries are estimated statistically by measuring the Gibbs states in the computational basis.
For technical reasons, we randomize the total number of samples according to a Poisson distribution Pois$(m)$, where $m$ is the expected number of samples needed. Then, our estimate for each diagonal entry is given by $\hat\rho_{ii}=N_i/m$, where $N_i$ represents the number of measurements resulting in the $i^\mathrm{th}$ state. Note that we are dividing by the expected number of measurements, not the actual one. 

\begin{lemma} \label{lem:sampling1}
    Let $\eta\leq 1/8$. 
		Given $N$ preparations of $\tilde\rho\in\RR^{n\times n}$, where $N$ is drawn from a Poisson distribution with mean $m=2.13\eta^{-2}(n\ln 2 + \ln(1/p))$, we can get estimates $\{\hat\rho_{ii}\}_{i\in\{1,\dots,n\}}$ s.t.\ $\sum_i|\tilde\rho_{ii}-\hat\rho_{ii}|\leq\eta$ with probability $1-p$.
\end{lemma}

\begin{proof}
	Let $N_i$ be the number of times the outcome $i$ was obtained, so that the total $N$ satisfies $N=\sum_i N_i$.
	The randomization of $N$ makes the $N_i$ independent:
	\begin{align}\begin{split}
			\Pr[N_i=x_i\, \forall i\in[n]]
			&=
			\frac{m^{\sum_i x_i} e^{-m}}{(\sum_i x_i)!} 
			(\sum_i x_i)!
			\prod_i \frac{p_i^{x_i}}{x_i!}  \\
			&=
			\prod_i \frac{(mp_i)^{x_i} e^{-mp_i}}{x_i!}  \\
			&=
			\prod_i \Pr[\Pois(mp_i) = x_i].
	\end{split}\end{align}
	The moment generating function for $N_i$ is given by
    \begin{align}
				M_{N_i}(\lambda)
				=
        \EE \left[ e^{\lambda N_i} \right]
        =
        \exp(m p_i (e^{\lambda} -1)). 
    \end{align}
    Define $\delta_i=\hat p_{i}-p_i$. 
		Then
    \begin{align}
			\begin{split}\label{eqn:delta_mgf}
				M_{\delta_i}(\lambda)
				&=
        \EE \left[ e^{\lambda\delta_i} \right] 
				=
				\EE \left[ e^{\lambda(N_i/m - p_i)} \right] \\
				&=
				e^{-\lambda p_i}
				\EE \left[ e^{\lambda/m N_i} \right] \\
				&=
				e^{-\lambda p_i}
				M_{N_i}(\lambda/m) \\
				&= 
        \exp(-\lambda p_i +m p_i (e^{\lambda/m} -1)).
			\end{split}
    \end{align}
		This moment generating function fulfills
		$
		M_{\delta_i}(|\lambda|) 
		\geq 
		M_{\delta_i}(\lambda) 
		$.
		Indeed, for $\lambda > 0$ a series expansion gives
		\begin{align}\begin{split}
			\ln M_{\delta_i}(\lambda) 
			&=
			-\lambda p_i +m p_i (e^{\lambda/m} -1) \\
			&=
			m p_i \sum_{k=2}^\infty \frac{(\lambda/m)^k}{k!} \\
			&\geq 
			m p_i \sum_{k=2}^\infty \frac{(-\lambda/m)^k}{k!}
			=
			\ln M_{\delta_i}(-\lambda) .
		\end{split}\end{align}
    Then, arguing as in the standard proof of the Chernoff bound,
    \begin{align}\begin{split}
        \Pr \left[ \sum_{i=1}^n |\delta_i| \geq \eta \right]
				&\leq e^{-\lambda \eta} \EE \left[ e^{\lambda \sum_{i=1}^n |\delta_i|} \right] \\
        &= e^{-\lambda \eta} \prod_{i=1}^n \EE \left[ e^{\lambda  |\delta_i|} \right]  \\
        &\leq 
        e^{-\lambda \eta} \prod_{i=1}^n \EE \left[ e^{\lambda \delta_i} + e^{-\lambda \delta_i} \right]  \\
        &\leq
        e^{-\lambda \eta} \prod_{i=1}^n 2 M_{\delta_i}(\lambda)  \\
        &= 2^n \exp\left( -\lambda\eta -\lambda \sum_{i=1}^n p_i + m\sum_{i=1}^n p_i(e^{\frac{\lambda}{m}} - 1) \right) \\
		&= 2^n \exp\left( -\lambda (\eta + 1) + m(e^{\frac{\lambda}{m}} - 1) \right). 
    \end{split}\end{align}
		Choosing $\lambda=m\ln(1+\eta)$ gives us
    \begin{align}
        \Pr \left[ \sum_{i=1}^n |\delta_i| \geq \eta \right]
        \leq
        2^n e^{m(\eta-(1+\eta)\ln(1+\eta))}.
    \end{align}
    Then, with $\ln(1+\eta)\geq\eta-\frac{\eta^2}{2}+\frac{29}{96}\eta^3$ for $\eta\leq\frac{1}{8}$,
    \begin{align}\begin{split}
        \Pr \left[ \sum_{i=1}^n |\delta_i| \geq \eta \right]
        &\leq
        2^n e^{m(\eta-(1+\eta)(\eta-\frac{\eta^2}{2}+\frac{29}{96}\eta^3))}\\
        &=
        2^n e^{m(-\frac{\eta^2}{2}+(\frac{1}{2}-\frac{29}{96})\eta^3+\frac{29}{96}\eta^4)} \\
        &\leq
        2^n e^{-\frac{m\eta^2}{2.13}}.
    \end{split}\end{align}
    Demanding $\Pr \left[ \sum_{i=1}^n |\delta_i| \geq \eta \right]\leq p$ and solving for $m$ concludes the proof.
\end{proof}

\begin{lemma}\label{lem:diag_estimation}
    Let $\epsilon\leq 1/4$. 
		Given $N$ approximate preparations $\tilde\rho$ of $\rho$ on a quantum computer with ${\sum_i|\tilde\rho_{ii} - \rho_{ii}| \leq \frac\epsilon 8}$, where $N$ is drawn from a Poisson distribution with mean $m=137\epsilon^{-2}(n\ln(2) + \ln(1/p))$, we can obtain classical estimates $\hat\rho_{ii}$ s.t.\ $\sum_i|\tilde\rho_{ii}-\hat\rho_{ii}|\leq\frac\epsilon 4$ with probability $1-p$.
\end{lemma}
\begin{proof}
    Using Lem.~\ref{lem:sampling1} with $\eta=\epsilon/8$ gives $\sum_i|\hat\rho_{ii} - \tilde\rho_{ii}| \leq \frac\epsilon 8$. Then, by triangular inequality we get
    \begin{align}
        \sum_i|\hat\rho_{ii} - \rho_{ii}| 
        \leq 
        \sum_i|\hat\rho_{ii}- \tilde\rho_{ii}| + \sum_i|\tilde\rho_{ii} - \rho_{ii}| 
        \leq \frac{\epsilon}{8} + \frac{\epsilon}{8} = \frac\epsilon 4.
    \end{align}
\end{proof}

For small $\epsilon$ the required mean number of samples approaches
\begin{align}
    m=128\epsilon^{-2}(n\ln(2) + \ln(1/p)).
\end{align}
For the benchmarking, we assume $\ln(1/p))\geq4$.

\subsection{Quantum HU algorithm}
Here we give the high-level HU routine that uses quantum subroutines for estimates involving the Gibbs state $\rho$. The quantum algorithms are discussed in Sec.~\ref{sec:quant_imp}.
\begin{algorithmCap}
\begin{algorithm}[H]
    \caption{Hamiltonian Updates with a quantum computer}\label{alg:qc_HU}
    \begin{algorithmic}[1]
        \Require Query access to $\epsilon/4$-precise subroutines 
        \begin{align*}
            &\textsc{qc\_gibbs\_trace\_product}: H, A\in\RR^{n\times n} \mapsto \tr(A\rho_H),   \\
            &\textsc{qc\_gibbs\_diagonals}: H\in\RR^{n\times n} \mapsto \{(\rho_H)_{ii}\}_{i=1,\dots,n},
        \end{align*}
        normalized cost matrix $C\in\RR^{n\times n}$, threshold objective value $\gamma$, precision parameter $\epsilon$, initial step lengths $\lambda_c$ and $\lambda_d$, momentum hyperparameter $\beta$  
        \Statex
        \noindent
        \Ensure \quad \begin{tabular}[t]{l|l}  
              \textbf{Condition} & \textbf{Output} \\
            \hline
            (\ref{eqn:exact_program}) is feasible & $H$, such that $\rho_H$ is $\epsilon$-feasible\\
            (\ref{eqn:exact_program}) is not $\epsilon$-feasible & false \\
            else & undefined ($H$, such that $\rho_H$ is $\epsilon$-feasible, or false)
        \end{tabular}	
        \Statex
        \Function{quantum\_hamiltonian\_updates}{$C, \gamma, \epsilon, \lambda_c, \lambda_d, \beta$}
        \State $P_c \gets - C + \gamma \Id$
        \State $H \gets 0_{n\times n}$
        \State {$M \gets 0_{n\times n}$}
        \State $F=-\ln(n)$ 
        \\
        \While{$F\leq 0$}
            \State $\tr(P_c\rho) \gets \textsc{qc\_gibbs\_trace\_product}(H, P_c)$
            \If{$\tr(P_c\rho)>\frac 3 4\epsilon$}
                \State $\Delta H \gets \tr(P_c \rho) P_c  + \frac{\beta}{\lambda_c} M$
                
                \State $H, F, \lambda_c \gets \textsc{quantum\_update}(H,\Delta H, F, \epsilon, \lambda_c)$
                \Comment{Cost update}
                \State $M \gets \lambda_c\Delta H$
                \State \textbf{continue}
            \EndIf
            \\
            \State $\{\rho_{ii}\} \gets \textsc{qc\_gibbs\_diagonals}(H)$ 
            \Comment{Estimate diagonal of $\rho$ on QC}
            \If{$\sum_i|\rho_{ii}-1/n|> \frac 3 4\epsilon$}
                \State $P_d^{\ell_2} \gets \sum_i (\rho_{ii} - 1/n) |i\rangle\langle i| / \max_i|\rho_{ii} - 1/n|$
                \State $\Delta H \gets P_d^{\ell_2} + \frac{\beta}{\lambda_d} M$
                
                \State $H, F, \lambda_d \gets \textsc{quantum\_update}(H,\Delta H, F, \epsilon, \lambda_d)$
                \Comment{Diag.\ update}
                \State $M \gets \lambda_d\Delta H$
                \State \textbf{continue}
            \EndIf
            \\
            \State \Return $H$
        \EndWhile
        \State \Return false
        \EndFunction
    \end{algorithmic}
\end{algorithm}
\end{algorithmCap}

\begin{algorithmCap}
\begin{algorithm}[H]
    \caption{Update function with a quantum computer}\label{alg:qc_HU_update}
    \begin{algorithmic}[1]
        \Require Query access to $\epsilon/4$-precise subroutine 
        \begin{align*}
            &\textsc{qc\_gibbs\_trace\_product}: H, A\in\RR^{n\times n} \mapsto \tr(A\rho_H),
        \end{align*}
        Hamiltonian $H$, update matrix $\Delta H$, free energy bound $F$, precision parameter $\epsilon$, current step length $\lambda_c$ or $\lambda_d$, number of trace estimations for free energy bound $J$
        \Statex
        \Ensure Updated Hamiltonian $H_\mathrm{new}$, updated free energy bound $F$, updated step length $\lambda_c$ or $\lambda_d$
        \Statex
        \Function{quantum\_update}{$H, \Delta H, F, \epsilon, \lambda$}
        \State $H_\mathrm{new} \gets H + \lambda \Delta H$ 
        \\
        \State $\tr(\Delta H \rho_\mathrm{new}) \gets \textsc{qc\_gibbs\_trace\_product}(H_\mathrm{new}, \Delta H)$
        \While{$\tr(\Delta H \rho_\mathrm{new}) < \epsilon/4$}: 
        \Comment{Check for overshoots; finite precision}
        \State $\lambda \gets 0.5 \lambda$
        \State $H_\mathrm{new} \gets H + \lambda \Delta H$ 
        \State $\tr(\Delta H \rho_\mathrm{new}) \gets \textsc{qc\_gibbs\_trace\_product}(H_\mathrm{new}, \Delta H)$
        \EndWhile
        \\
        \For{$j \in [J]$}:
        \Comment{Estimate the change in free energy}
        \State $\tr(\rho_{H+ \frac{j}{J}\lambda\Delta H} \Delta H) 
        \gets \textbf{qc\_gibbs\_trace\_product}(H+ \frac{j}{J}\lambda\Delta H, \Delta H)$
        \State $ F \gets F + \frac 1J \left(\tr(\rho_{H+ \frac{j}{J}\lambda\Delta H} \Delta H) - \frac \epsilon 4\right)$
        \Comment{Account for finite precision}
        \EndFor
        \\
        \State $\lambda \gets 1.3 \lambda$
        \\
        \State \Return $H_\mathrm{new}, F, \lambda$
        \EndFunction
    \end{algorithmic}
\end{algorithm}
\end{algorithmCap}

\end{document}